\documentclass{article}

\usepackage{inputenc}
\usepackage[margin = 1.25in]{geometry}

\usepackage{setspace}
\usepackage{hyperref}
\usepackage{dsfont}
\usepackage{amsmath,amssymb,dsfont,amsthm}

\usepackage{graphicx}
\usepackage{caption}
\usepackage{subcaption}
\usepackage{booktabs}

\usepackage[inline]{enumitem}
\usepackage{pdflscape}

\graphicspath{{../../analysis/}{../../descriptive/}}

\usepackage[backend = bibtex,
            style = authoryear,
            maxnames = 6,
            maxcitenames = 2,
            doi = false,
            eprint = false]{biblatex}
\DeclareCiteCommand{\citeyear}
	{}
	{\bibhyperref{\printdate}}
	{\multicitedelim}
	{}

\makeatletter
\def\@makefnmark{%
  \leavevmode
  \raise.9ex\hbox{\fontsize\sf@size\z@\normalfont\tiny\@thefnmark}}
\makeatother

\newtheorem{assu}{Assumption}
\newtheorem{prop}{Proposition}
\newtheorem{definition}{Definition}

\newcommand{\Z}{\mathcal{Z}}
\newcommand{\MW}{\underline{W}}
\newcommand{\mw}{\underline{w}}
\newcommand{\wkp}{\text{wkp}}
\newcommand{\res}{\text{res}}
\newcommand{\fed}{\text{fed}}
\newcommand{\pre}{\text{Pre}}
\newcommand{\post}{\text{Post}}

\DeclareMathOperator{\E}{\mathbb{E}}

\newcommand{\ZIPMWeventsUnbal}{\textnormal{7,626}}
\newcommand{\StateMWeventsUnbal}{\textnormal{82}}

\newcommand{\CityCountyMWeventsUnbal}{\textnormal{121}}
\newcommand{\ZIPMWeventsBase}{\textnormal{2,782}}

\newcommand{\WkpOnResCoeffBaseTen}{8.63}

\newcommand{\OnlyResGammaBase}{0.0372}

\newcommand{\OnlyResGammaBaseTen}{0.37}
\newcommand{\OnlyResGammaBaseTenSE}{0.15}
\newcommand{\OnlyResGammaBasetStat}{2.57}
\newcommand{\OnlyWkpBetaBase}{0.0449}

\newcommand{\OnlyWkpBetaBaseTen}{0.45}
\newcommand{\OnlyWkpBetaBaseTenSE}{0.16}
\newcommand{\OnlyWkpBetaBasetStat}{2.88}
\newcommand{\BothGammaBase}{-0.0219}

\newcommand{\BothGammaBasetStat}{-1.25}
\newcommand{\BothBetaBase}{0.0685}

\newcommand{\BothBetaBaseTen}{0.69}
\newcommand{\BothBetaBaseTenSE}{0.29}
\newcommand{\BothBetaBasetStat}{2.38}
\newcommand{\BothSumBase}{0.0466}

\newcommand{\BothSumBaseTen}{0.47}

\newcommand{\BothSumBasetStat}{2.95}
\newcommand{\GammaEqBetaBasePval}{0.051}
\newcommand{\BothWkpDynGammaBase}{-0.0231}

\newcommand{\BothWkpDynGammaBasetStat}{-1.28}
\newcommand{\BothWkpDynBetaBase}{0.0695}

\newcommand{\BothWkpDynBetaBasetStat}{2.45}
\newcommand{\BothWkpDynSumBase}{0.0464}

\newcommand{\BothWkpDynSumBasetStat}{2.90}
\newcommand{\GammaEqBetaBaseDynPval}{0.045}
\newcommand{\BetaPretrendDynBasePVal}{0.563}

\newcommand{\BothSumStackTen}{\textnormal{0.463}}

\newcommand{\TildeBetaZero}{0.097}
\newcommand{\TildeBetaZeroSE}{0.030}

\newcommand{\TildeBetaZeroPlusBetaOne}{0.181}
\newcommand{\TildeBetaZeroPlusBetaOneSE}{0.065}
\newcommand{\gammaCf}{-0.0219}
\newcommand{\betaCf}{0.0685}
\newcommand{\epsilonCf}{0.1013}
\newcommand{\totIncidenceFedNine}{\textnormal{0.093}}
\newcommand{\totIncidenceCentsFedNine}{\textnormal{9.3}}
\newcommand{\rhoMedianFedNine}{\textnormal{0.103}}
\newcommand{\cbsaLowIncFedNine}{\textnormal{61}}
\newcommand{\rhoMedCentsIndirFedNine}{\textnormal{15.9}}
\newcommand{\rhoMedCentsDirFedNine}{\textnormal{9.7}}
\newcommand{\zipNoIncFedNine}{\textnormal{1,043}}
\newcommand{\zipIncFedNine}{\textnormal{5,741}}
\newcommand{\rhoMedianCentsFedNine}{\textnormal{10}}
\newcommand{\zipcodesFedNine}{\textnormal{6,784}}
\newcommand{\zipNoIncPctFedNine}{\textnormal{15.4}}
\newcommand{\zipIncPctFedNine}{\textnormal{84.6}}

\newcommand{\zipBoundFedNine}{\textnormal{3,616}}

\newcommand{\totIncidenceCentsChiFourteen}{\textnormal{11.2}}

\newcommand{\rhoMedCentsIndirChiFourteen}{\textnormal{15.8}}
\newcommand{\rhoMedCentsDirChiFourteen}{\textnormal{9.2}}
\newcommand{\zipNoIncChiFourteen}{\textnormal{323}}
\newcommand{\zipIncChiFourteen}{\textnormal{62}}
\newcommand{\stateBindingMW}{\textnormal{30}}

\newcommand{\localBindingMW}{\textnormal{37}}
\newcommand{\corrShWkrMedInc}{\textnormal{-0.26}}
\newcommand{\corrShWkrS}{\textnormal{0.30}}
\newcommand{\BottomDecPrRent}{\textnormal{60}}
\newcommand{\TopDecPrRent}{\textnormal{12}}

\title{ From Workplace to Residence: The Spillover Effects of \\
        Minimum Wage Policies on Local Housing Markets%
        \thanks{This article was previously circulated under the title
        ``Minimum Wage as a Place-Based Policy: Evidence from US Housing Rental 
        Markets.''
        We are grateful to Kenneth Chay, John N.\ Friedman, Oded Galor, 
        Peter Hull, Bobby Pakzad-Hurson, Jonathan Roth, Jesse M.\ Shapiro, 
        Neil Thakral, and Matthew Turner 
        for their guidance and support during the development of this project.
        Additionally, we are thankful to Jesse Bruhn, Kyle Butts, Ben Hyman, 
        Samuel K. Hughes, Rafael La Porta, Lorenzo Lagos, Matthew Pecenco, 
        Juan Pereira, Atsushi Yamagishi, and three anonymous 
        referees for valuable suggestions to previous versions of this article.
        We thank seminar participants at 
        the 11th European Meeting of the Urban Economics Association, 
        the Montevideo Graduate Workshop at Universidad Católica del Uruguay, and
        the Applied Micro Lunch at Brown University for valuable comments.
        We acknowledge support from 
        the James M.\ and Cathleen D.\ Stone Wealth and Income Inequality Project and 
        the Orlando Bravo Center for Economic Research, both at Brown University.
        Finally, we thank Martín Gallardo for excellent research assistance.
        All errors are our own.}}

\author{Gabriele Borg \and Diego Gentile Passaro \and Santiago Hermo%
        \footnote{Corresponding author: Santiago Hermo (email: \url{santiago_hermo@brown.edu}).
        Borg: Amazon;
        Gentile Passaro: Amazon;
        Hermo: Department of Economics, Brown University.}}

\date{October 27, 2023}

\begin{document}


\maketitle

\begin{abstract}
    \noindent
    The recent rise of sub-national minimum wage (MW) policies in the US has 
    resulted in significant dispersion of MW levels within urban areas.
    In this paper, we study the spillover effects of these policies on local
    rental markets through commuting.
    To do so, for each USPS ZIP code we construct a ``workplace'' MW measure 
    based on the location of its resident's jobs, and use it to estimate the 
    effect of MW policies on rents.
    We use a novel identification strategy that exploits the fine timing of 
    differential changes in the workplace MW across ZIP codes that share the 
    same ``residence'' MW, defined as the same location's MW.
    Our baseline results imply that a 10 percent increase in the workplace MW 
    increases rents at residence ZIP codes by $\BothBetaBaseTen$ percent.
    To illustrate the importance of commuting patterns, we use our estimates 
    and a simple model to simulate the impact of federal and city counterfactual 
    MW policies.
    The simulations suggest that landlords pocket approximately 10 cents of each 
    dollar generated by the MW across directly and indirectly affected areas, 
    though the incidence on landlords varies systematically across space.
\end{abstract}



\clearpage


\section{Introduction}\label{sec:intro}


Many US jurisdictions have recently enacted minimum wage policies surpassing the 
federal level of \$7.25, creating considerable variation in minimum wage 
(hereafter MW) levels across and even within metropolitan areas.
These policies are inherently \textit{place-based} in that they are tied to 
a location, and workers may live and work in locations under different 
statutory MW levels, suggesting the presence of spatially heterogeneous policy 
effects.
While most research on the effects of the MW has focused on employment and 
wages irrespective of residence and workplace location
\parencite[e.g.,][]{CardKrueger1994, CegnizEtAl2019},
a full account of the welfare effects of MW policies requires an understanding 
of how they affect different markets and how their effects spill over across 
neighborhoods.
In fact, while the MW appears to lower income inequality through the labor 
market \parencite{Lee1999, AutorEtAl2016},
its overall effect on income for low-wage workers may be smaller if there is 
a significant pass-through from MW changes to prices, including housing
\parencite{Macurdy2015}.

In this paper, we study the effect of MW policies on local rental housing 
markets estimating their effects across neighborhoods within a metropolitan
area.
Consider the introduction of a new MW policy in certain neighborhoods of a
metropolitan area in which low-wage workers are more likely to reside outside 
the jurisdiction that enacted it.
The higher MW will cause an increase in the wage income of those low-wage 
workers that commute to work in the affected neighborhoods, causing a boost in 
housing demand and rental prices in their residence neighborhoods.
This effect, arising from the MW at the workplace, could undermine the 
distributional objective of the policy.
Additionally, the MW hike may affect the jurisdiction that enacted the policy, 
for instance by increasing prices of non-tradable consumption.
This effect, operating through the MW at the residence, will affect the 
demand for housing as well, and consequently rental prices.
Thus, commuting patterns become an essential ingredient to understand the 
heterogeneous effects of local MW policies.

To operationalize this insight we collect granular data on commuting patterns 
and construct, for each USPS ZIP code (hereafter ZIP code) and month,
the \textit{workplace MW}, defined as the log statutory MW where
the average worker of the ZIP code works.
We also define the \textit{residence MW}, which is just the log statutory MW in 
the same ZIP code.
Figure \ref{fig:map_mw_chicago_jul2019} visually represents these MW-based 
measures by illustrating their changes for the Chicago-Naperville-Elgin 
Core-Based Statistical Area (hereafter CBSA) in July 2019,
when the city of Chicago and Cook County increased the MW from \$12 to \$13 and 
from \$11 to \$12, respectively.
Even though the statutory MW only changed in some locations in the CBSA, the 
increase affected the workplace MW of most locations.
We formulate a simple partial-equilibrium model that suggests that these 
measures are sufficient to determine the change in rents in a 
local housing market.

Studying the within-city spillover effects of the MW requires granular data on 
rents, which is why we employ a novel ZIP code-level panel dataset from Zillow.
Our main rent variable is calculated as the median rental price per square foot 
for listings within a specific ZIP code-month for Single Family houses, 
Condominiums, and Cooperative units (SFCC).%
\footnote{Single family houses are standalone housing units, while condominiums 
    and cooperatives are multi-unit buildings with varying ownership structures 
    \parencite{ZillowTypesOfHomes}.}
This variable captures the posted price of newly available units, 
thereby avoiding tenure biases and more accurately reflecting current market 
conditions \parencite{AmbroseEtAl2015}.
We find that low-wage households are more likely to be renters,
tend to reside in these housing types,
and that rents per square foot are surprisingly uniform across the income 
distribution.
These findings suggest that the Zillow data can feasibly capture any MW effects.
Moreover, the data varies monthly, aligning with the frequency of MW 
changes, thus allowing us to construct an estimation strategy that exploits the 
exact timing of hundreds of policy changes staggered across jurisdictions and 
months.

To estimate the spillover effects of MW policies on rents, we develop a novel 
difference-in-differences strategy that exploits our granular and high-frequency 
data to compare the evolution of rents across ZIP codes differentially exposed 
to workplace MW changes, conditional on the residence MW.
To further illustrate the importance of commuting patterns in the propagation 
of MW shocks, we use our simple model and our main estimated elasticities to 
evaluate two MW policies: 
a federal MW increase and
a local MW increase in the city of Chicago.
We estimate the share of each dollar of extra income (generated by the MW) that 
accrues to landlords both combining all affected areas and in each particular 
location.
We then discuss our results' implications for assessing the distributional 
impact of MW policies.


We start by introducing a motivating partial equilibrium model of a ZIP code's 
rental market, which is part of a larger geography.
The model is populated by workers who demand housing, and the interaction with 
a supply of rental units by absentee landlords determines the equilibrium 
rental price.
Importantly, residents of the ZIP code can commute to work in other ZIP 
codes, possibly under a different MW policy.
Workers' demand for square footage of homogeneous housing space is modelled as 
a function of prices of non-tradable consumption and income, both of which are 
influenced by the MW levels at residence and workplace locations.
The model illustrates that the impact of a change in MW legislation would vary 
across ZIP codes, depending on whether it alters the MW at the workplace,
the residence, or both.
The model implies that the impact of MW changes in certain ZIP codes on rents 
can be summarized by the workplace and residence MW measures,
emphasizing the need to account for the residence MW in the empirical analysis.


Guided by the theoretical model, we pose an empirical model where log rents in 
a location depend linearly on
leads and lags of the workplace MW,
the residence MW,
ZIP code and time period fixed effects, and 
time-varying controls.
This compares ZIP codes that are differentially exposed to the workplace MW 
but equally affected by the residence MW, conditional on other factors that 
affect the evolution of rents.
The identification assumption is that, within a ZIP code, 
changes in the workplace MW are strictly exogenous with respect to 
unobserved changes in rents after partialing out the confounding variation 
generated by the residence MW.
Given that MW policies are not typically enacted considering their spillover
effects on local rental markets, we argue that this assumption is plausible.
In an appendix, we discuss a general potential outcomes framework following
\textcite{CallawayEtAl2021}.
We demonstrate that, under the assumptions of \textit{parallel trends} and 
\textit{no selection on gains}, 
the effects of the workplace MW and residence MW are identified from the 
conditional slope of log rents with respect to each MW measure.


Our preferred specification implies that 
a 10 percent rise in the workplace MW (holding constant the residence MW) 
increases rents by $\BothBetaBaseTen$ percent (SE=$\BothBetaBaseTenSE$).
Failing to control for the residence MW results in an estimated effect of 
$\OnlyWkpBetaBaseTen$ (SE=$\OnlyWkpBetaBaseTenSE$), and a model that uses 
the residence MW only results in an even lower estimate of 
$\OnlyResGammaBaseTen$ (SE=$\OnlyResGammaBaseTenSE$).
The reasons these estimates are lower are two-fold.
First, by accounting for the difference between workplace and residence 
locations, the workplace MW is a better measurement of the change in the MW 
that is relevant for a ZIP code's wage income, thus reducing measurement 
error.
Second, controlling for the residence MW removes confounding variation in
rents that is generated by unobserved factors that may respond to it, 
such as prices of non-tradable consumption.
Using a rough approximation to the share of MW workers in each ZIP code, we show 
that the elasticity of rents to the workplace MW is larger in locations 
with more MW residents, consistent with the fact that the effect operates by
changing the income of low-wage workers.
Likewise, we find a lower elasticity in locations with larger average incomes.
These results imply that MW changes spill over spatially through commuting, 
affecting local housing markets in places beyond the boundary of the 
jurisdiction that originally enacted the policy.

When including both the workplace and residence MW in our model, we find that 
the coefficient on the residence MW is negative, consistent with the notion 
that increases in local non-tradable consumption ameliorate the effect of the
MW on rents, as in the theoretical model.
However, the coefficient is not statistically significant in our baseline
estimates, and the lack of data on prices of non-tradable consumption of a 
ZIP code's residents prevents us from drawing strong conclusions about 
this effect.


We provide support for our identification assumptions with a battery of 
additional analysis.
First, we test for pre-period coefficients and construct a non-parametric 
analysis of the relationship between log rents and the MW measures.
We find that future MW changes do not predict rents, and the conditional
relationship of log rents with respect to each MW measure is nearly linear,
suggesting that the identification assumptions are plausible.
Second, we estimate our model using a rental index constructed by Zillow that 
controls for variation in the available housing stock at each time.
This variable alleviates concerns that changes in the composition of available 
units, coinciding with MW changes, drive our estimates.
Third, our estimates are robust to using commuting shares for different
years, and they are stronger when we use shares based on jobs below a 
certain nominal income threshold or on younger workers, both of which are 
more likely to be affected by the MW.
This is consistent with the view that identification arises from the ``shares,''
as in \textcite{GoldsmithpinkhamEtAl2020}.
Finally, we construct a ``stacked'' regression model, similar to 
\textcite{CegnizEtAl2019}, that explicitly compares ZIP codes within 
metropolitan areas where some but not all experienced a change in the 
statutory MW.
This helps alleviate concerns that our estimates stem from undesired 
comparisons in difference-in-differences models with staggered treatment 
timing, as highlighted by recent literature 
\parencite{deChaisemartinEtAl2022,RothEtAl2022}.%
\footnote{We also estimate a model that includes the lagged first difference 
	of rents as a control, and is estimated via instrumental variables following 
    \textcite{ArellanoBond1991}.}

Our results remain robust across different sets of controls,
alternative samples of ZIP codes, and reweighing observations to match 
demographics of the population of urban ZIP codes.
We find similar (but noisier) results when we use median rents in 
different housing categories, which are available for smaller samples of 
ZIP codes.


In the final part of the paper, we use our motivating model and simulate 
counterfactual exercises to capture the incidence of MW policies on landlords.
We compute the share pocketed by landlords in each ZIP code, and also
compute the total incidence summing across locations.
We simulate two counterfactual MW policies in January 2020, keeping all other
MW policies in their December 2019 levels.
In the first scenario, we change the federal MW from \$7.25 to \$9.
In the second, we propose a rise in the Chicago City MW from \$13 to \$14.
We estimate that landlords capture $\totIncidenceCentsFedNine$ cents of each 
dollar across locations in affected CBSAs in the former, and 
$\totIncidenceCentsChiFourteen$ cents of each dollar across locations in the 
Chicago-Naperville-Elgin CBSA in the latter.
We find systematic spatial variation in incidence,
with the share pocketed usually being larger in locations that experience an
increase in the workplace MW but not in the residence MW.
These exercises illustrate that commuting patterns are essential to 
understanding the spatial incidence of MW policies within metropolitan areas.


Our analysis has some important limitations.
First, while our workplace MW measure is consistent with a large body of 
empirical work relying on variation generated by shift-share instruments 
(e.g., see recent work by \cite{GoldsmithpinkhamEtAl2020} and
\cite{BorusyakHullJaravel2021}), our formal justification in the theoretical 
model relies on strong constant-elasticity assumptions.
We discuss their plausibility in the body of the paper.
Second, while our model is useful to motivate our empirical strategy, it does 
not account for general equilibrium effects such as changes in commuting 
patterns.
We discuss the potential consequences of relaxing this assumption for our 
empirical results in the context of our model.
Additionally, we caution that our counterfactual simulations based on the model 
should be taken as an approximation to the effects of a small change in the MW.
Third, our exercises do not capture the full welfare effect of MW policies
in the US.
Such an effort would require a general equilibrium model that accounts for 
responses to the MW across several margins.
However, as low-wage households are more likely to rent and thus will be more 
affected by rent effects, our analysis suggest that such a model should 
consider the homeownership status of households.


Our findings expand our understanding of the effects of MW policies on
housing rents by showing how they spill over across local housing markets.
To our knowledge, the only papers whose goal is to estimate the effect of the 
MW on rents in the same location are \textcite{Tidemann2018} and 
\textcite{Yamagishi2021}.
\textcite{AgarwalEtAl2022} studies whether the MW affects eviction 
probabilities, and presents complementary estimates of the effect of the MW on 
rents.
Our paper also relates to \textcite{Hughes2020}, who studies the effect of 
MW policies on rent-to-income ratios and also presents estimates of their 
effect on rents.
None of these papers account for the spatial spillover effects of the MW,
which is essential in a context of within-city variation in MW policies, 
such as the recent experience in the US.

We also contribute to the understanding of place-based policies and the spatial 
transmission of shocks.
\textcite{KlineMoretti2014} argue that place-based policies may result in 
welfare losses due to finite housing supply elasticites.
\textcite{AllenEtAl2020} estimate the within-city transmission of expenditure 
shocks in Barcelona.
We contribute by, first, showing the MW policies result in rent increases
and real income losses for affected workers, and second, by showing that
that local MW policies also transmit across space.

More broadly, our paper relates to the large literature estimating the effects
of MW policies on employment
(see \cite{Dube2019} and \cite{NeumarkShirley2021} for recent reviews of the 
literature), 
the distribution of income \parencite[e.g.,][]{Lee1999, AutorEtAl2016, 
	Dube2019Income}, 
and the overall welfare effect of the MW \parencite{AhlfeldtEtAl2022,
	BergerHerkenhoffMongey2022}.%
\footnote{Our paper is also related to work studying 
	the effects of local MW policies 
	\parencite[e.g.,][]{DubeLindner2021, JardimEtAl2022seattle}, 
	the effect of MW policies on commuting and migration 
	\parencite[e.g.,][]{Cadena2014, Monras2019, PerezPerez2021}, 
	and prices of consumption goods 
	\parencite[e.g.,][]{Aaronson2001, AllegrettoReich2018, Leung2021}.} 
Our contributions are to incorporate spillovers across locations 
\parencite[as in the recent work by][]{JardimEtAl2022discontinuity} and to show 
that rent increases erode some income gains of low-wage workers.
We also contribute by developing a novel panel dataset of MW levels at the 
ZIP code level for the entire US.

Finally, our paper relates to work in econometrics that focuses on spillover 
effects across units,
both in the context of MW policies 
\parencite{Kuehn2016, JardimEtAl2022discontinuity}, 
and more generally of any policy that spills over spatially
\parencite{DelgadoFlorax2015, Butts2021}.
Our approach is similar to \textcite{GiroudMueller2019}: we specify an explicit 
model for spillovers across units that allows us to estimate rich effect 
patterns of the MW on rents.

The rest of the paper is organized as follows.
Section \ref{sec:model} introduces a motivating model of the rental market.
In Section \ref{sec:data} we discuss the empirical relationship between income 
and housing and present our estimation data.
In Section \ref{sec:empirical_strategy} we discuss our empirical strategy and
identification assumptions.
In Section \ref{sec:results} we present our estimation results.
Section \ref{sec:counterfactual} discusses counterfactual MW policies, and
Section \ref{sec:conclusion} concludes.

\section{A Partial-Equilibrium Model}\label{sec:model}
    
In this section we lay out a simple demand and supply model of local rental 
markets.
We use the model to motivate our research design and interpret our empirical
findings.
Specifically, we obtain two results.
First, the model shows that a new MW legislation will have a different effect 
depending on whether it affects the workplace location, residence location, or 
both.
Second, the model shows that under certain conditions the effect of a MW 
policy on rents can be summarized in two MW-based measures: 
the workplace MW and the residence MW.
Derivations rely on several assumptions, the importance of which is 
discussed in the last part of this section.

\subsection{Setup}

We consider the rental market of some ZIP code $i$ embedded in a larger geography 
composed of a finite number of ZIP codes $\Z$.
Workers with residence $i$ work in ZIP codes $z\in\Z(i)$, where 
$\Z(i)\subseteq\Z$.
We let $L_{iz}$ denote the number of $i$'s residents who work in $z$ and 
$L_i = \sum_{z \in \Z(i)} L_{iz}$ the number of residents in $i$.%
\footnote{To simplify, we assume that all of $i$' residents work, so that the 
number of residents equals the number of workers.}
Commuting shares are given by $\pi_{iz} = \frac{L_{iz}}{L_i}$.
We assume that the vector of shares $\{\pi_{iz}\}_{z\in\Z(i)}$ is fixed,
which we think is a good approximation for our empirical setting where we 
observe MW changes at a monthly frequency.%
\footnote{\textcite{AllenEtAl2020} study the within-city transmission of 
expenditure shocks by tourists within Barcelona over a period of two years.
The authors maintain an analogous assumption of constant shares of income that
each location in the city earns from every other location.}
We discuss the consequences of relaxing this assumption later in this section.

\paragraph{Minimum Wages.}

Each ZIP code has a binding nominal minimum wage.
The vector of binding MW levels relevant for $i$ is $\{\MW_z\}_{z\in\Z(i)}$.

\paragraph{Housing demand.}

Each group $(i,z)$ consumes
square feet of living space $H_{iz}$, 
a non-tradable good produced in their residence $C_{iz}^{NT}$, and
a tradable good $C_{iz}^T$.
A representative $(i,z)$ worker chooses between these alternatives by maximizing
a quasi-concave utility function 
$u_{iz} = u \left(H_{iz}, C^{\text{NT}}_{iz}, C^{\text{T}}_{iz}\right)$
subject to a budget constraint
$R_i H_{iz} + P_i(\MW_i) C^{\text{NT}}_{iz} + C^{\text{T}}_{iz} \leq Y_{iz}(\MW_z).$
In this equation 
$R_i$ gives the rental price of housing per square feet,
$P_i(\MW_i)$ gives the price of local consumption,
the price of tradable consumption is normalized to one, and 
$Y_{iz}(\MW_z)$ is an income function.
We specify the effect of MWs on these functions below.

\begin{assu}[Effect of Minimum Wages]\label{assu:mws}
    We assume that
    (i) the price of non-tradable goods is increasing in $i$'s MW, 
    $\frac{d P_i}{d \MW_i} > 0$, and
    (ii) incomes are weakly increasing in $z$'s MW, 
    $\frac{d Y_{iz}}{d \MW_z} \geq 0$, with strict inequality 
    for at least one $z\in\Z(i)$.
\end{assu}

The problem's structure, along with Assumption \ref{assu:mws}, is aligned with 
the prevailing literature.
First, \textcite{MiyauchiEtAl2021} show that individuals tend to consume close 
to home.
Consequently, it's expected that they would be sensitive to local consumption 
prices within their own neighborhoods, justifying the assumption that workers 
consume non-tradables in the same ZIP code.%
\footnote{An extension of the model would allow workers to consume in any ZIP 
    code.
    While theoretically feasible, this extension would require data on 
    consumption trips, which we lack.
    We think of our model as an approximation.}
Second, MW hikes have been shown to increase prices of local consumption 
\parencite[e.g.,][]{Leung2021},
and also to increase wage income even for wages above the MW level 
\parencite[e.g.,][]{CegnizEtAl2019, Dube2019Income}.%
\footnote{An extension would allow separate wage income and business income in 
the budget constraint.
If firm owners tend to live where they work, and MW increases damage profits
\parencite[as found by, e.g.,][]{DracaMachinVanreenen2011, HarasztosiLidner2019},
then business income would depend negatively on the MW level.}

For convenience, we define the per-capita housing demand function as 
$h_{iz} \equiv \frac{H_{iz}}{L_{iz}}$.
The solution to the worker's problem for each $z$ then yields a set of 
continuously differentiable per-capita housing demand functions 
$\{h_{iz} (R_i, P_i, Y_z)\}_{z\in\Z(i)}$.
The following assumption summarizes the properties of these functions.

\begin{assu}[Housing demand]\label{assu:housing_demand}
    Consider the set of functions $\{h_{iz} (R_i, P_i, Y_z)\}_{z\in\Z(i)}$.
    We assume that
    (i) housing is a normal good, 
    $\frac{d h_{iz}}{d Y_z} > 0$ for all $z\in\Z(i)$,
    and
    (ii) housing demand is decreasing in prices of non-tradable consumption, 
    $\frac{d h_{iz}}{d P_i} < 0$.
\end{assu}

Using the first assumption, standard arguments imply that 
$\frac{d h_{iz}}{d R_i} < 0$.
For the second assumption to hold, 
a sufficient (albeit not necessary) condition is that housing and non-tradable
consumption are complements.%
\footnote{To formalize the required condition, let $h_{iz}$ and $c_{iz}$ denote 
    per-capita Marshallian demands resulting from the choice problem, and 
    $\tilde h_{iz}$ denote the corresponding Hicksian housing demand.
    The Slutsky equation implies that
    $$\frac{\partial h_{iz}}{\partial P_i}
    = \frac{\partial \tilde h_{iz}}{\partial P_i}
    - \frac{\partial h_{iz}}{\partial Y_{iz}} c_{iz}.$$
    To obtain $\frac{\partial h_{iz}}{\partial P_i} < 0$, we require that
    $\frac{\partial \tilde h_{iz}}{\partial P_i}
    < \frac{\partial h_{iz}}{\partial Y_{iz}} c_{iz}$, i.e., the income effect 
    of an increase in non-tradable prices is larger than the corresponding 
    substitution effect.}
While direct empirical evidence on this particular channel is lacking,
we view the evidence of workers sorting towards locations with high housing 
costs and more expensive amenities as consistent with it 
\parencite[e.g.,][]{CoutureEtAl2019}.

\paragraph{Housing supply.}

We assume that absentee landlords supply square feet in $i$ according to the 
function $S_i(R_i)$,
and we assume that this function is weakly increasing in $R_i$, 
$\frac{d S_i(R_i)}{d R_i} \ge 0$.
Note that this formulation allows for an upper limit on the number of housing 
units at which point the supply becomes perfectly inelastic.

\subsection{Equilibrium and Comparative Statics}

Total demand of housing in ZIP code $i$ is given by the sum of the demands of 
each group.
Thus, we can write the equilibrium condition in this market as
\begin{equation}\label{eq:equilibrium}
	L_i \sum_{z\in\Z(i)} \pi_{iz} h_{iz} \left(R_i, P_i(\MW_i), Y_z(\MW_z)\right) = S_i(R_i) .
\end{equation}
Given that per-capita housing demand functions are continuous and decreasing 
in rents, under a suitable regularity condition there is a unique equilibrium 
in this market.%
\footnote{To see this, assume that 
$S_i(0)/L_i - \sum_{z\in\Z(i)} \pi_{iz} h_{iz} (0, P_i, Y_z) < 0$
and apply the intermediate value theorem.
Intuitively, we require that at low rental prices per-capita demand exceeds 
per-capita supply.}
Equilibrium rents are a function of the entire set of minimum wages, formally, 
$R^*_i = f(\{\MW_i\}_{i\in\Z(i)})$.

We are interested in two questions.
First, what is the effect of a change in the vector of MWs 
$(\{d \ln \MW_i\}_{i\in\Z(i)})'$ on equilibrium rents?
Second, under what conditions can we reduce the dimensionality of the rents 
function and represent the effects of MW changes on equilibrium rents in a 
simpler way?
We start with the first question.

\begin{prop}[Comparative Statics]\label{prop:comparative_statics}
    Consider residence ZIP code $i$ and a change in MW policy at a larger
    jurisdiction such that for $z\in\Z_0 \subset \Z(i)$ binding MWs increase 
    and for $z\in\Z(i)\setminus \Z_0$ binding MWs do not change,
    where $\Z_0$ is non-empty.
    Under the assumptions of unchanging $\{\pi_{iz}\}_{z\in\Z(i)}$ 
    and Assumptions \ref{assu:mws} and \ref{assu:housing_demand},
    we have that
    \begin{enumerate}
        \item[(a)]
        for any $z'\in\Z_0\setminus\{i\}$ for which $\frac{d Y_{iz'}}{d \MW_{z'}}>0$, 
        the policy has a positive partial effect on rents, 
        $\frac{d\ln R_i}{d\ln\MW_{z'}} > 0$;
        \item[(b)]
        the partial effect of the MW increase in $i$ on rents is ambiguous, 
        $\frac{d\ln R_i}{d\ln\MW_i} \lessgtr 0$; and
        \item[(c)]
        as a result, the overall effect on rents is ambiguous if $i\in\Z_0$ 
        and positive if $i\notin\Z_0$.
    \end{enumerate}
\end{prop}

Proofs are available in Online Appendix \ref{sec:proofs}.

The first part of Proposition \ref{prop:comparative_statics} shows that,
if at least some low-wage worker commutes to a ZIP code $z'$ where the MW 
increased  (so that $\frac{d Y_{iz'}}{d \MW_{z'}}>0$),
then the MW hike will tend to increase rents.
This follows from an increase in housing demand in $i$ due to the increase in
income of workers who commute to $z'$.
The second part of Proposition \ref{prop:comparative_statics} establishes that 
decreasing rents may follow if the minimum wage also increases in ZIP code $i$.
This is because the increase in $i$ lowers housing demand by a substitution
effect, so that the overall effect on rents is ambiguous.
Consequently, the sign of the overall effect of the policy in $i$ is not 
determined a priori.

As apparent from the proof of Proposition \ref{prop:comparative_statics}, 
the effect of the MW on rents at workplaces depends on the elasticities of 
per-capita housing demand to incomes
$\xi^Y_{iz} = \frac{d h_{iz}}{d Y_z} \frac{Y_z}{\sum_z \pi_{iz} h_{iz}}$ and
on the elasticities of income to minimum wages
$\epsilon_{iz}^Y = \frac{d Y_z}{d \MW_z} \frac{\MW_z}{Y_z}$.
These $(i,z)$-specific terms weigh the change in MW levels at workplaces,
and their sum over $z$ impacts rents.
The following proposition establishes conditions under which we can reduce the 
dimensionality of the rent gradient.

\begin{prop}[Representation]\label{prop:representation}
    Assume that for all ZIP codes $z\in\Z(i)$ we have
    (a) homogeneous elasticity of per-capita housing demand to incomes,
    $\xi^Y_{iz}=\xi^Y_{i}$, and
    (b) homogeneous elasticity of income to minimum wages,
    $\epsilon_{iz}^Y=\epsilon_i^Y$.
    Then, we can write
    \begin{equation} \label{eq:theory_representation}
        d r_i = \beta_i  d \mw^{\wkp}_i
              + \gamma_i d \mw^{\res}_i
    \end{equation}
    where 
    $r_{i} = \ln R_i$,
    $\mw_{i}^{\wkp} = \sum_{z\in\Z(i)} \pi_{iz} \ln \MW_z$ 
    is ZIP code $i$'s \textbf{workplace MW}, 
    $\mw^{\res}_i = \ln \MW_i$ 
    is ZIP code $i$'s \textbf{residence MW}, and 
    $\beta_i > 0$ and $\gamma_i < 0$ are parameters.
\end{prop}

Proposition \ref{prop:representation} shows that, under a homogeneity assumption
on the elasticities of per-capita housing demand to income and 
of income to the MW,%
\footnote{The assumptions stated in Proposition \ref{prop:representation} are 
actually stronger than needed.
It is enough to have that the product $\xi^Y_{iz} \epsilon_{iz}^Y$ does not vary 
by $z$.}
the change in rents following a small change in the profile of MWs can be 
expressed as a function of two MW-based measures:
one summarizing the effect of MW changes in workplaces $z\in\Z(i)$,
and another one summarizing the effect of the MW in the same ZIP code $i$.
This motivates our empirical strategy, where we regress log rents on the 
empirical counterparts of these measures.

\subsection{Summary and discussion}
\label{sec:model_summary}

Under the stated assumptions, Proposition \ref{prop:comparative_statics} shows 
that the effect of a MW increase on rents depends on whether it affects the 
ZIP code via the workplace or the residence.
ZIP codes exposed via their workplace only will tend to experience an 
increase in rents.
ZIP codes exposed to both via their workplace and residence will tend to 
experience smaller increases, or even decreases, in rents.
Controlling for the residence MW is thus important to characterize the spatial 
effects of MW policies.
Proposition \ref{prop:representation} provides guidance for the empirical
analysis of spillover effects of the MW by formally justifying the 
workplace MW, our summary measure of the changes in the MW at workplaces.

In this subsection we discuss the plausibility of the assumptions that yield 
these results and the consequences of relaxing them.

\paragraph{The mechanics of rent adjustments.}

The model is static, in the sense that everyone chooses their housing demand 
simultaneously.
Once a MW policy is enacted, a new equilibrium is reached where rents adjust
to the new demand levels to clear the market.
This adjustment takes place because workers can in principle move within ZIP 
codes, and even across ZIP codes, in search of housing that satisfies their 
demand, as long as the commuting shares remain unchanged.
In fact, \textcite{AgarwalEtAl2022} finds an increase in the probability of 
moving after a MW increase, suggesting changes in housing demand.
Online Appendix \ref{sec:dyn_theory_model} discusses an extension of the model
with discrete time periods in which the process of workers moving within the 
ZIP code to renew expiring rental contracts is modelled explicitly.
Therefore, we see that rents adjust because workers can demand more or less 
housing in response to the MW change.%
\footnote{This is a supply and demand story. 
    Another story that would yield increasing rents as response to workplace MW
    hikes is one where landlords and tenants bargain over rents, 
    and the MW improves the outside options of tenants.}

\paragraph{Homogeneity assumptions and the workplace MW.}

How likely are the assumptions that yield Proposition \ref{prop:representation}
to hold?
The assumption that the elasticity of income to the MW is constant will fail if 
the income of some $(i,z)$ groups is more sensitive to the MW than others.
This would be the case if, for example, the share of low-wage workers within 
each $\pi_{iz}$ varies strongly by workplace.
The assumption that the elasticity of housing demand to income is constant 
will hold trivially for preferences
$h_{iz} = g\left(R_i, P_i\right) Y_i$ for some $g\left(\cdot\right)$, such as 
those in Cobb-Douglas or Constant Elasticity of Substitution utility functions.
However, one would expect the elasticity of $(i,z)$ groups with many low-wage 
workers to be larger, suggesting that this type of preferences may not be 
appropriate.

We thus see that the homogeneity assumptions are strong and will likely not hold 
exactly.
However, we expect our empirical model based on Proposition 
\ref{prop:representation} to offer a reasonable approximation to study the 
spillover effects of MW policies on the housing market.
In fact, unless the heterogeneity in $\{\xi_{iz}^Y\epsilon_{iz}^Y\}_{z\in\Z(i)}$ 
has a strongly asymmetric distribution across workplace locations, we expect to 
correctly capture the average contribution of the workplace MW on rents.
In other words, the value of $\beta_i  d \mw^{\wkp}_i$ is likely to be close 
to the value of the elasticity-weighted changes in workplace MW levels that, 
according to the model, determine rents.%
\footnote{More precisely, say that 
    $\xi^Y_{iz}\epsilon_{iz}^Y = \overline{\xi\epsilon}_i + \nu_{iz}$ where 
    $\nu_{iz}$ has a mean of zero.
    In that case, a similar logic than the one in the proof of Proposition 
    \ref{prop:representation} will result in the following expression for 
    rents changes:
    $$
        d r_i = \gamma_i d \mw^{\res}_i
            + \frac{\overline{\xi\epsilon}_i}
                    {\eta_{i} - \sum_z \pi_{iz} \xi_{iz}^R} \sum_z d\ln \MW_z
            + \frac{1}
                    {\eta_{i} - \sum_z \pi_{iz} \xi_{iz}^R} \sum_z \nu_{iz} d\ln \MW_z .
    $$
    The second term on the right-hand side is equivalent to $\beta_i \mw_{i}^{\wkp}$
    in Proposition \ref{prop:representation}.
    The third term reflects the heterogeneity.
    If $\nu_{iz}$ has a symmetric distribution, and $d\ln \MW_z$ is the same across 
    workplaces (because it originates from a single jurisdiction), then this third 
    term will equal zero.}
Moreover, in our empirical exercises we allow for heterogeneity in elasticities
based on observable characteristics of workers, 
such as the share of MW workers residing in each location, empirically 
exploring this channel.

\paragraph{Representation under changing commuting shares.} 

Online Appendix \ref{sec:model_endogenous_shares} discusses the consequences
of relaxing the assumption of fixed commuting shares.
In particular, we assume that the commuting shares are declining in the MW level
at the workplace \parencite[as found by, e.g.,][]{PerezPerez2021}.
In this case, Proposition \ref{prop:representation_endog_shares} in the online
appendix shows that MW hikes at workplace locations will affect rents through 
two channels:
a positive channel via increases in housing demand, and
a negative channel via decreases in commuting shares.
Therefore, potential declines in commuting shares at the month of the MW change 
will tend to attenuate the positive effect of the workplace MW on rents, biasing 
our estimates towards not finding an effect.
Alternatively, if one thinks that commuting shares adjust more slowly than
posted rental prices, then an increase in the workplace MW will tend to
lower rents after the MW change.
We do not find evidence of this in our estimates, suggesting that the assumption
of fixed commuting shares when focusing on monthly changes in rents is
reasonable.

\section{Context and Data}\label{sec:data}

We begin the section by describing the construction of a ZIP code by month panel
of MW levels in the US.
We use our panel to describe trends in MW policies in the 2010s.
Later, we discuss the relationship between income and housing consumption at the
household level.
The data suggests that rents are likely to respond to MW changes.
We also explore how housing expenditure varies across ZIP codes.
Finally, we document the construction of our analysis sample and discuss
its strengths and limitations.

\subsection{Minimum Wage Policies in the 2010s}
\label{sec:data_mw_panel}

We collected data on federal-, state-, county-, and city-level statutory MW 
levels from \textcite{VaghulZipperer2016}.
We extended their data, available up to 2016, using data from 
\textcite{BerkeleyLaborCenter} and from official government offices for the 
years 2016--2020.%
\footnote{Some states and cities issue different MW levels for small businesses
    (usually identified by having less than 25 employees).
    In these cases, we select the general MW level as the prevalent one.
    In addition, there may be different (lower) MW levels for tipped employees.
    We do not account for them because employers are typically required to make 
    up for the difference between the tipped MW plus tips and the actual MW.}
%
%
Most ZIP codes are contained within a jurisdiction, and for them the statutory 
MW is simply the maximum of the federal, state, and local levels.
Some ZIP codes cross jurisdictions, and so are bound by multiple statutory MW 
levels.
In these cases we assign a weighted average of the statutory MW levels in its
constituent census blocks, exploiting a novel correspondence table between 
blocks and ZIP codes detailed in Online Appendix \ref{sec:blocks_to_uspszip}, 
where weights correspond to the number of housing units.
The result is a ZIP code-month panel of statutory MW levels in the US between
January 2010 and June 2020.
More details on the construction of the panel can be found in Online Appendix 
\ref{sec:assigning_mw_levels}.

Online Appendix Figure \ref{fig:mw_policies} shows the different levels of 
binding MW policies over time in our data.
Panel A focuses on state-level MW policies.
There are $\stateBindingMW$ states with MW policies in 2010--2020, all of 
which started prior to January 2010.
Panel B shows sub-state MW policies.
In total, there are $\localBindingMW$ counties and cities with some binding MW
policy in this period.
The number of new local jurisdictions instituting a MW policy increases strongly 
after 2013 and declines after 2018.
Overall, we observe strong variations in MW levels across jurisdictions.

Figure \ref{fig:map_mw_perc_changes} maps the percentage change in 
the statutory MW level from January 2010 to June 2020 in each ZIP code.
We observe a great deal of spatial heterogeneity in MW levels within the US.
Importantly, many metropolitan areas across and within state borders have 
differential MW changes, which will be central to distinguishing the effect 
of the two MW-based measures proposed in Section \ref{sec:model}.
We describe the construction of these measures later in this section.

\subsection{Households, Income, and Housing}
\label{sec:data_income_housing}

We compare individuals and households within metropolitan areas using data 
from the 2011 and 2013 waves of the American Housing Survey \parencite{ahs2020}.
Figure \ref{fig:ahs_pr_renters} shows that low-income households are much
more likely to rent.
While only $\TopDecPrRent$ percent of households in the top income quintile 
are renters, around $\BottomDecPrRent$ percent of them are when focusing on 
the bottom one.
Online Appendix Figure \ref{fig:ahs_hhead} shows that, while low income 
individuals are less likely to be household heads, many of them are.
The average probability for the bottom three income deciles is 50 percent.
Online Appendix Figure \ref{fig:ahs_rent_sqft} shows that, among households 
that rent, rents per square foot are surprisingly constant across household 
income levels.
These facts suggest that the MW is likely to affect household income, at least 
for lower income households, and that rents per square foot can plausibly
respond to MW changes.
Online Appendix Figure \ref{fig:ahs_unit_types} shows the type of building 
households live in by household income decile.
Low-income households are more likely to live in buildings with more units,
though they are spread across all building types.


We explore variations over space in housing expenditure.
To do so, we collected Individual Income Tax Statistics aggregated at the 
ZIP code level from the IRS \parencite{IRS},
and Small Area Fair Market Rents (SAFMRs hereafter) data from the HUD 
\parencite{hudSAFMR}.%
\footnote{SAFMRs data are constructed by the HUD as an extension of the
    Fair Market Rents (FMRs) data using, for each year, ZIP code-level information
    from previous years' American Community Survey 
    \parencite[][, p.\ 35]{SafmrReport2018}.
    SAFMRs are an estimate of the 40th percentile of the rents distribution
    based on constant housing quality \parencite[][, p.\ 1]{SafmrReport2018}.}
For each ZIP code in 2018, we constructed a housing expenditure share dividing 
the 2 bedroom SAFMR rental value from the HUD by the average monthly wage per 
household from the IRS.%
\footnote{We impute a small share of missing values using a regression model 
	where the ZIP code-level covariates include data from LODES and the US 
	Census.
	See Online Appendix \ref{sec:measure_housing_expenditure} for details.}%
\textsuperscript{,}%
\footnote{This computation will be a good approximation for the housing 
	expenditure share insofar total housing expenditure and total wage income 
	are roughly proportional to their averages under the same constant of 
	proportionality.
	This computation also assumes away differences in the number of bedrooms 
	across ZIP codes.}
Online Appendix Figure \ref{fig:map_hous_exp_chicago} maps our estimates for the 
Chicago CBSA.
There is considerable variation in housing expenditure over space, with poorer
areas generally spending a higher share of their income on housing.

To get a sense of the spatial distribution of minimum wage earners we construct 
a proxy variable using the number of workers across income bins in the 5-year 
2010-2014 American Community Survey \parencite[ACS;][]{CensusACS}.
See details in Online Appendix \ref{sec:assigning_mw_levels}.
Our variable for the share of MW workers is negatively correlated with median 
household income from the ACS (corr.\ $=\corrShWkrMedInc$) and 
positively correlated with our estimate of the housing expenditure share 
(corr.\ $=\corrShWkrS$).
This latter correlation also suggests that the MW is likely to affect rents.

\subsection{Estimation Data and Samples}

\subsubsection{Rents Data}
\label{sec:data_rents}

Zillow is the leading online real estate platform in the US, hosting more than 
170 million unique monthly users in 2019 \parencite{ZillowFacts}.
Zillow provides the median rental and sales price among units listed on the 
platform for different house types and at different geographic and time 
aggregation levels \parencite{ZillowData}.%
\footnote{As of the release of this article, the data are no longer available 
    for download.
    See \textcite{ZillowDataArchive} for a snapshot of the website as of 
    February 2020, the last month the data were available.}
We collected the ZIP code level data, available from February 2010 
to December 2019.
There is variation in the entry of a ZIP code to the data, and locations with a 
small number of listings are omitted. %

Our main analyses use the median rental price per square foot among housing 
units listed in the category Single Family houses, Condominium and Cooperative 
units (SFCC).
This is the most populated time series, as it includes the most common US 
rental house types \parencite{Fernald2020}.
We focus on rents \textit{per square foot} to account for systematic 
differences in housing size.
Online Appendix Figure \ref{fig:trend_zillow_safmr} shows that this series 
follows a similar trend over time when compared to SAFMR.
It is important to note that these data reflect rents of newly available 
units, for which new information is likely to be quickly incorporated into 
prices \parencite{AmbroseEtAl2015}.
As a result, we expect them to react quickly to economic shocks, such as 
changes in the MW.
On the other hand, rents among the universe of leased units should react 
more slowly, as they are only updated when the lease is renewed.
This is the pattern of results in \textcite{AgarwalEtAl2022} who use
data from contract rents.

We use Zillow's Observed Rental Index (ZORI) for a robustness check,
computed following the repeat rent index methodology used by
\textcite{AmbroseEtAl2015}, among others.
To compute the index, \textcite{ZillowZORI} uses all units with more than one 
transaction in a ZIP code and estimates a weighted regression of the 
log of the change in rents between two months on year-month indicators, 
upweighting observations that correspond to housing types that are 
underrepresented in the Zillow sample relative to census data.
The published index is a smoothed version of these coefficients, achieved 
through a three-month moving average that uses data from previous months. 
To account for this smoothing, we shift our MW measures when using the index as 
outcome in our regression models.

The Zillow data have several limitations.
First, Zillow's market penetration dictates the sample of ZIP codes available.
Online Appendix Figure \ref{fig:map_zillow_sample} shows that the sample of ZIP 
codes with SFCC rents data coincides with areas of high population density.
Second, we only observe the median rental value.
No data on the distribution of rents, nor the number of units listed for rent, 
are available.
Finally, we observe posted rents rather than contract rents.
We did not find any information on how correlated posted rents and contract 
rents are, so we decided to ask this as a question in the online platform Quora.
Online Appendix \ref{sec:posted_rents} shows a selection of quotes from the 
replies, which overwhelmingly suggest that contract rents generally do not 
differ from posted rents.
Some answers also suggest that rents of long-tenured landlords may not
reflect current market conditions.

\subsubsection{The residence and workplace minimum wage measures}
\label{sec:data_mw_measures}

Using the panel described in Section \ref{sec:data_mw_panel} at hand, computing 
the residence MW is straightforward.
We define it as $\mw^{\res}_{it} = \ln \MW_{it}$.

To construct the workplace MW we need commuting data, which we obtain from the 
Longitudinal Employer-Household Dynamics Origin-Destination Employment Statistics 
\parencite[LODES;][]{CensusLODES} for the years 2009 through 2018.
We collected the datasets for ``All Jobs.''
The raw data are aggregated at the census block level. 
We further aggregate it to ZIP codes using the original correspondence between 
census blocks and USPS ZIP codes described in Online Appendix 
\ref{sec:blocks_to_uspszip}.
This results in residence-workplace matrices that, for each ZIP code and year, 
indicate the number of jobs of residents in every other ZIP code.

We use the 2017 residence-workplace matrix to build exposure weights.
Let $\Z(i)$ be the set of ZIP codes in which $i$'s residents work 
(including $i$).
We construct the set of weights $\{\pi_{iz}\}_{z\in\Z(i)}$ as 
$ \pi_{iz} = N_{iz}/{N_i} , $
where 
$N_{iz}$ is the number of jobs with residence in $i$ and workplace in $z$, 
and $N_i$ is the total number of jobs originating in $i$.
The workplace MW measure is defined as
\begin{equation*}\label{eq:mw_wkp_def}
    \mw^{\wkp}_{it} = \sum_{z\in\Z(i)} \pi_{iz} \ln \MW_{zt} \ .
\end{equation*}
The workplace MW has a shift-share structure.
Our strategy, which exploits differential exposure to common shocks for 
identification, is most related to recent work in this area by 
\textcite{GoldsmithpinkhamEtAl2020}.

While our baseline uses commuting shares from 2017,
for robustness we present estimates in which the workplace MW measure
is constructed using alternative sets of weights.
In particular, we use different years and alternative job categories,
such as jobs for young or low-income workers.%
\footnote{The LODES data reports origin-destination matrices for the number of 
    workers 29 years old and younger, and the number of workers earning less 
    than \$1,251 per month.
The resulting workplace MW measures with any set of weights are highly correlated 
among each other (corr.\ $>0.99$ for every pair).}

Figure \ref{fig:map_mw_chicago_jul2019}, already discussed in the introduction,
illustrates the difference in the MW-based measures mapping their change in the 
Chicago CBSA on July 2019.
For completeness, Online Appendix Figure \ref{fig:map_rents_chicago_jul2019} 
shows the changes in our main median rents variable around the same date.

\subsubsection{Other data sources}\label{sec:data_other}

While our MW assignment recognizes that ZIP codes cross census geographies, 
we assign to each ZIP code a unique geography based on where the largest 
share of its houses fall.
We do this for descriptive purposes and also to use geographic indicators  
in our estimates.
Additionally, we collect ZIP code demographics from the ACS 
\parencite{CensusACS} and the 2010 US Census \parencite{CensusDecennial}.
We collect these data at the block or tract levels, and assign them to ZIP codes
using the correspondence table described in Online Appendix 
\ref{sec:blocks_to_uspszip}.

To proxy for local economic activity we collect data from the 
Quarterly Census of Employment and Wages \parencite[QCEW;][]{QCEW} 
at the county-quarter and county-month levels for several industrial divisions 
and from 2010 to 2019.%
\footnote{The QCEW covers the following industrial aggregates: 
    ``Natural resources and mining,'' ``Construction,'' ``Manufacturing,'' 
    ``Trade, transportation, and utilities,'' ``Information,'' 
    ``Financial activities'' (including insurance and real state), 
    ``Professional and business services,'' ``Education and health services,'' 
    ``Leisure and hospitality,'' ``Other services,'' ``Public Administration,''
    and ``Unclassified.''
    We observe, for each county-quarter-industry cell, the number of 
    establishments and the average weekly wage, and for each 
    county-month-industry cell, the level of employment.}
We use these data as controls for the state of the local economy in our 
regression models.

\subsubsection{Estimation samples}\label{sec:data_final_panel}

We put together an unbalanced monthly panel of ZIP codes available in Zillow in 
the SFCC category from February 2010 to December 2019.
This panel contains $\ZIPMWeventsUnbal$ MW changes at the ZIP code level, 
which arise from $\StateMWeventsUnbal$ state and $\CityCountyMWeventsUnbal$ 
county and city changes.
Online Appendix Figure \ref{fig:mw_changes_dist_zillow} shows the distribution 
of positive MW increases among ZIP codes in the Zillow data.
To prevent our estimates from being affected by changes in sample composition,
we construct a ``baseline'' panel keeping ZIP codes with valid rents data 
starting on January 2015.
The resulting fully-balanced panel contains $\ZIPMWeventsBase$ MW changes at 
the ZIP code level.%
\footnote{To avoid losing observations in models with leads and lags we include 
    six leads and lags of the MW measures, so the dataset actually runs from 
    July 2014 to June 2020.}

Table \ref{tab:stats_zip_samples} compares the sample of ZIP codes in the Zillow
data to the population of ZIP codes along sociodemographic dimensions.
The first and second columns report data for the universe of ZIP codes and 
for the set of urban ZIP codes, respectively.
The third column shows the set of ZIP codes in the Zillow data with any 
non-missing value of rents per square foot in the SFCC category.
Finally, the fourth column shows descriptive statistics for our estimation 
sample, which we call the ``baseline'' sample.
While our baseline sample contains only 11.8 percent of urban ZIP codes, it 
covers 25.0 percent of their population and 25.8 percent of their households.
With respect to demographics, ZIP codes in the baseline sample tend to be 
more populated, richer, with a higher share of Black and Hispanic inhabitants, 
and with a higher share of renter households than both the average ZIP code 
and the average urban ZIP code.
This is so because Zillow is present in large urban regions, but 
it does not usually operate in smaller urban or rural areas.
In an attempt to capture the treatment effect for the average urban ZIP code 
we conduct an exercise where we re-weight our sample to match the average 
of a handful of characteristics of those.

Finally, Online Appendix Table \ref{tab:stats_est_panel} shows statistics 
of our baseline panel.
The distribution of the residence and workplace MW measures is, as expected,
quite similar.
We also show median rents in Zillow in the SFCC category.
The average monthly median rent is \$1,757.9 and \$1.32 per square foot, 
although these variables show a great deal of variation.
Finally, we show average weekly wage, employment, and establishment count 
for the QCEW industries we use as controls in our models.

\section{Empirical Strategy}\label{sec:empirical_strategy}

In this section we discuss our empirical strategy.
We start with an intuitive presentation of our identification argument, which
is formalized in an appendix.
Next, we specialize our discussion under the functional form suggested by
the model in Section \ref{sec:model}. 
We also discuss alternative estimation strategies, concerns related to
the sample of ZIP codes we use, and heterogeneity of estimated effects.

\subsection{Intuitive Identification Argument}

Our data consist of rents, the residence and workplace MW measures, and 
economic controls.
We can learn the effect of the workplace MW from the slope of the relationship 
between the workplace MW and rents conditioning to places with a similar change 
in the residence MW.
We need to condition on the residence MW to remove confounding variation 
as it may affect locations through other channels, 
such as changes in prices of non-tradable consumption.
Likewise, a similar argument suggests that identifying the effect of the 
residence MW requires controlling for the workplace MW. 

For these slopes to correspond to causal effects, we need to make 
two assumptions.
The first one is a form of \textit{parallel trends}: among ZIP codes with 
the same residence MW, ZIP codes with higher and lower workplace MW levels would 
have had parallel trends in rents if not for the change in the workplace MW.
The second one is \textit{no selection on gains}: ZIP codes that receive
different levels of the workplace MW must experience a similar treatment effect
on average, conditional again on the residence MW.
If these assumptions hold, the (conditional) slope of the relationship between
the workplace MW and rents is actually the causal effect. 
We similarly need these assumptions to hold for the residence MW if we hope 
to give a causal interpretation to its coefficient.
Online Appendix \ref{sec:potential_outcomes} formalizes these assumptions in a 
potential outcomes framework following \textcite{CallawayEtAl2021}.
We discuss the plausibility of these assumptions later in this section.

\subsection{Parametric Model}

Consider the two-way fixed effects model relating rents and the MW measures
given by
\begin{equation} \label{eq:func_form}
    r_{it} = \alpha_i + \tilde{\delta}_t 
           + \beta \mw^{\wkp}_{it} + \gamma \mw^{\res}_{it} 
           + \mathbf{X}^{'}_{it}\eta
           + \upsilon_{it} ,
\end{equation}    
where
$i$ and $t$ index ZIP codes and time periods (months), respectively,
$r_{it}$ represents the log of rents per square foot,
$\mw^{\wkp}_{it} = \sum_{z\in\Z(i)} \pi_{iz}\ln \MW_{zt}$ is the ZIP code's 
workplace MW,
$\mw^{\res}_{it} = \ln \MW_{it}$ is the ZIP code's residence MW,
$\alpha_i$ and $\tilde{\delta}_t$ are fixed effects, and 
$\mathbf{X}_{it}$ is a vector of time-varying controls.
Time runs from January 2015 $\left(\underline{T}\right)$ 
to December 2019 $\left(\overline{T}\right)$.
The parameters of interest are $\beta$ and $\gamma$ which, 
following the model in Section \ref{sec:model}, 
we interpret as the elasticity of rents to the residence MW and the workplace MW, 
respectively.

By taking first differences in equation \eqref{eq:func_form} we obtain
\begin{equation}\label{eq:fd}
    \Delta r_{it} = \delta_t
                  + \gamma \Delta \mw^{\res}_{it} + \beta \Delta \mw^{\wkp}_{it}
                  + \Delta \mathbf{X}^{'}_{it}\eta
                  + \Delta \upsilon_{it} ,
\end{equation}
where $\delta_t = \tilde{\delta}_t - \tilde{\delta}_{t-1}$.
We estimate the model in first differences because we expect unobserved shocks
to rental prices to be serially autocorrelated over time, making the levels
model less efficient.
Online Appendix Table \ref{tab:autocorrelation} shows strong evidence of serial 
auto-correlation in the error term of the model in levels.
While estimated coefficients are similar in levels and in first differences, 
standard errors are seven to nine times larger in the former.

A standard requirement for a linear model like \eqref{eq:fd} to be
estimable is a rank condition, which implies that both MW measures must have 
independent variation, conditional on the controls.
For instance, if there were a single national minimum wage level or if everybody 
lived and worked in the same location, then we would have
$\Delta \mw^{\res}_{it} = \Delta \mw^{\wkp}_{it}$ for all $(i,t)$.
If so, $\gamma$ and $\beta$ could not be separately identified.
We check in the data that the rank condition is satisfied.

The main results of the paper are obtained under the model in \eqref{eq:fd}. 
In order to compare with the literature we also estimate versions of the 
model that exclude either one of the MW measures.

\subsection{Validity of Identification Assumptions}

The model in \eqref{eq:func_form} imposes a linear functional form.
This property rules out selection on gains, since then ZIP codes receiving a 
particular level of the MW measures will experience the same (constant) effect 
than ZIP codes that receive a different level.
This is one of the assumptions required for identification according to 
Online Appendix \ref{sec:potential_outcomes}.
We view this as a reasonable assumption.
For it to not hold, workers would need to anticipate not only future MW policies 
but also how future rental markets would be affected by them given the commuting 
structure, and select their residence so that rents react differently to the 
MW in different ZIP codes with similar levels of the MW measures.
We show in the results section that the (conditional) slope of log rents with 
respect to each of the MW measures appears linear, suggesting that the 
assumption of no selection on gains is plausible.

For estimates of $\beta$ and $\gamma$ from \eqref{eq:func_form} to have a causal 
interpretation we need another assumption: the error term $\Delta\upsilon_{it}$ 
must be \textit{strictly exogenous} with respect to the MW measures, and in 
particular with respect to the workplace MW.
This is in the spirit of parallel trends, the second assumption required for
identification in Online Appendix \ref{sec:potential_outcomes}.
This assumption implies that rents prior to a change in the workplace MW must
evolve in parallel.
We test for pre-trends adding leads and lags of the workplace MW in 
\eqref{eq:fd},%
\footnote{Specifically, we estimate
    \begin{equation*}
        \Delta r_{it} = \delta_t
                    + \gamma \Delta \mw^{\res}_{it} 
                    + \sum_{k=-s}^{s} \beta_k \Delta \mw^{\wkp}_{ik}
                    + \Delta \mathbf{X}^{'}_{it}\eta
                    + \Delta \upsilon_{it} ,
    \end{equation*}
    where $s=6$.
    Our results are very similar for different values of the window $s$.}
though we also experiment with adding leads and lags of the residence MW.
We only shift the workplace MW because estimating its effect is our focus, 
seeing the residence MW as a key control.
In fact, Online Appendix \ref{sec:potential_outcomes} suggests that we 
only need to condition on one of the MW measures for parallel trends of the 
other measure to hold.
Under the assumption of no anticipatory effects in the housing market, we 
interpret the absence of pre-trends as evidence against the presence 
of unobserved economic shocks driving our results.
Given the high frequency of our data and the focus on short windows around 
MW changes, the assumption of no anticipatory effects seems plausible.%
\footnote{We can also interpret the absence of pre-trends as a test for 
    anticipatory effects if we are willing to assume that the controls embedded 
    in $\mathbf{X}_{it}$ capture all relevant unobserved heterogeneity arising 
    from local business cycles.
    While we find the interpretation given in the text more palatable, the data 
    are consistent with both.}

Another implication of the strict exogenity assumption is that it allows for 
arbitrary correlation between $\alpha_i$ and both MW variables.
This means that our empirical strategy is robust to the fact that districts 
with more expensive housing tend to vote for MW policies.

We worry that unobserved shocks, such as those caused by local business cycles, 
may systematically affect both rent changes and MW changes, violating the 
strict exogeneity assumption.
To account for common trends in the housing market we include time-period 
fixed effects $\delta_t$, which in some specifications are allowed to vary by 
jurisdiction.
To control for variation in local labor markets trends we include economic 
controls from the QCEW in the vector $\mathbf{X}_{it}$.
Specifically, we control for average weekly wage and establishment counts at the 
county-quarter level, and for employment counts at the county-month level, 
for the sectors ``Professional and business services,'' ``Information,'' and 
``Financial activities.''%
\footnote{We assume that these sectors are not affected by the MW.
    In fact, according to the \textcite[][Table 5]{MinWorkersReportBLS}, in 
    2019 the percent of workers paid an hourly rate at or below the federal MW 
    in those industries was 0.8, 1.5, and 0.2, respectively.
    In comparison, 9.5 percent of workers in ``Leisure and hospitality'' were 
    paid an hourly rate at or below the federal MW.}
We also try models where we control for ZIP code-specific linear
trends, which should account for time-varying heterogeneity at the ZIP 
code-level that follows a linear pattern.

A second worry is that changes in the composition of rentals may drive the 
results.
For instance, if on the same month of the MW change more expensive listings go
to the market, then what looks like a rent increase may actually be changes in 
quality.%
\footnote{We thank an anonymous referee for pointing out this concern.}
We note that changes in housing size, which seem to be the key driver 
in price heterogeneity according to Online Appendix Figure 
\ref{fig:ahs_rent_sqft}, are controlled for because we use rents per square foot
as our outcome.
To more directly address this concern we present evidence using Zillow's 
observed rental index (ZORI), which is constructed using rental prices for the 
same housing unit in different moments in time.
Given that the index is averaged using three lags, a regression analysis
that relates period-$t$ MWs would be expected to affect the ZORI index
at $t-3$, and its first difference at $t-4$.
To adjust for this, in these models we use the 4th lead of the change in 
the MW measures.

\subsection{Alternative Strategies}\label{sec:alt_emp_strategies}

Recent literature has shown that usual estimators in a difference-in-differences 
setting do not correspond to well-defined average treatment effects when the 
treatment roll-out is staggered and there is treatment-effect heterogeneity 
\parencite{deChaisemartinEtAl2022,RothEtAl2022}.
While our setting does not correspond exactly to the models discussed in this
literature, we worry about the validity of our estimator.
To ease these concerns, in an appendix we construct a ``stacked'' implementation 
of equation \eqref{eq:fd} in which we take six months of data around MW changes 
for ZIP codes in CBSAs where some ZIP codes received a direct MW change and 
some did not, 
and then estimate the model on this restricted sample including event-by-time 
fixed effects.
This strategy limits the comparisons used to compute the coefficients of 
interest to ZIP codes within the same metropolitan area and event.

In a separate exercise we relax the strict exogeneity assumption.
We do so in an appendix as well, where we propose a model that includes 
lagged rents as an additional control.
In such a model, $\beta$ and $\gamma$ have a causal interpretation under a 
weaker \textit{sequential exogeneity} assumption
\parencite{ArellanoBond1991, ArellanoHonore2001}.
This alternative assumption requires innovations to rents to be uncorrelated 
only with past changes in the MW measures, and thus allows for feedback of 
rent shocks onto MW changes in future periods.
We estimate this model using an IV strategy in which the first lag of the 
change in rents is instrumented with the second lag.

\subsection{Heterogeneity and Sample Selection Concerns}\label{sec:emp_start_heterogeneity}

We explore heterogeneity of our results with respect to pre-determined 
variables.
Given the mechanism proposed in Section \ref{sec:model}, we expect the effect 
of the residence MW to be stronger in locations where many workers earn 
close to the MW.
The reason is that the production of non-tradable goods presumably uses more
low-wage work, and thus the increase in the MW would affect prices more.
Similarly, we expect the effect of the workplace MW to be stronger in locations
with lots of MW workers as residents since income would increase more 
strongly there.
We then estimate the following model:
\begin{equation*}\label{eq:fd_heterogeneity}
    \Delta r_{it} = \Xi_t
                  + \tilde\gamma_0 \Delta \mw^{\res}_{it}
                  + \tilde\gamma_1 \iota_i \Delta \mw^{\res}_{it}
                  + \tilde\beta_0 \Delta \mw^{\wkp}_{it}
                  + \tilde\beta_1 \iota_i \Delta \mw^{\wkp}_{it}
                  + \Delta \mathbf{X}^{'}_{it}\tilde\eta
                  + \Delta \tilde\upsilon_{it} ,
\end{equation*}
where $\iota_i$ represents the standardized share of MW workers residing in $i$.
Because we cannot estimate the share of MW workers working in a given location,
we interact both the residence and workplace MW with the estimated share of 
MW residents.
We conduct a similar exercise using median household income and the share of 
public housing units.

As explained in Section \ref{sec:data_final_panel}, 
the model in equation \eqref{eq:fd} relies on a selected sample.
In an alternative estimation exercise we use an unbalanced panel with all 
ZIP codes with Zillow rental data in the SFCC category 
from February 2010 to December 2019, controlling for time period by 
quarterly date of entry fixed effects.
However, even all ZIP codes available in the Zillow data may be 
a selected sample of the set of urban ZIP codes.
To approximate the average treatment effect in urban ZIP codes we follow
\textcite{Hainmueller2012} and estimate our main models re-weighting 
observations to match key moments of the distribution of characteristics of 
those.

\section{Estimation Results}\label{sec:results}

In this section we present our empirical results.
First, we show our baseline estimates and discuss our identifying assumptions
and other robustness checks.
Second, we present results of models that use alternative empirical strategies.
Third, we discuss concerns that arise from the selectivity of our sample of 
ZIP codes and show heterogeneity analyses. 
Finally, we summarize our results and compare them with existing literature.

\subsection{Baseline estimates}
\label{sec:results_main}

\subsubsection{Main results}

Table \ref{tab:static} displays our estimates using the baseline sample 
described in Section \ref{sec:data_final_panel} and the parametric model
in equation \eqref{eq:fd}.
Column (1) shows the results of a regression of the workplace MW on the 
residence MW, economic controls, and time fixed effects.
We observe that a 10 percent increase in the residence MW is associated with an
$\WkpOnResCoeffBaseTen$ percent increase in the workplace MW.
While the measures are strongly correlated, this estimate shows that this 
correlation is far from exact, confirming the presence of independent variation
that allows the inclusion of both MW measures in our models.

Columns (2) and (3) of Table \ref{tab:static} show estimates of models that 
include a single MW measure.
Column (2) uses only the residence MW.
In this model, only locations with a statutory MW change are assumed to 
experience effects, similar to the existing literature.
The elasticity of rents to the MW is estimated to be 
$\OnlyResGammaBase$ ($t=\OnlyResGammaBasetStat$).
Column (3) uses solely the workplace MW.
The coefficient on the MW variable increases slightly to 
$\OnlyWkpBetaBase$ ($t=\OnlyWkpBetaBasetStat$), 
supporting the view that changes in the workplace MW are a better measure of 
the changes in the MW that are relevant for a ZIP code.

Column (4) of Table \ref{tab:static} show estimates of equation \eqref{eq:fd}
including both MW measures.
The coefficient on the workplace MW ($\beta$) increases to $\BothBetaBase$ and 
is statistically significant ($t=\BothBetaBasetStat$).
These results suggest that omitting the residence MW generates a downward
bias on the coefficient of the workplace MW.
Consistent with the theoretical model in Section \ref{sec:model}, the 
coefficient on the residence MW ($\gamma$) now turns negative and equals 
$\BothGammaBase$, although it is not statistically significant 
($t=\BothGammaBasetStat$).
We reject the hypothesis that $\gamma=\beta$ at the 10\% significance level 
($p = \GammaEqBetaBasePval$).
Finally, $\gamma+\beta$ is estimated to be $\BothSumBase$, which is highly 
significant ($t=\BothSumBasetStat$).
Thus, our results imply that a 10 percent increase in both MW measures will 
increase rents by approximately $\BothSumBaseTen$ percent.
However, our results also imply substantial heterogeneity across space.
If only the residence MW increases then rents are expected to decline,
and if only the workplace MW goes up then the rents increase will likely 
be larger.

\subsubsection{Identification assumptions}

A central concern with these results is whether our identifying assumptions are 
likely to hold.
Panel A of Figure \ref{fig:dynamic_workplace} shows estimates of the 
parametric first-differences model including leads and lags of the workplace MW, 
so that the coefficients are 
$\{\{\beta_s\}_{s=-6}^{-1},\beta,\{\beta_s\}_{s=1}^6,\gamma\}$.
We cannot reject the hypothesis that $\beta_{-6}=...=\beta_{-1}=0$ 
($p = \BetaPretrendDynBasePVal$).
Post-event coefficients $\{\beta_s\}_{s=1}^6$ are also estimated to be 
statistically zero.
The only significant estimate is $\beta$ at $\BothWkpDynBetaBase$ 
($t=\BothWkpDynBetaBasetStat$).
The estimate of $\gamma$ is $\BothWkpDynGammaBase$ 
($t=\BothWkpDynGammaBasetStat$).
We now reject the hypothesis of equality of coefficients at the 5\% 
significance level ($p = \GammaEqBetaBaseDynPval$).
Our estimate of $\gamma+\beta$ is $\BothWkpDynSumBase$.
It is significant ($t=\BothWkpDynSumBasetStat$) and almost identical to our 
baseline.
Panel B shows the implied effects in levels of our first-differences models,
assuming that pre- and post-coefficients are zero.
Online Appendix Figure \ref{fig:dynamic_residence} shows that a similar story 
obtains when we add leads and lags of the residence MW only.
We interpret these results as evidence in favor of the parallel trends 
assumption.

Online Appendix Figure \ref{fig:non_parametric} plots the relationship between 
log rents and each of the MW measures for ZIP codes in CBSAs and months 
in which at least one residence MW changed.
Panel A displays the raw data, which shows a positive correlation between log 
rents and both MW measures.
Panel B displays the same relationships after residualizing each variable on 
ZIP code fixed effects and indicators for different values of the other MW 
measure.
We observe a positive slope for the workplace MW, and a negative one for
the residence MW.
This provides evidence in favor of the assumption of no selection on gains, and
also of the linear functional form assumed in equation \eqref{eq:func_form}. 
Furthermore, the slopes in these figures show a similar magnitude to our 
baseline estimates of $\gamma$ and $\beta$.

Online Appendix Figure \ref{fig:map_residuals_chicago_jul2019} illustrates the 
identifying variation we use by mapping the residualized change in workplace MW 
and the residualized change in log rents.%
\footnote{To maximize the number of ZIP codes with valid data on this map we
    use the results of the unbalanced panel discussed in Section 
    \ref{sec:results_heterogeneity}.}
Panel A of Online Appendix Figure \ref{fig:map_residuals_chicago_jul2019}, to be 
contrasted with the left panel of Figure \ref{fig:map_mw_chicago_jul2019}, 
shows that the residualized change in the workplace MW is high outside of Cook 
County, where the statutory MW increased.
For completeness, Panel B of Online Appendix Figure 
\ref{fig:map_residuals_chicago_jul2019} shows residualized rents.

\subsubsection{Zillow's repeat rental index}

Online Appendix Figure \ref{fig:dynamic_zori} shows regression results
using the ZORI index.
Since the index has more missing values than our baseline before 2020,
we use the entire sample of ZIP codes and set the window to 4.%
\footnote{While our baseline analysis in Table \ref{tab:static} uses 80,241 
    observations, the models that are estimated using the ZORI index use 
    only 22,984.}
Panel A controls for CBSA by year-month fixed effects, and shows 
a strong increase in rents following a change in the workplace MW.
We observe that the effects lasts for a few months, which is to be expected
as this variable is computed as an average over 3 months.
The residence MW has a negative coefficient, although it is not statistically 
significant.
We observe precisely estimated null pre-trends.
Panel B controls for year-month fixed effects only, and shows similar though 
weaker patterns.
We see this evidence as supportive of the view that the estimated effects
using our main rental variable are not driven by changes in the composition
of listings that coincide with changes in the MW.

\subsubsection{Robustness checks}

Table \ref{tab:robustness} shows how our results change when we vary the
specification of the regression model and the commuting shares used 
to construct the workplace MW measure.
Each row of the table shows estimates analogous to those of columns (1) and (4)
of Table \ref{tab:static}.

Panel A of Table \ref{tab:robustness} groups the results when varying the 
regression equation.
Row (b) shows that our results are very similar when we exclude the 
economic controls from the QCEW.
Rows (c) and (d) show that interacting our time fixed effects with indicators 
for county or CBSA yields similar conclusions.
In all these cases our baseline point estimates are contained in relevant 
confidence intervals and, in the case of CBSA by month fixed effects, 
the results seem even larger.
This supports the view that our results are not caused by regional trends 
in housing markets correlated with our MW variables.
Row (e) shows that the results are non-significant when using state by 
monthly date fixed effects.
While our baseline estimates are within relevant confidence intervals, the 
signs of the point estimates are flipped.
However, these results are much noisier.
In fact, a formal test does not reject the hypothesis that the coefficients
on the workplace MW ($p=0.2168$) and the residence MW ($p=0.2668$) are 
equal to our baseline.
Row (f) includes ZIP code fixed effects in the first differences model, which
is equivalent to allowing for a ZIP code-specific linear trend in the model in 
levels.
These results are also very similar to our baseline.

Panel B of Table \ref{tab:robustness} estimates the baseline model but 
computing the workplace MW using alternative commuting structures.
Rows (g) and (h) use commuting shares from 2014 or 2018 instead of 2017.
Row (i) allows the commuting shares to vary by year, introducing additional
cross-year variation in the workplace MW measure that does not arise from 
changes in the statutory MW.
The fact that these specifications yield very similar results suggests that 
changes in commuting correlated with MW changes are unlikely to be the driver
of the results.
Furthermore, as discussed in the last subsection of Section \ref{sec:model_summary},
changing commuting shares would load on the workplace MW.
However, we see no effects of the workplace MW beyond the month of the 
change, suggesting that commuting shares are empirically quite stable.
Finally, rows (j) and (k) use 2017 commuting shares for workers earning 
less than \$1,251 per month and workers that are less than 29 years old, 
respectively.
If anything, the results seem stronger and more significant in this case, 
consistent with the idea that these workers are more likely to earn close to the 
minimum wage.

\subsubsection{County-level estimates}

Can we rely on larger geographical units to estimate the effect of the 
workplace MW?
We compare our results with estimates obtained from an alternative panel 
where the unit of observation is the county by month, using median rents 
from Zillow and commuting shares aggregated at this level.
Online Appendix Figure \ref{fig:dynamic_county_month} shows our estimates 
of a dynamic model.
The point estimates on impact have the same sign as our main results,
although they are much noisier.
We also observe some evidence of pre-trends in the coefficients of this 
model, suggesting that estimates obtained at a larger geographical 
resolution may not rely on plausibly exogenous identifying variation.
Finally, we note that these estimates would miss rich heterogeneity of 
effects within counties.

\subsection{Alternative Strategies}
\label{sec:results_alternative_strategies}

Online Appendix Table \ref{tab:stacked_w6} estimates our main models using a 
``stacked'' sample, as discussed in Section \ref{sec:alt_emp_strategies}.
Our sample contains 184 ``events,'' that is, CBSA-month pairs that had some 
strict subset of ZIP codes increasing the residence MW and had at least 10
ZIP codes.
These estimates interact the year-month fixed effects with event ID indicators, 
limiting comparisons to ZIP codes in the same event.
This is in line with recent literature that focuses on carefully selecting the 
comparison groups in difference-in-differences settings
\parencite{deChaisemartinEtAl2022, RothEtAl2022}.
We find that our key MW-based measures have little predictive power on their own,
but the model including both measures yields similar patterns as our baseline.
If anything, results are stronger in this case.
Now, both MW measures are strongly significant.
A 10 percent increase in both MW measures is estimated to increase rents 
by $\BothSumStackTen$ percent.
Online Appendix Figure \ref{fig:dynamic_stacked} shows the results of a similar 
model that includes leads and lags of the workplace MW.
Estimates of leads and lags are statistically non-distinguishable from zero.
However, they are noisier than in our baseline.

Online Appendix Table \ref{tab:arellano_bond} shows estimates of a model that 
includes the lagged difference in log rents as a covariate.
This specification relaxes the strict exogeneity assumption and allows for 
feedback effects of rent increases on the MW measures.
To avoid the well-known endogeneity problem of including this covariate, the 
models are estimated using an IV strategy where we instrument the first lag of 
the change in rents with the second lag of this variable 
\parencite{ArellanoBond1991, ArellanoHonore2001}.
Columns (1) and (2) show estimates of models in levels, both of which imply
confidence intervals for the coefficients that include our preferred estimates.
Columns (3) and (4) show preferred models in first differences, where results
are very similar across strategies. 

\subsection{Heterogeneity and Sample Selection Concerns}
\label{sec:results_heterogeneity}

Table \ref{tab:heterogeneity} explores heterogeneity of our estimates.
Column (1) reproduces the baseline results.
Column (2) presents estimates interacting the MW measures with an estimated 
share of MW workers residing in each ZIP code.
At the mean share of MW workers, our estimates indicate that the coefficient
on the workplace MW is $\TildeBetaZero$ (SE $=\TildeBetaZeroSE$).
For a ZIP code that is one standard deviation above the average share of MW 
workers, the effect of the workplace MW is stronger at
$\TildeBetaZeroPlusBetaOne$ (SE$=\TildeBetaZeroPlusBetaOneSE$).
Housing more MW workers in a ZIP code implies that income is likely to be more 
sensitive to the MW and so, consistent with our model, the effect of the MW
on rents is larger.
The coefficient on the residence MW also presents significant heterogeneity.

Column (3) of Table \ref{tab:heterogeneity} interacts both MW measures with the
standardized median household income from the ACS.
We find analogous patterns to Column (2), as a higher median income is 
correlated with a lower share of MW workers.
Column (4) interacts the MW measures with the standardized share of public 
housing units.
The effects for ZIP codes with more public housing seems larger,
although the coefficient on the interaction is not statistically significant. 
This result suggests that public housing does not necessarily diminish the 
scope for landlords to increase rents.
However, it is possible that this variable is capturing a high presence of 
low-wage residents and workers who, per our previous discussion, are more 
affected by the MW.

Online Appendix Table \ref{tab:static_sample} explores the sensitivity of our 
estimates to the sample of ZIP codes used in estimation.
Column (1) replicates our baseline estimates.
In column (2) we estimate the same model but re-weighting observations 
to match pre-treatment characteristics of the sample of urban ZIP codes,
defined in Table \ref{tab:stats_zip_samples}.%
\footnote{Our weights follow \textcite{Hainmueller2012} and are designed to 
    match the averages of three variables from the 2010 US Census:
    the share of urban households,
    the share of renter-occupied households, and
    the share of white households.}
The coefficient on the workplace MW is somewhat smaller, but it remains 
strongly significant.
Column (3) uses an unbalanced sample of ZIP codes and controls for 
quarter-of-entry by year-month fixed effects.
The coefficient on the workplace MW is again somewhat smaller than our baseline,
and now it is significant only at the 10\% level.

\subsection{Alternative rental categories}
\label{sec:alternative_categories}

Online Appendix Table \ref{tab:zillow_categories} shows how our results change 
when we use other rental categories available in the Zillow data.
For each rental variable we use an unbalanced panel that controls for
year-month fixed effects interacted with indicators for the quarter of entry
to the data in the given rental category.
We note that the number of observations varies widely across housing categories, 
and is always much lower than for our baseline SFCC variable.

Given the reduced precision of these estimates is hard to obtain strong 
conclusions on what type of housing is reacting more strongly to MW changes.
We observe that the sum of the coefficients on our MW variables is 
statistically significant at conventional levels in the categories 
``Single Family'' (SF),  ``Condominium and Cooperative Houses'' (CC), and 
``Multifamily 5+ units.''
Online Appendix Figure \ref{fig:ahs_unit_types} shows that low-wage households 
are likely to reside in these type of housing units.
However, the coefficients on each of the MW measures are typically much noisier 
than baseline.
We observe inconsistent results for the category ``1 bedroom,'' for which the 
sign of the coefficients is flipped relative to baseline.
However, these estimates are not statistically significant.

\subsection{Summary and Discussion}
\label{sec:results_discussion}

We find strong evidence that increases in the MW at workplace locations 
increase rents, supporting the view that MW policies spill over across local
housing markets through commuting.
Our baseline estimated elasticity of rents to the MW is $\BothBetaBase$.
The magnitude of our estimates is similar to estimates of the elasticity of 
restaurant prices to the MW 
(\cite{AllegrettoReich2018} finds an average restaurant price elasticity of 0.058),
and the elasticity of grocery store prices to the MW 
(\cite{Leung2021} finds an elasticity ranging from 0.06 to 0.08; see also 
\cite{RenkinEtAl2020}).

Our results are also consistent with existing research on the elasticity 
of wages to the MW.
For instance, estimates in \textcite{CegnizEtAl2019} (obtained using state MW 
events), and an assumed share of MW workers of 0.14 (as the average in Table 
\ref{tab:heterogeneity}), imply an elasticity of income to the MW of 0.094.%
\footnote{\textcite[][Table I]{CegnizEtAl2019} find that a ``MW event'' 
    increases wages by 6.8 percent, and their average MW event 
    represents an increase of 10.1 percent.
    Assuming that 14 percent of workers in a location earn the MW results 
    in an elasticity of $(6.8/10.1)\times 0.14 \approx 0.0943$.
    In their data, the authors find that ``8.6\% of workers were below the 
    minimum wage.''
    Our estimates of this share is likely larger because we account for local 
    MW policies.}
\textcite[][Table 1]{Hughes2020} finds an elasticity of household income of 
affected households to the MW of 0.189.
\textcite{Dube2019Income} finds minimum wage elasticities of household income 
between 0.152 and 0.430 for households at the lower end of the distribution.
\textcite{WiltshireEtAl2023} explores recent large MW increases in the US and 
finds an earnings elasticity of 0.18 for fast food workers.
For a given elasticity of wages to the MW $\varepsilon$, and a share of housing 
expenditure $s$, our elasticity of rents to the workplace MW implies that
$\gamma = 100 \times \left(s \BothBetaBase /\varepsilon\right)$ percent 
of the income generated by the MW is spent in housing.
For instance, if we assume $s=1/3$ and $\varepsilon=0.1$, then $\gamma=22.8$.

We find that our estimates of the elasticity of the MW to rents are on
the lower end of the range of estimates in the literature.
Relying on repeated cross-sectional survey data from the ACS, 
\textcite[][, Table 1]{Hughes2020} finds an elasticity of rents to the MW 
for affected households of 0.0543, similar to ours.
Our estimates are significantly smaller than those in 
\textcite{Yamagishi2021} and \textcite{AgarwalEtAl2022}.
Exploiting heterogeneity in MW levels across Japanese prefectures,
\textcite{Yamagishi2021} estimates an elasticity of 0.25--0.45.
\textcite{AgarwalEtAl2022} estimate the effect of the MW on rents using 
state-level MW changes between 2000 and 2009, and their findings
imply an elasticity of 0.73.
There are many factors that could account for the differential results.
For instance, while these authors use micro-data on total rental payments,
we rely on the median per-square foot rent in a ZIP code.
Additionally, \textcite{Yamagishi2021} focuses on low-quality apartments.
While precisely pinning down the source of the differences is hard,
given the existing evidence of the effect of the MW on income and consumption 
prices we see our results as plausible.

\section{Counterfactual Analysis}\label{sec:counterfactual}

We use our empirical results to explore the incidence of counterfactual 
MW policies across space.
We evaluate two policies:
an increase in the federal MW from \$7.25 to \$9, and
an increase in the local MW of the city of Chicago from \$13 to \$14.
To measure incidence, we compute the share of the extra income generated by the 
policy that is pocketed by landlords.

\subsection{Empirical Approach}\label{sec:emp_cf}

Following the notation in Section \ref{sec:model}, define the ZIP code-specific 
share pocketed by landlords as
\begin{equation*}
    \rho_i = \frac{\Delta H_i R_i}{\Delta Y_i} 
           = \frac{H^{\post}_i R^{\post}_i - H^{\pre}_i R^{\pre}_i}{\Delta Y_i} 
\end{equation*}
where
``$\pre$'' and ``$\post$'' denote moments before and after the MW change,
$H_i R_i = \sum_{z\in\Z(i)} H_{iz} R_{iz}$ denotes total housing expenditure 
in $i$, and
$Y_i = \sum_{z\in\Z(i)} Y_{iz}$ denotes total wage income in $i$.

Changes in rented square footage (if any) are unobserved.
Therefore, we assume $H^{\pre}_i = H^{\post}_i = H_i$ and the share becomes
\begin{equation}\label{eq:share_pocketed}
    \rho_i = \frac{H^{\post}_i R^{\post}_i - H^{\pre}_i R^{\pre}_i}{\Delta Y_i} = 
                H_i \frac{\Delta R_i}{\Delta Y_i} .
\end{equation}
If $\Delta H_i > 0$ instead, then our estimates of $\rho_i$ will be a lower 
bound.

We predict rent changes for all ZIP codes using the model in equation 
\eqref{eq:fd}.
Because we are interested in the partial effect of the policy, we hold constant 
common shocks affecting all ZIP codes,
local economic trends reflected in the controls, and
idiosyncratic shocks that show up in the error term.
Then,
\begin{equation}\label{eq:cf_rents_model}
    \Delta r_i = \beta \Delta \mw_i^{\wkp} + \gamma \Delta \mw_i^{\res} .
\end{equation}
We define the change in log total wages using a first-differenced model as well:
\begin{equation}\label{eq:cf_wages_model}
    \Delta y_i = \varepsilon \Delta \mw_i^{\wkp} ,
\end{equation}
where $y_i=\ln Y_i$.
The residence MW is excluded because we are considering the effect of the MW on 
nominal wages.
Estimates of the income elasticity $\varepsilon$ are not readily available 
in the literature.
Given the discussion in Section \ref{sec:results_discussion}, we set a 
baseline value of $\varepsilon = 0.1$.
However, we show how our results change for different values of $\varepsilon$.

Assuming that we know the value of $\varepsilon$, we can substitute
\eqref{eq:cf_rents_model} and \eqref{eq:cf_wages_model} into equation
\eqref{eq:share_pocketed} to obtain
\begin{equation*}\label{eq:rho}
    \begin{split}
        \rho_i & = H_i \left[ 
        \frac{\exp \left(\Delta r_i + r_i \right) - R_i }
             {\exp \left(\Delta y_i + y_i \right) - Y_i }
        \right] \\
        & = s_i \left[
            \frac{\exp \left( \beta \Delta \mw_i^{\wkp} + \gamma \Delta \mw_i^{\res} \right) - 1 }
                {\exp \left( \varepsilon \Delta \mw_i^{\wkp} \right) - 1 }
            \right]
    \end{split}
\end{equation*}
where $s_i = \left(H_i R_i\right)/Y_i$ is the share of $i$'s expenditure in 
housing.
As discussed in Section \ref{sec:data_income_housing},
we estimate this share as the ratio of the 2-bedroom SAFMR rental value, 
$\tilde R_i$, and monthly average wage per household, $\tilde Y_i$.

We also compute the total incidence of the policy on ZIP codes $i\in\Z_1$,
for some subset $\Z_1\subseteq\Z$, as follows:
\begin{equation*}\label{eq:tot_incidence}
    \rho_{\Z_1} = 
        \frac{\sum_{i\in\Z_1} \tilde R_i \left(\exp \left( \beta \Delta \mw_i^{\wkp} 
                                    + \gamma \Delta \mw_i^{\res} \right) - 1\right) }
            {\sum_{i\in\Z_1} \tilde Y_i \left( \exp \left( \varepsilon \Delta \mw_i^{\wkp} \right) 
                                    - 1\right) } .
\end{equation*}
In words, total incidence is defined as the ratio of the total change in rents
per household in $\Z_1$ to the total change in wage income per household 
in $\Z_1$.

\subsection{Results}\label{sec:results_cf}

We use our estimates to compute the shares 
$\{\{\rho_i\}_{i\in\Z_1},\rho_{\Z_1}\}$ for two counterfactual scenarios:
an increase of the federal MW from \$7.25 to \$9 and 
an increase in the Chicago City MW from \$13 to \$14.
In the federal case, we let $\Z_1$ be the set of ZIP codes located in urban 
CBSAs (as defined in Table \ref{tab:stats_zip_samples}) and exclude ZIP codes 
that are part of a CBSA where the average estimated increase in log total wages 
is less than 0.1\%.%
\footnote{\label{foot:restriction_on_zipcodes}
The goal of this restriction is to exclude metropolitan areas located 
in jurisdictions with a MW level above the new counterfactual federal level.
Because all those ZIP codes experience a small and similar increase in the 
workplace MW, the estimated share pocketed will be equal to the estimated
housing expenditure share times the constant 
$\left(\exp(\beta x)-1\right)/\left(\exp(\varepsilon x)-1\right)$,
where $x$ is the value of the workplace MW increase.
These estimates, however, are not economically meaningful because the increase
in income due to the policy is negligible.}
In the local case, we let $\Z_1$ represent ZIP codes in the 
Chicago-Naperville-Elgin CBSA, which are the most exposed to this policy.

\subsubsection{Counterfactual increases in residence and workplace MW levels}
\label{sec:cf_res_and_wkp_changes}

We compute the counterfactual statutory MW in January 2020 at a given ZIP code 
by taking the max between (i) the state, county, and local MW in December 2019, 
and (ii) the assumed value for the federal or city MW in January 2020.%
\footnote{To be more precise, we take the maximum between the MW levels of 
different jurisdictions at the level of the block.
Then, we aggregate up to ZIP codes using the correspondence table in 
Online Appendix \ref{sec:blocks_to_uspszip}.
We do so to account for the fact that the new MW policy may be partially 
binding in some ZIP codes.}
Then, we compute the counterfactual values of the residence MW and the workplace
MW following the procedure outlined in Section \ref{sec:data_mw_measures}.
Like in our baseline estimates, we use commuting shares for all workers in 2017.

\paragraph{Federal increase.}

The distributions of counterfactual increases in the MW measures are displayed 
in Online Appendix Figure \ref{fig:cf_hist_res_and_wkp_mw}.
Out of the $\zipcodesFedNine$ ZIP codes that satisfy our criteria, 
$\zipNoIncFedNine$ (or $\zipNoIncPctFedNine$\%) experience no increase in 
the residence MW at all.
The residence MW increases in $\zipIncFedNine$ ZIP codes  (or $\zipIncPctFedNine$\%), 
$\zipBoundFedNine$ of which were bound by the previous federal MW, and 
so the residence MW increases by $\ln(9)-\ln(7.25)\approx 0.2162$ in them.
Correspondingly, we observe mass points in the distribution of the residence MW,
with the two largest ones at $0$ and $0.2162$.
Since many people reside and work under the same statutory MW, the mass points 
are still visible in the histogram of the workplace MW.
However, we observe more places experiencing moderate increases in this measure.

Panel A of Online Appendix Figure \ref{fig:map_chicago_cf_wkp_res} maps the 
changes in the residence and workplace MW in the Chicago-Naperville-Elgin CBSA.
Unlike in Figure \ref{fig:map_mw_chicago_jul2019}, we observe the MW increasing 
from the outside of Cook County and spilling over inside it.

\paragraph{Local increase.}

In our second counterfactual experiment we increase the Chicago City MW 
from \$13 to \$14 on January 2020, keeping constant other MW policies.
Importantly, under this assumption the difference between the Chicago
and Cook County MW levels increases by \$1.

In this case, there are $\zipIncChiFourteen$ ZIP codes whose 
residence MW are affected by this change and $\zipNoIncChiFourteen$ 
that remain directly unaffected.
Panel B of Online Appendix Figure \ref{fig:map_chicago_cf_wkp_res} shows the 
changes in both MW measures after this policy.
As expected, we observe large increases in the workplace MW in the city, 
which become smaller as one moves away from it.

\subsubsection{The share of extra wage income pocketed by landlords}
\label{sec:cf_rents_and_wage_changes}

We couple the counterfactual increases in residence and workplace MW with 
estimates of $\beta$ and $\gamma$.
Following the results in Table \ref{tab:static}, we take 
$\beta = \betaCf$ and 
$\gamma = \gammaCf$.
Based on previous discussion, we set $\varepsilon = 0.1$.
We follow the procedure outlined in the previous subsection to estimate the 
incidence of the counterfactual policy.

\paragraph{Federal increase.}

Panel A of Figure \ref{fig:cf_hist_shares} displays a histogram of the
estimated shares $\{\rho_i\}_{i\in\Z_1}$.
The median share is $\rhoMedianFedNine$, which implies that at the median 
ZIP code landlords capture roughly $\rhoMedianCentsFedNine$ cents of each 
additional dollar generated by the MW change.
The distribution of the shares is skewed to the right.
However, we observe a long left-tail with a few negative values which arise due 
to declines in rents in locations where the increase in the residence MW is much 
larger than the increase in the workplace MW.

Panel A of Figure \ref{fig:map_chicago_cf_shares} maps the estimated shares 
in the Chicago-Naperville-Elgin CBSA.
Panel A of Online Appendix Figure \ref{fig:map_chicago_cf_rents_wages} shows
estimated increases in rents and wage income.
We estimate a larger share pocketed in Cook County.
The reason is that these ZIP codes experience the new policy only through
their workplace MW and, as a result, rents increase relatively more than
wage income.
We also observe a larger incidence on landlords in the south of Cook County,
where the housing expenditure share is larger.

The top rows of Panel A in Table \ref{tab:counterfactuals} show the medians of 
the key estimated objects for two groups:
ZIP codes where the residence MW did not change, and 
ZIP codes where it did.
ZIP codes in the first group see rent increases that are moderated by the 
negative effect of the residence MW.
The median incidence on landlords for this group is $\rhoMedCentsDirFedNine$ 
cents of each dollar.
Locations in the second group are only affected through changes in the 
workplace MW, so median incidence for this group is larger at 
$\rhoMedCentsIndirFedNine$ cents of each dollar.
The bottom row of Panel A in Table \ref{tab:counterfactuals} shows our estimate
of total incidence of the policy, which is given by $\totIncidenceFedNine$.
The share is lower than the median values reported earlier because landlords 
capture more in locations with lower rent increases.

More generally, one can think of the average share for different values of the 
gap between the residence MW and the workplace MW, i.e., 
$\Delta \mw_i^{\wkp} - \Delta \mw_i^{\res}$.
Figure \ref{fig:rho_by_decile_MW_gap} displays the average estimated share for 
each decile of that gap.
We observe a positive and nearly monotonic relation.
The share is lower in ZIP codes that had a low increase in the workplace MW 
relative to the residence MW, highlighting how the share pocketed depends on
the incidence of the federal MW increase on the MW measures.

\paragraph{Local increase.}

Panel B of Figure \ref{fig:cf_hist_shares} shows the distribution of the 
estimated shares in the Chicago-Naperville-Elgin CBSA.
Panel B of Table \ref{tab:counterfactuals} displays median values for ZIP codes
inside the city and outside it.
The incidence on landlords is of $\rhoMedCentsDirChiFourteen$ cents of each 
dollar for the median directly treated ZIP code and of 
$\rhoMedCentsIndirChiFourteen$ cents for the median not directly treated one.

Panel B of Figure \ref{fig:map_chicago_cf_shares} maps the shares.
Panel B of Online Appendix Figure \ref{fig:map_chicago_cf_rents_wages} shows the 
estimated changes in rents and total wages.
Unlike the previous exercise, the share pocketed by landlords is now higher 
right outside of Chicago City.
Many commuters reside there, and thus the workplace MW changes the most.
This translates into higher rent increases, implying a large share pocketed.%
\footnote{It is worth emphasizing that we estimate large increases in wage income
inside the city due to the fact that our model in \eqref{eq:cf_wages_model}
excludes heterogeneity based on the share of MW workers.
In a setting where this equation accounts for the share of MW workers we would 
not expect a strong effect on wages inside the city.} 

\paragraph{Sensitivity to $\varepsilon$.}

Our estimates of the incidence of the policy depend on the value of the
income elasticity $\varepsilon$.
Online Appendix Figure \ref{fig:cf_share_by_epsilon} displays total share 
pocketed by landlords for different values of $\varepsilon$, both for the 
federal and the local counterfactual policies.
As expected, the share pocketed is decreasing in $\varepsilon$.
For instance, if we assume $\varepsilon = 0.06$ instead of 0.1, then
for the federal policy the share pocketed would be about $0.16$ cents and
for the local policy it would be about $0.19$ cents.

\subsection{Discussion}\label{sec:discussion_cf}

Overall, we observe that landlords capture a significant portion of the income 
generated by MW policies.
We also found strong spatial heterogeneity in incidence depending on 
commuting patterns.
The share pocketed by landlords tends to be larger in ZIP codes located in 
jurisdictions where the MW policy did not change,
particularly those located close to the MW change as many of their residents
work under the new MW level and experience no change in the residence MW.
According to the model in Section \ref{sec:model}, the mechanism behind
this result is the offsetting effect of increases in prices of non-tradable 
consumption in the same location.

Because of the housing market, 
the impact of the MW will be less equalizing in terms of the distribution of
real incomes than nominal incomes.
There are many reasons for this.
First, poorer areas tend to have a higher share of expenditure in housing.
Second, as we discussed in Section \ref{sec:data_income_housing},
low-wage households are more likely to rent.
Finally, in the case of high-income cities enacting MW policies, affected 
low-wage workers are more likely to live outside the city where rent
increases will be larger.

\section{Conclusions}\label{sec:conclusion}

We explore whether minimum wage changes affect housing rental prices, and 
whether MW shocks propagate spatially through commuting.
To answer this question we develop a theoretical approach that accounts for
the fact that MW workers typically reside and work in different locations.
Our model suggests that MW changes at workplaces will tend to increase
rents, and highlights the importance of accounting for the MW at the residence
location when estimating the effect of the workplace MW on rents.

We collect data on rents, statutory MW levels, and commuting flows, and estimate 
the effect of the residence and workplace MW on rents.
We find evidence supporting the main conclusions of our model: the workplace MW
increases rents, and thus MW policies spill over spatially through commuting.
Our conclusions are robust to a variety of robustness checks, and suggest
stronger effects in locations that are residence to more MW workers.
Our two-parameter model is able to capture rich heterogeneity in the effect 
of the MW on rents depending on the prevailing commuting structure.

To explore the incidence of the MW on landlords, we explore two counterfactual 
MW policies.
Our results suggest that landlords pocket a non-negligible portion of the newly 
generated wage income, and that this share varies spatially.
Because low-wage households are more affected by MW policies and more likely to 
be renters,
the omission of the housing market channel would lead to an overstatement of the 
equalizing effects of the MW on disposable income.

Our analysis takes a partial equilibrium perspective, exploring the incidence 
of small increases in the MW within metropolitan areas.
However, one would expect general equilibrium adjustments to large changes 
in MW levels, such as worker mobility and changes in housing supply.
Exploring these issues in the context of a spatial model with worker mobility
that distinguishes between renters and homeowners appears as a fruitful 
avenue for future work.


\clearpage
\printbibliography

@article{Dube2019,
  title={Impacts of Minimum Wages: Review of the International Evidence},
  author={Dube, Arindrajit},
  journal={Independent Report to the UK Government},
  year={2019},
  howpublished = {\url{https://www.gov.uk/government/publications/impacts-of-minimum-wages-review-of-the-international-evidence}},
}

@article{Macurdy2015,
  title={How Effective Is the Minimum Wage at Supporting the Poor?},
  author={MaCurdy, Thomas},
  journal={Journal of Political Economy},
  volume={123},
  number={2},
  pages={497--545},
  year={2015},
  publisher={University of Chicago Press Chicago, IL}
}

@article{HarasztosiLidner2019,
  title={Who Pays for the Minimum Wage?},
  author={Harasztosi, P{\'e}ter and Lindner, Attila},
  journal={American Economic Review},
  volume={109},
  number={8},
  pages={2693--2727},
  year={2019}
}

@techreport{NeumarkShirley2021,
  title={Myth or Measurement: What Does the New Minimum Wage Research Say About Minimum Wages and Job Loss in the United States?},
  author={Neumark, David and Shirley, Peter},
  year={2021},
  institution={National Bureau of Economic Research},
  type = {NBER Working Papers},
  number = {28388}
}

@article{Lee1999,
  title={Wage Inequality in the United States During the 1980s: Rising Dispersion or Falling Minimum Wage?},
  author={Lee, David S.},
  journal={Quarterly Journal of Economics},
  volume={114},
  number={3},
  pages={977--1023},
  year={1999},
  publisher={MIT Press}
}

@article{CardKrueger1994,
  title={Minimum Wages and Employment: A Case Study of the Fast-Food Industry in New Jersey and Pennsylvania},
  author={Card, David and Krueger, Alan B.},
  journal={American Economic Review},
  volume={84},
  number={4},
  pages={772--93},
  year={1994}
}

@article{Dube2019Income,
  title={Minimum Wages and the Distribution of Family Incomes},
  author={Dube, Arindrajit},
  journal={American Economic Journal: Applied Economics},
  volume={11},
  number={4},
  pages={268--304},
  year={2019}
}

@techreport{WiltshireEtAl2023,
  title={High Minimum Wages and the Monopsony Puzzle},
  author={Wiltshire, Justin C. and McPherson, Carl and Reich, Michael},
  year={2023},
  institution={Institute for Research on Labor and Employment},
  type={Working Paper},
}

@article{AutorEtAl2016,
  title={The Contribution of the Minimum Wage to US Wage Inequality Over Three Decades: a Reassessment},
  author={Autor, David and Manning, Alan and Smith, Christopher L.},
  journal={American Economic Journal: Applied Economics},
  volume={8},
  number={1},
  pages={58--99},
  year={2016}
}

@article{CegnizEtAl2019,
  title={The Effect of Minimum Wages on Low-Wage Jobs},
  author={Cengiz, Doruk and Dube, Arindrajit and Lindner, Attila and Zipperer, Ben},
  journal={Quarterly Journal of Economics},
  volume={134},
  number={3},
  pages={1405--1454},
  year={2019},
  publisher={Oxford University Press}
}

@techreport{AhlfeldtEtAl2022,
  title={Optimal minimum wages},
  author={Ahlfeldt, Gabriel M and Roth, Duncan and Seidel, Tobias},
  year={2022},
  institution={Center for Economic and Policy Research},
  type={CEPR Discussion Paper},
  number={DP16913}
}

@techreport{BergerHerkenhoffMongey2022,
  title={Minimum wages, efficiency and welfare},
  author={Berger, David W and Herkenhoff, Kyle F and Mongey, Simon},
  year={2022},
  institution={National Bureau of Economic Research},
  type = {NBER Working Papers},
  number={29662}
}

@article{VaghulZipperer2016,
  title={Historical State and Sub-State Minimum Wage Data},
  author={Vaghul, Kavya and Zipperer, Ben},
  journal={Washington Center for Equitable Growth Working Paper},
  volume={90716},
  year={2016}
}

@article{JardimEtAl2022seattle,
    Author = {Jardim, Ekaterina and Long, Mark C. and Plotnick, Robert and van Inwegen, Emma and Vigdor, Jacob and Wething, Hilary},
    Title = {Minimum-Wage Increases and Low-Wage Employment: Evidence from Seattle},
    Journal = {American Economic Journal: Economic Policy},
    Volume = {14},
    Number = {2},
    Year = {2022},
    Month = {May},
    Pages = {263-314},
    DOI = {10.1257/pol.20180578},
    URL = {https://www.aeaweb.org/articles?id=10.1257/pol.20180578}
}

@article{DubeLindner2021,
  title={City Limits: What Do Local-Area Minimum Wages Do?},
  author={Dube, Arindrajit and Lindner, Attila},
  journal={Journal of Economic Perspectives},
  volume={35},
  number={1},
  pages={27--50},
  year={2021}
}

@article{DracaMachinVanreenen2011,
  title={Minimum Wages and Firm Profitability},
  author={Draca, Mirko and Machin, Stephen and Van Reenen, John},
  journal={American Economic Journal: Applied Economics},
  volume={3},
  number={1},
  pages={129--51},
  year={2011}
}

@techreport{JardimEtAl2022discontinuity,
  title={Boundary Discontinuity Methods and Policy Spillovers},
  author={Jardim, Ekaterina S and Long, Mark C and Plotnick, Robert and van Inwegen, Emma and Vigdor, Jacob L and Wething, Hilary},
  year={2022},
  institution={National Bureau of Economic Research},
  type = {NBER Working Papers},
  number={30075}
}

@article{Cadena2014,
  title={Recent Immigrants as Labor Market Arbitrageurs: Evidence From the Minimum Wage},
  author={Cadena, Brian C.},
  journal={Journal of Urban Economics},
  volume={80},
  pages={1--12},
  year={2014},
  publisher={Elsevier}
}

@article{Hughes2020,
  title={Housing Demand and Affordability for Low-Wage Households: Evidence from Minimum Wage Changes},
  author={Hughes, Samuel},
  journal={Available at SSRN 3541813},
  year={2020}
}

@article{PerezPerez2021,
  title={City Minimum Wages and Spatial Equilibrium Effects},
  author={P{\'e}rez P{\'e}rez, Jorge},
  journal={SocArXiv preprint}, 
  URL={https://osf.io/preprints/socarxiv/fpx9e/},
  year={2021},
  doi={10.31235/osf.io/fpx9e}
}

@article{Tidemann2018,
  title={Minimum Wages, Spatial Equilibrium, and Housing Rents},
  author={Tidemann, Krieg},
  journal={Job Market Paper},
  year={2018}
}

@article{AgarwalEtAl2022,
  title={Minimum Wage Increases and Eviction Risk},
  author={Agarwal, Sumit and Ambrose, Brent W. and Diop, Moussa},
  journal={Journal of Urban Economics},
  pages={103421},
  year={2022},
  publisher={Elsevier}
}

@article{Yamagishi2021,
  title={Minimum Wages and Housing Rents: Theory and Evidence},
  author={Yamagishi, Atsushi},
  journal={Regional Science and Urban Economics},
  volume={87},
  pages={103649},
  year={2021},
  publisher={Elsevier}
}

@article{Monras2019,
  title={Minimum Wages and Spatial Equilibrium: Theory and Evidence},
  author={Monras, Joan},
  journal={Journal of Labor Economics},
  volume={37},
  number={3},
  pages={853--904},
  year={2019},
  publisher={University of Chicago Press Chicago, IL}
}

@article{Leung2021,
  title={Minimum Wage and Real Wage Inequality: Evidence From Pass-Through to Retail Prices},
  author={Leung, Justin H.},
  journal={Review of Economics and Statistics},
  volume={103},
  number={4},
  pages={754--769},
  year={2021}
}

@article{Aaronson2001,
  title={Price Pass-Through and the Minimum Wage},
  author={Aaronson, Daniel},
  journal={Review of Economics and Statistics},
  volume={83},
  number={1},
  pages={158--169},
  year={2001},
  publisher={MIT Press}
}

@article{AllegrettoReich2018,
  title={Are Local Minimum Wages Absorbed by Price Increases? Estimates From Internet-Based Restaurant Menus},
  author={Allegretto, Sylvia and Reich, Michael},
  journal={ILR Review},
  volume={71},
  number={1},
  pages={35--63},
  year={2018},
  publisher={SAGE Publications Sage CA: Los Angeles, CA}
}

@article{RenkinEtAl2020,
  title={The Pass-Through of Minimum Wages Into US Retail Prices: Evidence From Supermarket Scanner Data},
  author={Renkin, Tobias and Montialoux, Claire and Siegenthaler, Michael},
  journal={Review of Economics and Statistics},
  pages={1--99},
  year={2020}
}

@article{KlineMoretti2014,
title = {People, Places, and Public Policy: Some Simple Welfare Economics of Local Economic Development Programs},
author = {Kline, Patrick and Moretti, Enrico},
journal = {Annual Review of Economics},
volume = {6},
number = {1},
pages = {629-662},
year = {2014},
doi = {10.1146/annurev-economics-080213-041024},
}

@article{GiroudMueller2019,
  title={Firms' internal networks and local economic shocks},
  author={Giroud, Xavier and Mueller, Holger M},
  journal={American Economic Review},
  volume={109},
  number={10},
  pages={3617--49},
  year={2019}
}

@techreport{CoutureEtAl2019,
  title={Income Growth and the Distributional Effects of Urban Spatial Sorting},
  author={Couture, Victor and Gaubert, Cecile and Handbury, Jessie and Hurst, Erik},
  year={2019},
  institution={National Bureau of Economic Research},
  type={NBER Working Papers},
  number={26142}
}

@article{AllenEtAl2020,
  title={Is Tourism Good for Locals? Evidence from Barcelona},
  author={Allen, Treb and Fuchs, Simon and Ganapati, Sharat and Graziano, Alberto and Madera, Rocio and Montoriol-Garriga, Judit},
  journal={Unpublished manuscript},
  year={2020}
}

@techreport{MiyauchiEtAl2021,
  title={Consumption Access and Agglomeration: Evidence from Smartphone Data},
  author={Miyauchi, Yuhei and Nakajima, Kentaro and Redding, Stephen J.},
  year={2021},
  institution={National Bureau of Economic Research},
  type={NBER Working Papers},
  number={28497}
}

@article{AngristImbens1995,
  title={Two-stage least squares estimation of average causal effects in models with variable treatment intensity},
  author={Angrist, Joshua D and Imbens, Guido W},
  journal={Journal of the American Statistical Association},
  volume={90},
  number={430},
  pages={431--442},
  year={1995},
  publisher={Taylor \& Francis}
}

@article{ArellanoBond1991,
  title={Some Tests of Specification for Panel Data: Monte Carlo Evidence and an Application to Employment Equations},
  author={Arellano, Manuel and Bond, Stephen},
  journal={Review of Economic Studies},
  volume={58},
  number={2},
  pages={277--297},
  year={1991},
  publisher={Wiley-Blackwell}
}

@incollection{ArellanoHonore2001,
  title={Panel Data Models: Some Recent Developments},
  author={Arellano, Manuel and Honor{\'e}, Bo},
  booktitle={Handbook of Econometrics},
  volume={5},
  pages={3229--3296},
  year={2001},
  publisher={Elsevier}
}

@article{CallawayEtAl2021,
  title={Difference-in-Differences With a Continuous Treatment},
  author={Callaway, Brantly and Goodman-Bacon, Andrew and Sant'Anna, Pedro HC},
  journal={arXiv preprint arXiv:2107.02637},
  year={2021}
}

@article{Hainmueller2012,
  title={Entropy Balancing for Causal Effects: A Multivariate Reweighting Method to Produce Balanced Samples in Observational Studies},
  author={Hainmueller, Jens},
  journal={Political Analysis},
  pages={25--46},
  year={2012},
  publisher={JSTOR}
}

@article{RothEtAl2022,
  title={What's Trending in Difference-in-Differences? A Synthesis of the Recent Econometrics Literature},
  author={Roth, Jonathan and Sant'Anna, Pedro HC and Bilinski, Alyssa and Poe, John},
  journal={arXiv preprint arXiv:2201.01194},
  year={2022}
}

@techreport{deChaisemartinEtAl2022,
  title={Two-way Fixed Effects and Differences-in-Differences with Heterogeneous Treatment Effects: A Survey},
  author={Cl{\'e}ment {de Chaisemartin} and Xavier {D'Haultfoeuille}},
  year={2022},
  journal = {Econometrics Journal}
}

@article{BorusyakHullJaravel2021,
  author = {Borusyak, Kirill and Hull, Peter and Jaravel, Xavier},
  title = {Quasi-Experimental Shift-Share Research Designs},
  journal = {Review of Economic Studies},
  volume = {89},
  number = {1},
  pages = {181-213},
  year = {2021}
}

@article{GoldsmithpinkhamEtAl2020,
  author = {Goldsmith-Pinkham, Paul and Sorkin, Isaac and Swift, Henry},
  title = {Bartik Instruments: What, When, Why, and How},
  journal = {American Economic Review},
  volume = {110},
  number = {8},
  year = {2020},
  month = {August},
  pages = {2586-2624},
  DOI = {10.1257/aer.20181047},
  URL = {https://www.aeaweb.org/articles?id=10.1257/aer.20181047}
}

@article{Kuehn2016,
  title={Spillover Bias in Cross-Border Minimum Wage Studies: Evidence From a Gravity Model},
  author={Kuehn, Daniel},
  journal={Journal of Labor Research},
  volume={37},
  number={4},
  pages={441--459},
  year={2016},
  publisher={Springer}
}

@article{DelgadoFlorax2015,
  title={Difference-in-Differences Techniques for Spatial Data: Local Autocorrelation and Spatial Interaction},
  author={Delgado, Michael S. and Florax, Raymond JGM},
  journal={Economics Letters},
  volume={137},
  pages={123--126},
  year={2015},
  publisher={Elsevier}
}

@article{Butts2021,
  title={Difference-in-Differences Estimation with Spatial Spillovers},
  author={Butts, Kyle},
  journal={arXiv preprint arXiv:2105.03737},
  year={2021}
}

@article{Fernald2020,
  title={Americas Rental Housing 2020},
  author={Fernald, M.},
  journal={Joint Center for Housing Studies of Harvard University, Cambridge, MA},
  year={2020}
}

@article{AmbroseEtAl2015,
  title={The Repeat Rent Index},
  author={Ambrose, Brent W and Coulson, N Edward and Yoshida, Jiro},
  journal={Review of Economics and Statistics},
  volume={97},
  number={5},
  pages={939--950},
  year={2015},
  publisher={The MIT Press}
}

@misc{ZillowData,
  title = {Zillow Research Data},
  author = {Zillow},
  howpublished = {\url{https://www.zillow.com/research/data/}},
  year={2020},
  note = {Accessed: 2020-02-15}
}

@misc{ZillowDataArchive,
  title = {Zillow Research Data, February 2020 snapshot},
  author = {{Internet Archive}},
  howpublished = {\url{https://web.archive.org/web/20200222220950/https://www.zillow.com/research/data/}},
  year = {2021},
  note = {Accessed: 2022-01-09}
}

@misc{ZillowZORI,
  title = {Methodology: Zillow Observed Rent Index (ZORI)},
  author = {Zillow},
  howpublished = {\url{https://www.zillow.com/research/methodology-zori-repeat-rent-27092/}},
  year={2023},
  note = {Accessed: 2023-08-02}
}

@misc{ZillowTypesOfHomes,
  title = {21 Different Types of Homes},
  author = {{Zillow}},
  howpublished = {\url{https://www.zillow.com/resources/stay-informed/types-of-houses/}},
  year = {2023},
  note = {Accessed: 2023-06-04}
}

@misc{ESRI,
  title = {Shapefile Layer Package USA ZIP Code Areas 2020},
  author = {ESRI},
  howpublished = {\url{https://www.arcgis.com/home/item.html?id=8d2012a2016e484dafaac0451f9aea24}},
  year={2020},
  note = {Accessed: 2021-01-09}
}

@misc{ZillowFacts,
  title = {Zillow Facts and Figures},
  author = {Zillow},
  howpublished = {\url{https://www.zillowgroup.com/facts-figures/}},
  year={2020},
  note = {Accessed: 2020-10-23}
}

@misc{QCEW,
  title = {Quarterly Census of Employment and Wages},
  author = {{US Bureau of Labor Statistics}},
  howpublished = {\url{https://www.bls.gov/cew/downloadable-data-files.htm}},
  year={2020},
  note = {Accessed: 2020-05-10}
}

@misc{IRS,
  title = {SOI Tax Stats - Individual Income Tax Statistics - ZIP Code Data (SOI)},
  author = {{Internal Revenue System}},
  howpublished = {\url{https://www.irs.gov/statistics/soi-tax-stats-individual-income-tax-statistics-zip-code-data-soi}},
  year={2022},
  note = {Accessed: 2022-02-03}
}

@misc{IRSfulltime,
  title = {Identifying Full-time Employees},
  author = {{Internal Revenue System}},
  howpublished = {\url{https://www.irs.gov/affordable-care-act/employers/identifying-full-time-employees}},
  year={2022},
  note = {Accessed: 2022-03-01}
}

@misc{CensusLODES,
  title = {LEHD Origin-Destination Employment Statistics Data (2009-2018) [version 7]},
  author = {{US Census Bureau}},
  howpublished = {US Census Bureau, Longitudinal-Employer Household Dynamics Program. \url{https://lehd.ces.census.gov/data/}},
  year = {2021},
  note = {Accessed: 2021-10-03}
}

@misc{CensusACS,
  title = {American Community Survey 5-year Data 2010-2014},
  author = {{US Census Bureau}},
  howpublished = {\url{https://www.census.gov/data/developers/data-sets/acs-5year.html}},
  year = {2022},
  note = {Accessed: 2022-04-01}
}

@misc{CensusDecennial,
  title = {Decennial Census Data 2010},
  author = {{US Census Bureau}},
  howpublished = {\url{https://www.census.gov/data/developers/data-sets/decennial-census.html}},
  year = {2022},
  note = {Accessed: 2022-03-01}
}

@misc{cbTiger,
  title = {Special Release - Census Blocks with Population and Housing Counts},
  author = {{US Census Bureau}},
  howpublished = {\url{https://www.census.gov/geographies/mapping-files/2010/geo/tiger-line-file.html}},
  year = {2012},
  note = {Accessed: 2022-02-10}
}

@misc{hudPreamble,
  title = {Federal Register},
  author = {{US Department of Housing and Urban Development}},
  volume = {82},
  number = {169},
  howpublished = {\url{https://www.huduser.gov/portal/datasets/fmr/fmr2018/FY2018-FMR-Preamble.pdf}},
  year = {2017}
}

@misc{hudSAFMR,
  title = {Small Area Fair Market Rents},
  author = {{US Department of Housing and Urban Development}},
  howpublished = {\url{https://www.huduser.gov/portal/datasets/fmr.html}},
  year = {2020},
  note = {Accessed: 2020-03-10}
}

@misc{hudCrosswalks,
  title = {USPS ZIP code crosswalk files},
  author = {{US Department of Housing and Urban Development}},
  howpublished = {\url{https://www.huduser.gov/portal/datasets/usps_crosswalk.html}},
  year = {2022},
  note = {Accessed: 2022-02-10}
}

@misc{hudHousing,
  title = {Assisted Housing: National and Local - ZIP code level data based on Census 2010 geographies},
  author = {{US Department of Housing and Urban Development}},
  howpublished = {\url{https://www.huduser.gov/portal/datasets/assthsg.html}},
  year = {2022},
  note = {Accessed: 2022-02-10}
}

@misc{ahs2020,
  title = {American Housing Survey 2011 and 2013},
  author = {{US Department of Housing and Urban Development}},
  year = {2020},
  howpublished = {\url{https://www.census.gov/programs-surveys/ahs.html}},
  note = {Accessed: 2022-06-01}
}

@misc{SafmrReport2018,
  title = {Small Area Fair Market Rent Demonstration Evaluation},
  author = {{US Department of Housing and Urban Development}},
  year = {2018},
  howpublished = {\url{https://www.huduser.gov/portal/sites/default/files/pdf/SAFMR-Evaluation-Final-Report.pdf}},
  note = {Accessed: 2022-06-06}
}

@book{wooldridge2010,
  title={Econometric analysis of cross section and panel data},
  author={Wooldridge, Jeffrey M.},
  year={2010},
  publisher={MIT press}
}

@misc{BerkeleyLaborCenter,
  author = {{UC Berkeley Labor Center}},
  title = {Inventory of US City and County Minimum Wage Ordinances},
  year = {2022},
  howpublished = {\url{https://laborcenter.berkeley.edu/inventory-of-us-city-and-county-minimum-wage-ordinances/}},
  note = {Accessed: 2022-05-30}
}

@misc{MinWorkersReportBLS,
  title={Characteristics of minimum wage workers, 2019},
  author={{US Bureau of Labor Statistics}},
  institution={United States Department of Labor},
  year={2020},
  howpublished = {\url{https://www.bls.gov/opub/reports/minimum-wage/2019/home.htm}},
  note = {Accessed: 2020-11-01}
}


\clearpage
\section*{Figures and Tables}

\begin{figure}[h!]
    \centering
    \caption{Changes in minimum wage measures in the Chicago-Naperville-Elgin CBSA, July 2019}
    \label{fig:map_mw_chicago_jul2019}

    \begin{subfigure}{0.5\textwidth}
        \centering
        \caption*{Residence MW}
        \includegraphics[width = 1\textwidth]
            {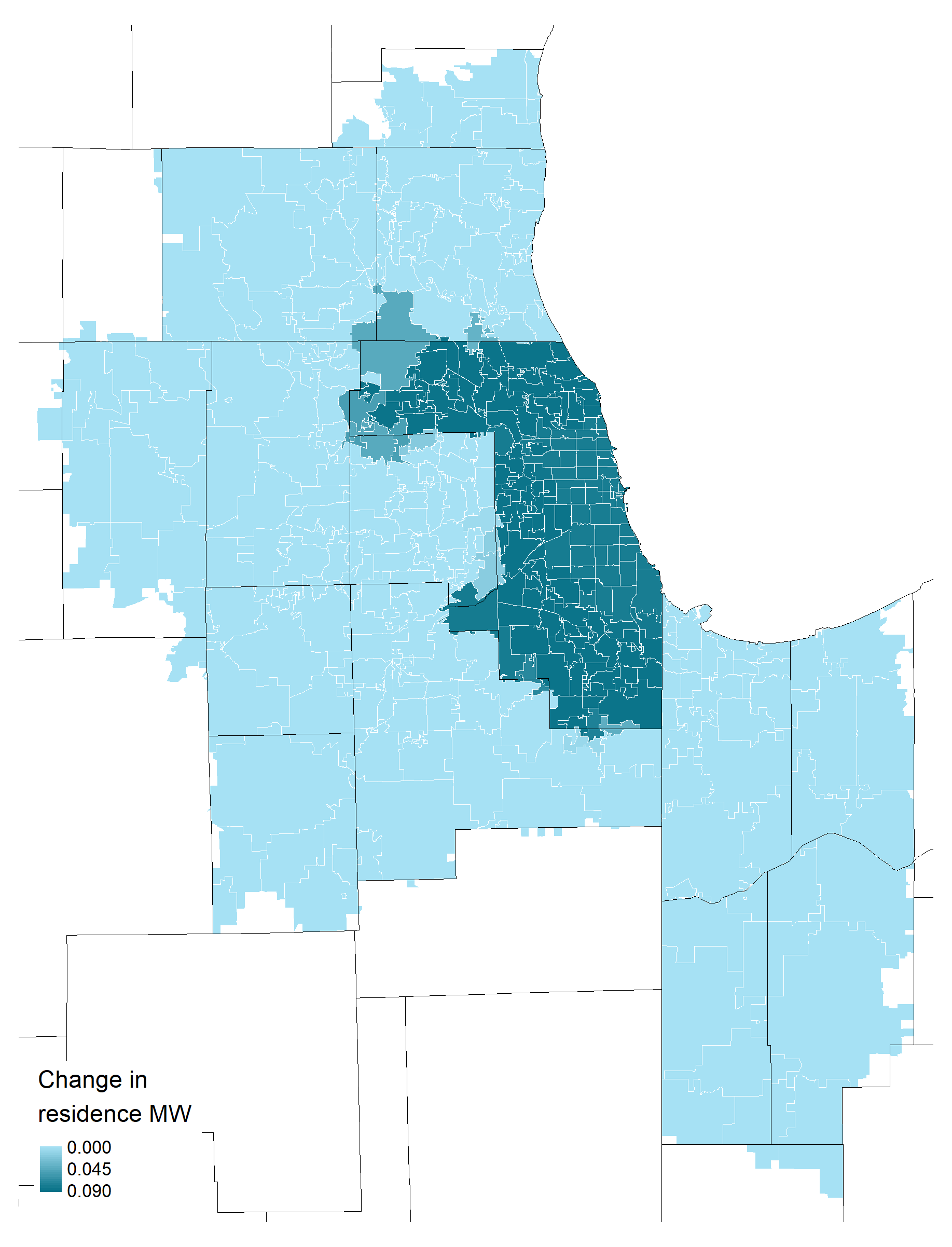}
    \end{subfigure}%
    \begin{subfigure}{0.5\textwidth}
        \centering
        \caption*{Workplace MW}
        \includegraphics[width = 1\textwidth]
            {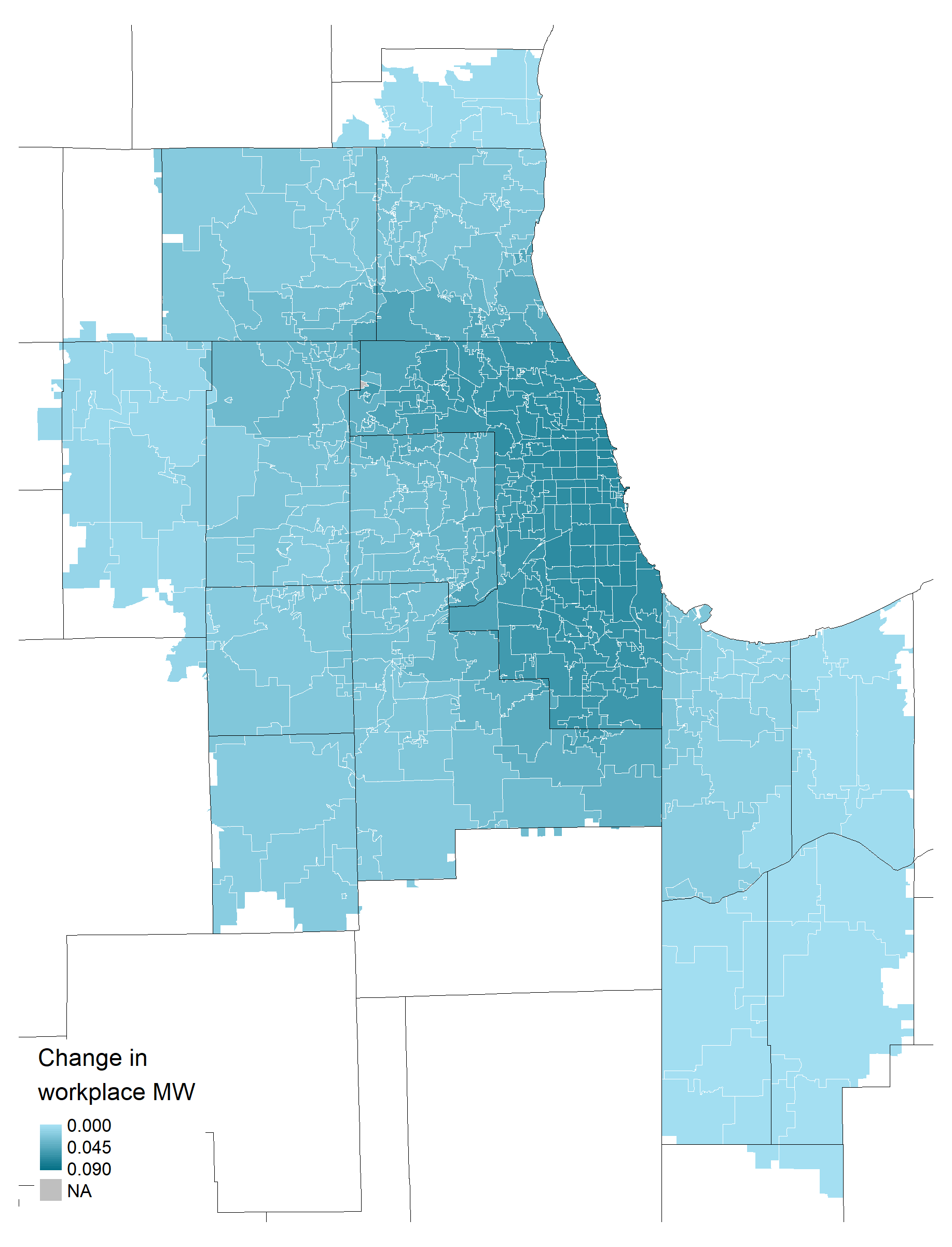}
    \end{subfigure}

    \begin{minipage}{.95\textwidth} \footnotesize
        \vspace{3mm}
        Notes: 
        Data are from the MW panel described in
        Section \ref{sec:data_mw_panel} and from LODES.
        The figures show changes in the MW measures in July 2019 in the 
        metropolitan area of Chicago.
        The figure on the left shows the change in the residence MW.
        The figure on the right shows the change in the workplace MW. 
        The residence MW is defined as the log of the statutory MW of the given
        ZIP code.
        The workplace MW is defined as the weighted average of the log of the
        statutory MW levels in workplace locations of a ZIP code's residents,
        where weights are given by commuting shares.
        Smaller colored polygons correspond to ZIP codes, and larger polygons 
        correspond to counties.
    \end{minipage}
\end{figure}

\clearpage
\begin{figure}[h!]
    \centering
    \caption{Spatial distribution of minimum wage changes between January 2010 
             and June 2020, mainland US}
    \label{fig:map_mw_perc_changes}

    \includegraphics[width = 1\textwidth]
        {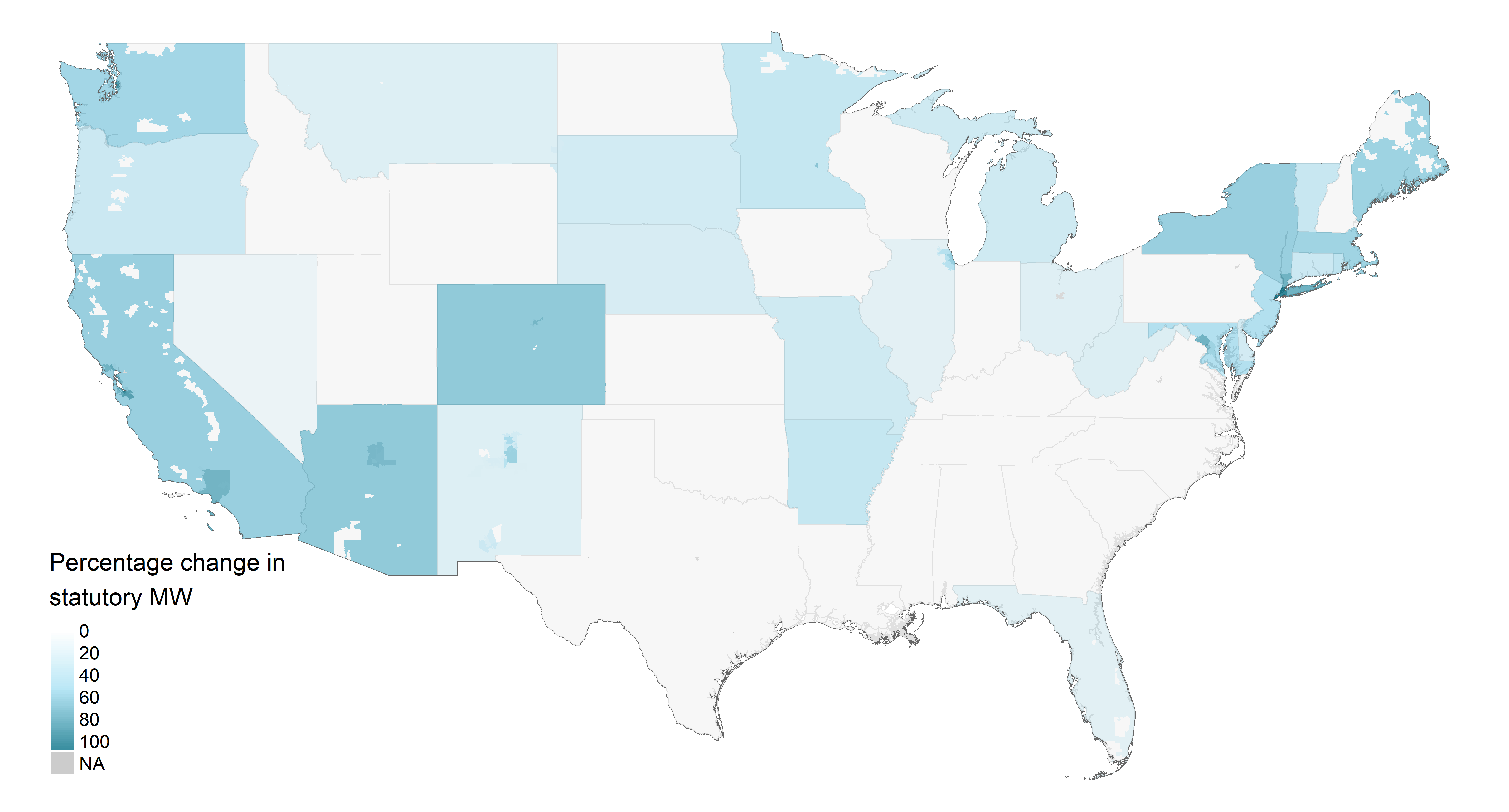}

    \begin{minipage}{.95\textwidth} \footnotesize
        \vspace{3mm}
        Notes: 
        Data are from the MW panel described in
        Section \ref{sec:data_mw_panel}.
        The figure maps the percentage change in the statutory MW
        level in each ZIP code from January 2010 to June 2020.
    \end{minipage}
\end{figure}

\clearpage
\begin{figure}[h!]
    \centering
    \caption{Probability of being a renter by household income decile,
             full sample}
    \label{fig:ahs_pr_renters}

    \includegraphics[width = 0.75\textwidth]{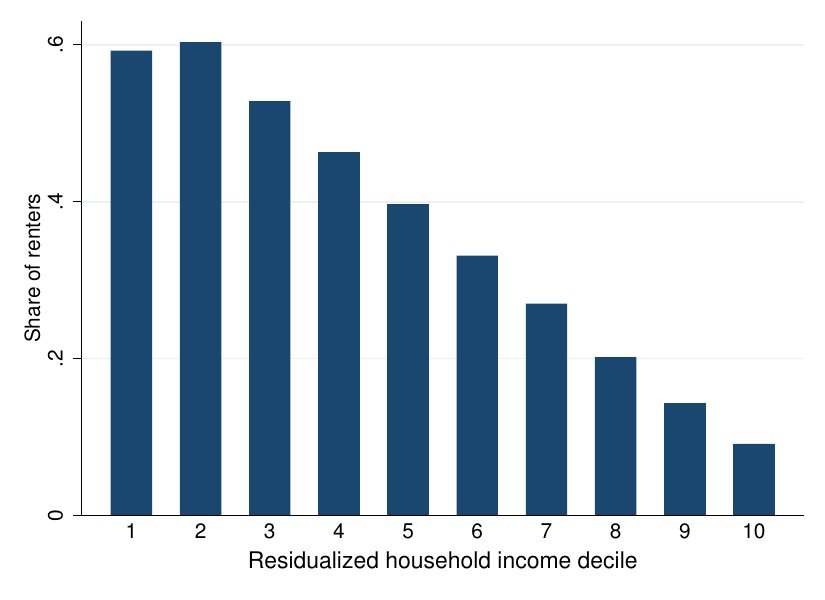}

    \begin{minipage}{.95\textwidth} \footnotesize
        \vspace{3mm}
        Notes: Data are from the 2011 and 2013 American Housing Surveys.
        The figure shows the probability of a household living in a
        rented unit by household income. 
        We construct the figure as follows.
        First, we residualize an indicator for being a renter and 
        household income by SMSA indicators, the closest analogue of CBSAs 
        available in the data.
        Second, we construct deciles of the residualized household
        income variable.
        Finally, we take the average of the residualized renter 
        indicator within each decile.
        We exclude from the calculation non-conventional housing units, 
        such as mobile homes, hotels, and others.
    \end{minipage}
\end{figure}

\clearpage

\begin{figure}[h!]
    \centering
    \caption{Estimates of the effect of the minimum wage on rents, baseline sample
             including leads and lags}
    \label{fig:dynamic_workplace}

    \begin{subfigure}{.65\textwidth}
        \caption{Estimates in first differences}
        \includegraphics[width = 1\textwidth]
            {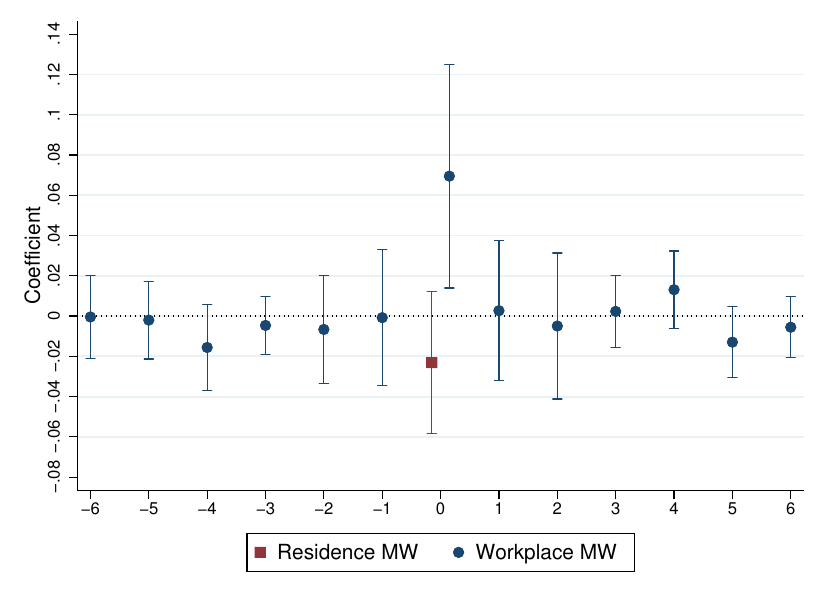}
    \end{subfigure}\\
    \begin{subfigure}{.65\textwidth}
        \caption{Implied path in levels}
        \includegraphics[width = 1\textwidth]
            {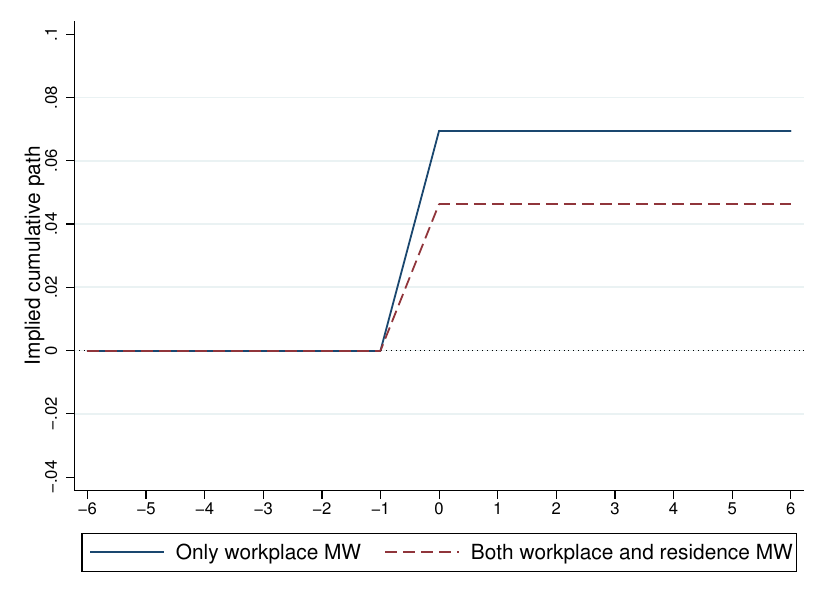}
    \end{subfigure}

    \begin{minipage}{.95\textwidth} \footnotesize
        \vspace{3mm}
        Notes:
        Data are from the baseline estimation sample described in Section 
        \ref{sec:data_final_panel}.
        The top panel shows coefficients from regressions of the change in 
        log of rents per square foot on leads and lags of the change in the 
        workplace MW and the change in the residence MW.
        The bottom panel shows the implied paths in levels given the estimated 
        coefficients assuming pre- and post-coefficients are equal to zero.
        The regression includes time-period fixed effects and economic controls 
        that vary at the county by month and county by quarter levels.
        The measure of rents per square foot correspond to the Single Family, 
        Condominium and Cooperative houses from Zillow.
        The residence MW is defined as the log statutory MW in the same ZIP code.
        The workplace MW is defined as the log statutory MW where the average 
        resident of the ZIP code works, constructed using LODES 
        origin-destination data.
        Economic controls from the QCEW include the change of the following 
        variables: the log of the average wage, the log of employment, and 
        the log of the establishment count for the sectors 
        ``Information,'' ``Financial activities,'' and ``Professional and 
        business services.''
        95\% pointwise confidence intervals are obtained from standard errors 
        clustered at the state level.
    \end{minipage}
\end{figure}

\clearpage
\begin{figure}[h!]
    \centering
    \caption{Estimated shares pocketed by landlords under counterfactual MW policies}
    \label{fig:cf_hist_shares}

    \begin{subfigure}{0.75\textwidth} \centering
        \caption{Increase in federal MW to \$9, urban ZIP codes}
        \includegraphics[width = .85\textwidth]{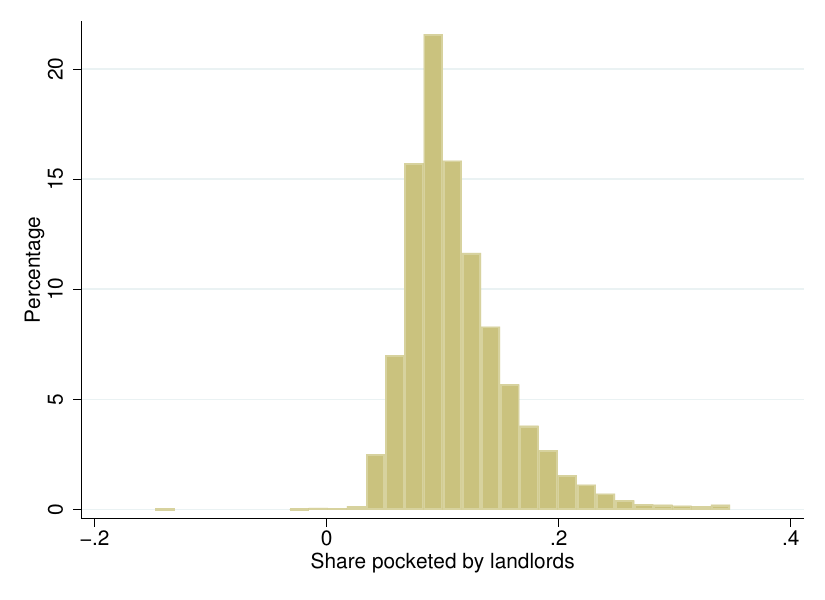}
    \end{subfigure}

    \begin{subfigure}{0.75\textwidth} \centering
        \caption{Increase in Chicago MW to \$14, Chicago-Naperville-Elgin CBSA}
        \includegraphics[width = .85\textwidth]{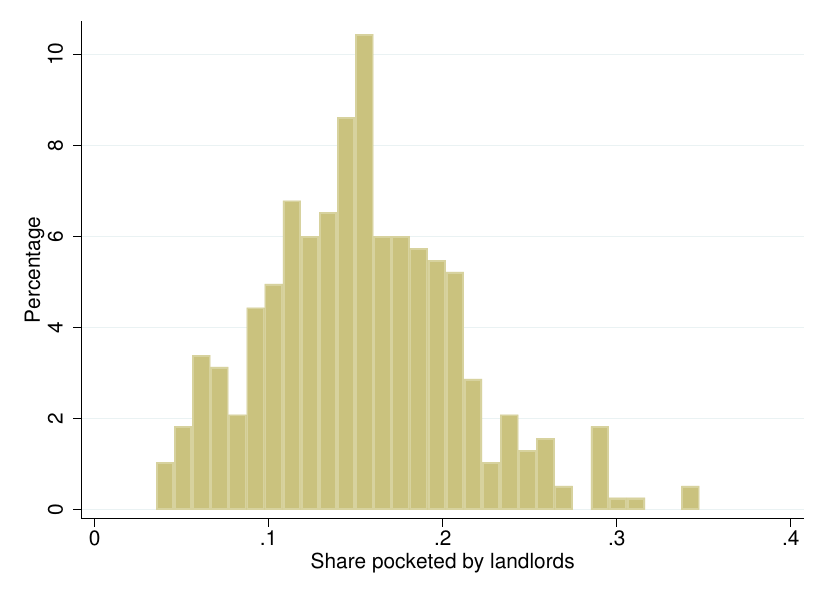}
    \end{subfigure}

    \begin{minipage}{.95\textwidth} \footnotesize
        \vspace{3mm}
        Notes:
        Data are from the MW panel described in section \ref{sec:data_mw_panel} 
        and from LODES.
        The figures show the distribution of the estimated ZIP-code specific
        shares of additional income pocketed by landlords (``share pocketed'')
        under different counterfactual policies.
        Panel A is based on a counterfactual increase to \$9 in the 
        federal MW in January 2020, holding constant other MW policies in their 
        December 2019 levels.
        Panel B is based on a counterfactual increase from \$13 to \$14 in the 
        Chicago City MW, also holding constant other MW policies.
        The unit of observation is the ZIP code.
        Panel A includes ZIP codes located in urban CBSAs where the estimated 
        increase in income was higher than 0.1.
        Panel B includes ZIP codes in the Chicago-Naperville-Elgin CBSA.
        The share pocketed is defined as the ratio between the percent increase 
        in rents and the percent increase in total wages multiplied by the share 
        of housing expenditure in the ZIP code.
        To estimate it we follow the procedure described in Section 
        \ref{sec:counterfactual}, assuming the following parameter values: 
        $\beta = \betaCf$, $\gamma = \gammaCf$, and $\varepsilon = 0.1$.
    \end{minipage}
\end{figure}

\clearpage
\begin{figure}[h!]
    \centering
    \caption{Estimated shares pocketed by landlords under counterfactual MW policies, 
             Chicago-Naperville-Elgin CBSA}
    \label{fig:map_chicago_cf_shares}

    \begin{subfigure}{.5\textwidth}
        \caption{Increase in federal MW\\from \$7.25 to \$9}
        \includegraphics[width = 1\textwidth]
            {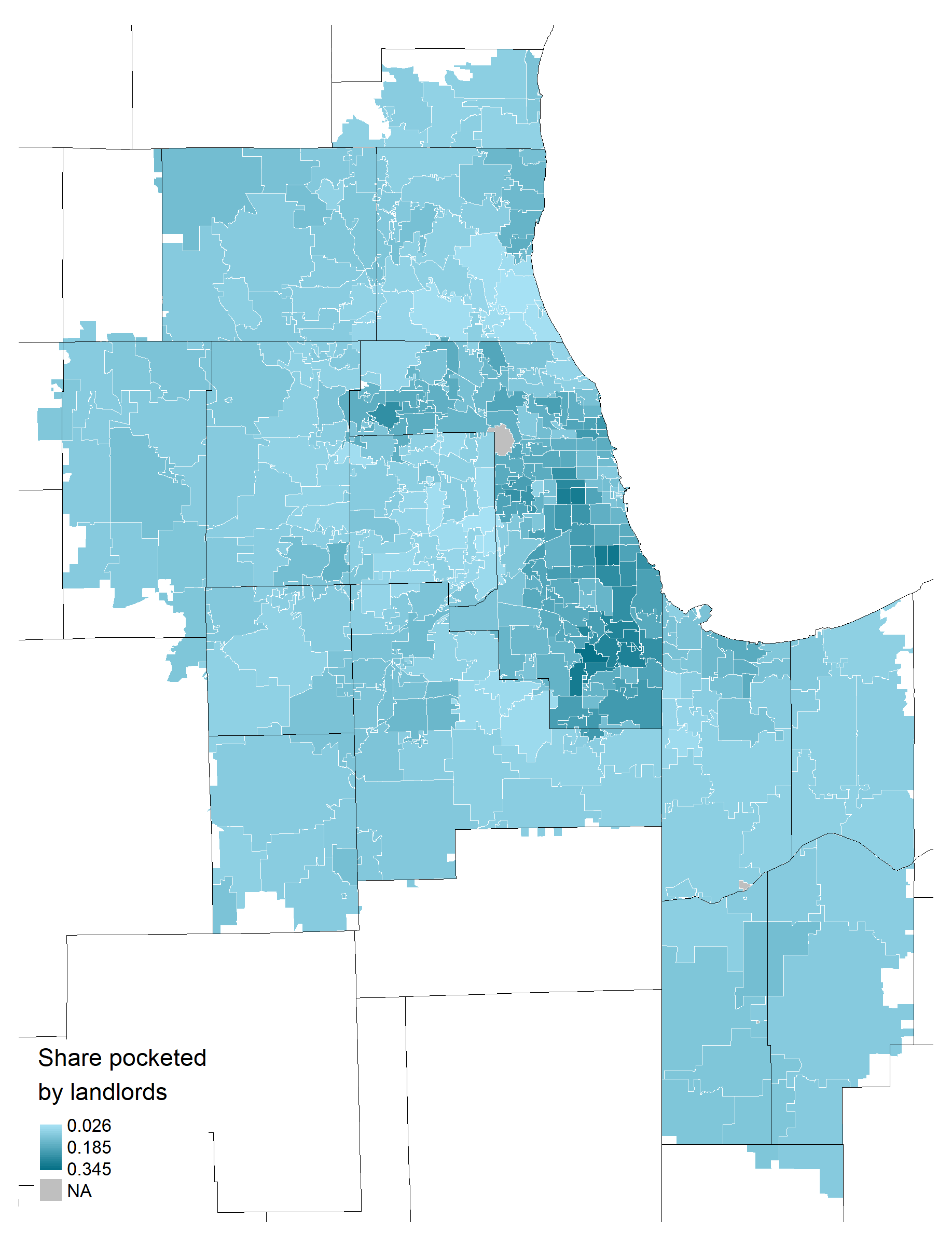}
    \end{subfigure}%
    \begin{subfigure}{.5\textwidth}
        \caption{Increase in Chicago MW\\from \$13 to \$14}
        \includegraphics[width = 1\textwidth]
            {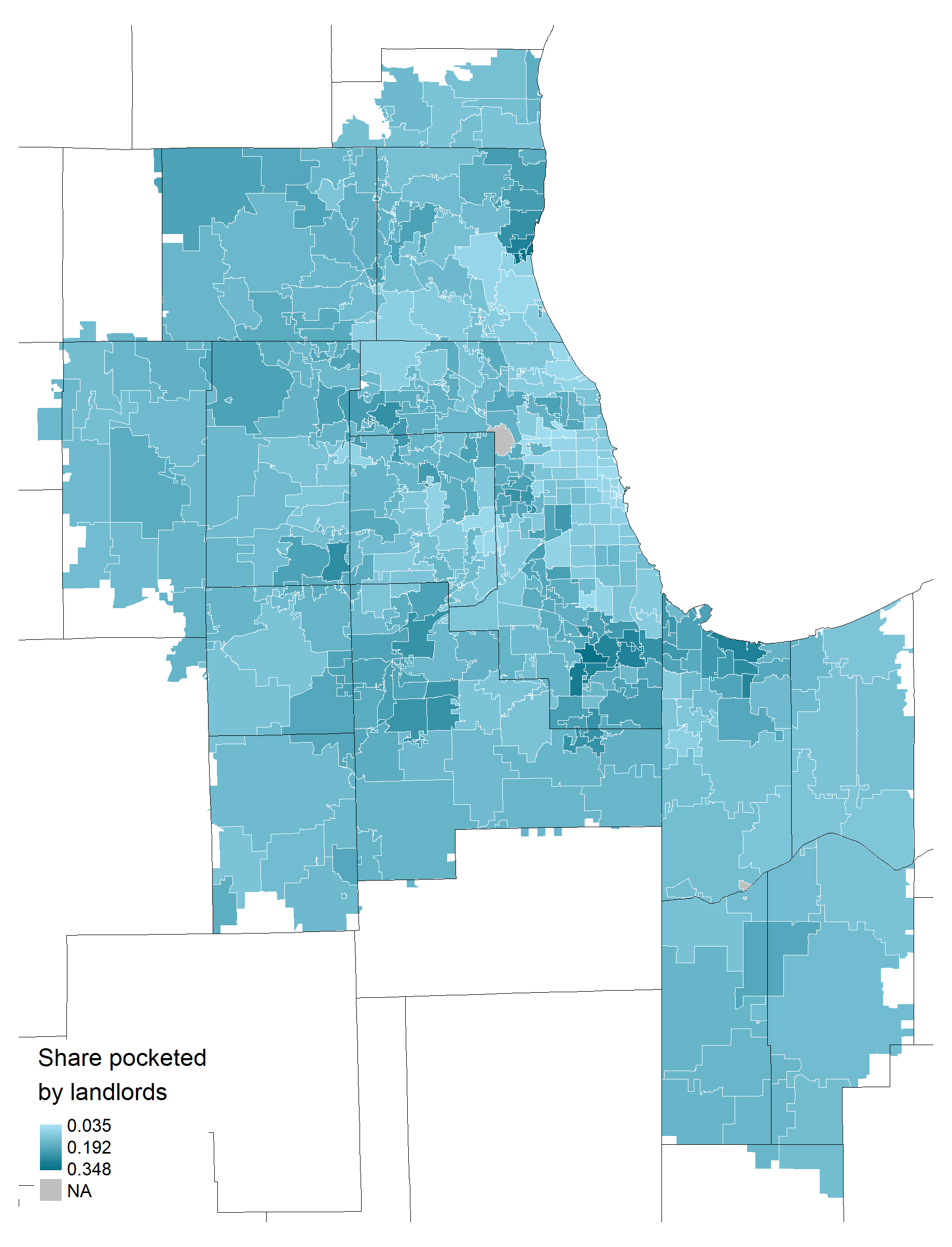}
    \end{subfigure}\\

    \begin{minipage}{.95\textwidth} \footnotesize
        \vspace{3mm}
        Notes: 
        Data are from the MW panel described in Section \ref{sec:data_mw_panel} 
        and from LODES.
        The figures map the estimated ZIP code-specific shares of additional 
        income generated by the MW that are pocketed by landlords, 
        for different counterfactual MW policies.
        Panel A is based on a counterfactual increase from \$7.25 to \$9 in the 
        federal MW in January 2020, holding constant other MW policies in their 
        December 2019 levels.
        Panel B is based on a counterfactual increase from \$13 to \$14 in the 
        Chicago City MW, also holding constant other MW policies.
        The share pocketed is defined as the ratio between the percent increase 
        in rents and the percent increase in total wages multiplied by the share 
        of housing expenditure in the ZIP code.
        To estimate it we follow the procedure described in Section 
        \ref{sec:counterfactual}, assuming the following parameter values: 
        $\beta = \betaCf$, $\gamma = \gammaCf$, and $\varepsilon = 0.1$.
    \end{minipage}
\end{figure}

\clearpage
\begin{figure}[h!]
    \centering
    \caption{Share pocketed by landlords by intensity of treatment, 
             urban ZIP codes under federal MW increase to \$9}
    \label{fig:rho_by_decile_MW_gap}

	\includegraphics[width = 0.75\textwidth]{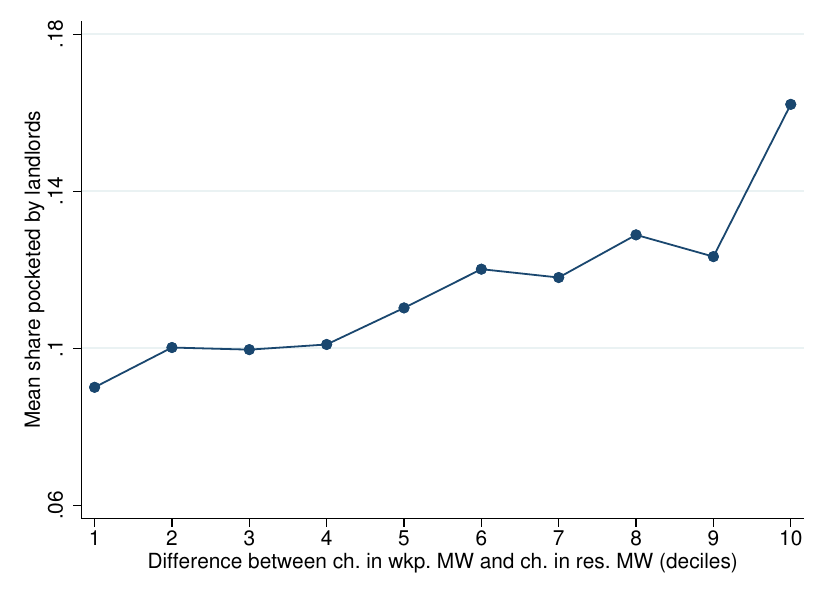}

    \begin{minipage}{.95\textwidth} \footnotesize
        \vspace{3mm}
        Notes:
        Data are from the MW panel described in Section \ref{sec:data_mw_panel} 
        and from LODES.
        The figure shows the average estimate of the shares of additional
        income pocketed by landlords $\rho_i$ for each decile of the 
        difference $\Delta \mw_i^{\wkp} - \Delta \mw_i^{\res}$.
        Estimates for lower deciles correspond to ZIP codes where the increase 
        in residence MW was relatively large.
        The unit of observation is the urban ZIP code, where we define a ZIP code 
        as urban if it belongs to a CBSA with at least 80\% of its population 
        classified as urban by the 2010 Census.
        The share pocketed is defined as the ratio between the percent increase 
        in rents and the percent increase in total wages multiplied by the share 
        of housing expenditure in the ZIP code.
        To estimate it we follow the procedure described in Section 
        \ref{sec:counterfactual}, assuming the following parameter values: 
        $\beta = \betaCf$, $\gamma = \gammaCf$, and $\varepsilon = \epsilonCf$.
        The figure excludes ZIP codes located in the $\cbsaLowIncFedNine$ CBAs for 
        which the average estimated change in log total wages was below 0.1.
    \end{minipage}
\end{figure}

\clearpage

\begin{landscape}
\begin{table}[hbt!] \centering
    \caption{Descriptive statistics of different samples of ZIP codes}
    \label{tab:stats_zip_samples}
    \begin{tabular}{@{}lcccc@{}}
        \toprule
                                                         & \multicolumn{1}{p{2cm}}{\centering All\\ZIP codes}
                                                         & \multicolumn{1}{p{2cm}}{\centering Urban\\ZIP codes}
                                                         & \multicolumn{1}{p{2cm}}{\centering Zillow\\sample}
                                                         & \multicolumn{1}{p{2cm}}{\centering Baseline\\sample}  \\ \midrule
        \textit{Panel A: 2010 Census}                        &       &       &        &             \\
        $\quad$Total population (thousands)                  & 308,129.6 & 204,585.8 & 111,709.2  & 51,181.1     \\
        $\quad$Total number of households (thousands)        & 131,396.0 & 83,919.6 & 47,424.5  & 21,628.7     \\
        $\quad$Mean population                               & 9,681.7 & 18,018.8 & 33,687.9  & 38,052.9     \\
        $\quad$Mean number of households                     & 4,128.6 & 7,391.2 & 14,301.7  & 16,080.8     \\
        $\quad$Share of urban population                     & 0.391    & 0.725   & 0.960   & 0.972          \\
        $\quad$Share of renter-occupied households           & 0.224    & 0.283   & 0.340   & 0.333          \\
        $\quad$Share of black population                     & 0.075    & 0.100   & 0.153   & 0.161          \\
        $\quad$Share of white population                     & 0.834    & 0.765   & 0.679   & 0.667          \\
        \textit{Panel B: 2014 IRS}                           &       &       &        &              \\
        $\quad$Share of households with wage income          & 0.820    & 0.830   & 0.836   & 0.843          \\
        $\quad$Share of households with business income      & 0.152    & 0.161   & 0.176   & 0.182          \\
        $\quad$Mean AGI per household (thousand \$)          & 60.4 & 76.3 & 83.0 & 83.9     \\
        $\quad$Mean wage income per household (thousand \$)  & 39.7 & 49.8 & 53.2 & 55.2     \\
        \textit{Panel C: 2014 SAFMR}                         &       &       &        &               \\
        $\quad$Mean 40th perc.\ 2BR apt.\ rent (\$)          & 936.17   & 1,028.33  & 1,087.42  & 1,131.95          \\
        \textit{Panel D: Minimum wage}                       &       &       &        &              \\
        $\quad$Min.\ in Dec.\ 2014 (\$)                      & 7.25   & 7.25  & 7.25  & 7.25         \\
        $\quad$Mean in Dec.\ 2014 (\$)                       & 7.74   & 7.97  & 7.94  & 7.87         \\
        $\quad$Max.\ in Dec.\ 2014 (\$)                      & 15.00   & 15.00  & 11.27  & 10.74         \\
        $\quad$Min.\ in Dec.\ 2019 (\$)                      & 7.25   & 7.25  & 7.25  & 7.25         \\
        $\quad$Mean in Dec.\ 2019 (\$)                       & 8.85   & 9.52  & 9.40  & 9.23         \\
        $\quad$Max.\ in Feb.\ 2019 (\$)                      & 16.09   & 16.09  & 16.00  & 16.00         \\
        \textit{Panel E: Geographies}                        &       &       &        &               \\
        $\quad$Number of ZIP codes                           & 31,826  & 11,354 & 3,316 & 1,345             \\
        $\quad$Number of counties                            & 3,135  & 605 & 487 & 244             \\
        $\quad$Number of states                              & 51  & 47 & 49 & 41             \\ \bottomrule
    \end{tabular}

    \begin{minipage}{.95\linewidth} \footnotesize
        \vspace{2mm}
        Notes: The table shows characteristics of different samples of ZIP codes.
        The first column uses all ZIP codes that are matched to a census block
        following Online Appendix \ref{sec:blocks_to_uspszip}.
        The second column restricts to ZIP codes located in urban CBSAs, where 
        we define a CBSA as urban if at least 80\% of its population was 
        classified as urban by the 2010 US Census.
        The third column uses ZIP codes with valid SFCC rents per square
        foot in any month.
        The fourth column uses our baseline estimation sample, as described in
        Section \ref{sec:data_final_panel}.
        Panel A uses data from the 2010 US Census \parencite{CensusDecennial}.
        Panel B uses data from the 2014 IRS ZIP code-level aggregates
        \parencite{IRS}. AGI is an acronym for Average Gross Income.
        Panel C uses data from the 2014 Small-Area Fair Market 
        Rents \parencite[SAFMR;][]{hudSAFMR}.
        Panel D uses data from the panel of MW levels
        described in Section \ref{sec:data_mw_panel}.
        Panel E counts the number of different geographies present in each set
        of ZIP codes, assigned as explained in Section \ref{sec:data_other}.
    \end{minipage}
\end{table}
\end{landscape}

\clearpage
\begin{table}[hbt!] \centering
    \caption{Estimates of the effect of the minimum wage on rents, baseline sample}
    \label{tab:static}
    \begin{tabular}{l*{4}{c}}
        \toprule
        & \multicolumn{1}{c}{\shortstack{Change wkp.\\MW $\Delta\mw_{it}^{\wkp}$}}
            & \multicolumn{3}{c}{\shortstack{Change log rents\\$\Delta r_{it}$}} \\ \cmidrule(lr){2-2}\cmidrule(lr){3-5}
                                           & (1)   & (2)   & (3)   & (4)            \\ \midrule
        Change residence MW 
                  $\Delta\mw_{it}^{\res}$  &  0.8627  &  0.0372  &       &  -0.0219     \\
                                           & (0.0374) & (0.0145) &       & (0.0175)    \\
        Change workplace MW 
                   $\Delta\mw_{it}^{\wkp}$ &       &       &  0.0449  & 0.0685      \\
                                           &       &       & (0.0156) & (0.0288)    \\ \midrule
        Sum of coefficients                &       &       &       &  0.0466     \\
                                           &       &       &       & (0.0158)    \\ \midrule
        Economic controls                 &  Yes  & Yes   & Yes   & Yes      \\
        P-value equality                   &       &       &       & 0.0514      \\
        R-squared                          &  0.9444  &  0.0212  &  0.0213  & 0.0213      \\
        Observations                       & 80,241  & 80,241  & 80,241  & 80,241     \\\bottomrule
    \end{tabular}

    \begin{minipage}{.95\textwidth} \footnotesize
        \vspace{2mm}
        Notes:
        Data are from the baseline estimation sample described in Section 
        \ref{sec:data_final_panel}.
        Column (1) shows the results of a regression of the workplace MW measure
        on the residence MW measure.
        Column (2) through (4) show the results of regressions of the log of 
        median rents per square foot on our MW-based measures.
        All regressions include time-period fixed effects and economic controls 
        that vary at the county by month and county by quarter levels.
        The measure of rents per square foot corresponds to the Single Family, 
        Condominium and Cooperative houses from Zillow.
        The residence MW is defined as the log statutory MW in the same ZIP code.
        The workplace MW is defined as the statutory MW where the average 
        resident of the ZIP code works, constructed using LODES 
        origin-destination data.
        Economic controls from the QCEW include the log of the average wage, 
        the log of employment, and the log of the establishment count from the 
        sectors ``Information'', ``Financial activities'', and ``Professional
        and business services''.
        Standard errors in parentheses are clustered at the state level.
    \end{minipage}
\end{table}

\clearpage
\begin{landscape}
\begin{table}[ht!]
    \centering
    \caption{Robustness of estimates of the effect of the minimum wage on rents,
             baseline sample}
    \label{tab:robustness}
        
    \begin{tabular}{@{}lccccc@{}}
        \toprule
                                                         & \multicolumn{1}{c}{\shortstack{Change wkp.\ MW\\$\Delta\mw_{it}^{\wkp}$}} 
                                                         & \multicolumn{3}{c}{\shortstack{Change log rents\\$\Delta r_{it}$}} 
                                                         &                                                                           \\ \cmidrule(lr){2-2}\cmidrule(lr){3-5}
                                                             & \multicolumn{1}{c}{\shortstack{Change res. MW\\$\Delta\mw_{it}^{\res}$}}
                                                             & \multicolumn{1}{c}{\shortstack{Change res. MW\\$\Delta\mw_{it}^{\res}$}}
                                                             & \multicolumn{1}{c}{\shortstack{Change wkp.\ MW\\$\Delta\mw_{it}^{\wkp}$}} 
                                                             & \shortstack{Sum of\\coefficients}
                                                             & N                                                                      \\ \midrule
        $\quad$(a) Baseline                                  &  0.8627  &  -0.0219  &  0.0685  &  0.0466  & 80,241 \\
                                                             & (0.0374) & (0.0175) & (0.0288) & (0.0158) &      \\
        \textit{Panel A: Vary specification}                 &       &       &       &       &      \\
        $\quad$(b) No controls                               &  0.8632  &  -0.0200  &  0.0668  &  0.0468  & 80,692 \\
                                                             & (0.0374) & (0.0180) & (0.0291) & (0.0162) &      \\
        $\quad$(c) County by time FE                         &  0.2857  &  -0.0606  &  0.1559  &  0.0953  & 75,593 \\
                                                             & (0.0399) & (0.0511) & (0.1116) & (0.0811) &      \\
        $\quad$(d) CBSA by time FE                           &  0.5081  &  -0.0358  &  0.0944  &  0.0587  & 78,293 \\
                                                             & (0.0387) & (0.0295) & (0.0610) & (0.0343) &      \\
        $\quad$(e) State by time FE                          &  0.5405  &  0.0142  &  -0.0076  &  0.0066  & 80,393 \\
                                                             & (0.0629) & (0.0239) & (0.0526) & (0.0320) &      \\
        $\quad$(f) ZIP code-specific linear trend            &  0.8596  &  -0.0217  &  0.0711  &  0.0494  & 80,241 \\
                                                             & (0.0390) & (0.0167) & (0.0264) & (0.0132) &      \\
        \textit{Panel B: Vary workplace MW measure}          &       &       &       &       &      \\
        $\quad$(g) 2014 commuting shares                     &  0.8625  &  -0.0199  &  0.0662  &  0.0463  & 80,241 \\
                                                             & (0.0377) & (0.0193) & (0.0299) & (0.0158) &      \\
        $\quad$(h) 2018 commuting shares                     &  0.8626  &  -0.0217  &  0.0683  &  0.0466  & 80,241 \\
                                                             & (0.0372) & (0.0177) & (0.0292) & (0.0159) &      \\
        $\quad$(i) Time-varying commuting shares             &  0.8806  &  -0.0292  &  0.0792  &  0.0500  & 64,236 \\
                                                             & (0.0372) & (0.0207) & (0.0309) & (0.0166) &      \\
        $\quad$(j) 2017 commuting shares, low-income workers &  0.8566  &  -0.0348  &  0.0841  &  0.0493  & 80,241 \\
                                                             & (0.0371) & (0.0221) & (0.0341) & (0.0160) &      \\
        $\quad$(k) 2017 commuting shares, young workers      &  0.8569  &  -0.0332  &  0.0822  &  0.0490  & 80,241 \\
                                                             & (0.0390) & (0.0180) & (0.0294) & (0.0156) &      \\ \bottomrule
    \end{tabular}

    \begin{minipage}{.95\linewidth} \footnotesize
        \vspace{2mm}
        Notes: 
        Data are from the baseline estimation sample described in Section 
        \ref{sec:data_final_panel}.
        Each row of the table shows two estimations on the same sample of ZIP 
        codes and months.
        The first column shows the results of a regression of the change in the 
        workplace MW on the change in the residence MW.
        The second through fourth columns show the results of a regression of 
        the change in log rents on the change in the residence MW and the 
        workplace MW, with the fifth column showing the sum of the coefficients 
        on the MW measures.
        The rents variable corresponds to the median rent per square foot in
        the SFCC category in Zillow.
        Row (a) repeats the results of Table \ref{tab:static}, including fixed
        effects for each year month and economic controls from the QCEW.
        Specifications in Panel A vary the set of fixed effects included in the 
        regression relative to row (a).
        Row (f) includes ZIP code fixed effects in the first-differenced model,
        which in the level model can be interpreted as a ZIP-code specific 
        linear trend.
        Specifications in Panel B vary the commuting shares used to construct 
        the workplace MW measure relative to row (a).
        Row (i) uses data from 2015 to 2018 only.
        Standard errors in parentheses are clustered at the state level.
    \end{minipage}
\end{table}
\end{landscape}

\clearpage
\begin{table}[hbt!] \centering
    \caption{Heterogeneity of estimates of the effect of the minimum wage on rents, 
             baseline sample}
    \label{tab:heterogeneity}
    \begin{tabular}{@{}lcccc@{}}
        \toprule
            & \multicolumn{4}{c}{Change log rents $\Delta r_{it}$}                                                \\ \cmidrule(l){2-5} 
            & (1) & (2) & (3) & (4)                                                               \\ \midrule
        Change res.\ MW $\Delta\mw_{it}^{\res}$                &  -0.0219   &  -0.0483  &  -0.0377   &  -0.0158   \\
                                                               & (0.0175)  & (0.0196) & (0.0254)  & (0.0150)  \\
        Change res.\ MW $\times$ Std.\ share of MW workers     &        &  -0.0801  &        &        \\
                                                               &        & (0.0392) &        &        \\
        Change res.\ MW $\times$ Std.\ median household income &        &       &  0.0506   &        \\
                                                               &        &       & (0.0266)  &        \\
        Change res.\ MW $\times$ Std.\ share of public housing &        &       &        &  -0.0336   \\
                                                               &        &       &        & (0.0330)  \\
        Change wkp.\ MW $\Delta\mw_{it}^{\wkp}$                &  0.0685   &  0.0969  &  0.0862   &  0.0645   \\
                                                               & (0.0288)  & (0.0300) & (0.0356)  & (0.0258)  \\
        Change wkp.\ MW $\times$ Std.\ share of MW workers     &        &  0.0841  &        &        \\
                                                               &        & (0.0445) &        &        \\
        Change wkp.\ MW $\times$ Std.\ median household income &        &       &  -0.0608   &        \\
                                                               &        &       & (0.0344)  &        \\
        Change wkp.\ MW $\times$ Std.\ share of public housing &        &       &        &  0.0306   \\
                                                               &        &       &        & (0.0374)  \\ \midrule
        Mean heterogeneity variable                            &        & 0.1497   &  60,457  & 0.0044    \\
        Std.\ dev heterogeneity variable                       &        & 0.0468   &  22,923  & 0.0173    \\ \midrule
        R-squared                                              &  0.0213   &  0.0212  &  0.0211   &  0.0213   \\
        Observations                                           &  80,241  &  75,329 &  77,197  &  79,701  \\ \bottomrule
    \end{tabular}

    \begin{minipage}{.95\linewidth} \footnotesize
        \vspace{2mm}
        Notes: 
        Data are from the baseline estimation sample described in Section 
        \ref{sec:data_final_panel}.
        In all columns we report the results of regressions of the log of median 
        rents per square foot on our MW-based measures.
        Column (1) reproduces estimates our baseline results from Table 
        \ref{tab:static}.
        In column (2) the changes in residence and workplace MW levels are 
        interacted with the standardized share of MW workers residing in 
        the ZIP code, estimated as in Online Appendix 
        \ref{sec:assigning_mw_levels}.
        In column (3) they are interacted with standardized median household 
        income from the ACS \parencite{CensusACS}.
        In column (4) they are interacted with the standardized share of 
        public housing units.
        To construct this share we use total units of public housing in 2017 
        \parencite{hudHousing}, and the number of households in the 2010
        US Census \parencite{CensusDecennial}.
        Standard errors in parentheses are clustered at the state level.
    \end{minipage}
\end{table}

\clearpage
\begin{table}[hbt!]
    \centering
    \caption{Median effect of counterfactual minimum wage policies by treatment status}
    \label{tab:counterfactuals}

    \begin{tabular}{@{}lccccc@{}}
        \multicolumn{6}{c}{Panel A: Increase in federal MW to \$9, urban ZIP codes} \\
        \addlinespace[0.5em]
        \toprule
                         & N & \shortstack{Change in\\res.\ MW}
                             & \shortstack{Change in\\wkp.\ MW}
                             & \shortstack{Share of\\housing exp.}  
                             & \shortstack{Share\\Pocketed}                              \\ \midrule
        Effect in ZIP codes with...          &      &        &        &       &          \\
        $\quad$previous MW $\leq\$9\quad$    & 5,741 & 0.216 & 0.204  & 0.214 & 0.097    \\
        $\quad$previous MW $>\$9\quad$       & 1,043 & 0.000 & 0.013  & 0.232 & 0.159    \\ 
        Total incidence                      & 6,784 &       &        &       & 0.093    \\ \bottomrule
        \addlinespace[1.2em]
        \multicolumn{6}{c}{Panel B: Increase in Chicago MW to \$14, Chicago-Naperville-Elgin CBSA} \\
        \addlinespace[0.5em]
        \toprule
                         & N & \shortstack{Change in\\res.\ MW}
                             & \shortstack{Change in\\wkp.\ MW}
                             & \shortstack{Share of\\housing exp.}  
                             & \shortstack{Share\\Pocketed}                              \\ \midrule
        Effect in ZIP codes with...          &     &        &        &       &           \\
        $\quad$previous MW $\geq\$13\quad$   & 62  & 0.074  & 0.046  & 0.252 &  0.092    \\
        $\quad$previous MW $<\$13\quad$      & 323 &  0.000 & 0.009  & 0.231 & 0.158     \\ 
        Total incidence                      & 385 &        &        &       & 0.112     \\ \bottomrule
    \end{tabular}
    
    \begin{minipage}{.95\textwidth} \footnotesize
        \vspace{3.5mm}
        Notes: 
        Data are from LODES origin-destination statistics, 
        Small Area Fair Market Rents, 
        IRS ZIP code aggregate statistics, and 
        the MW panel described in Section \ref{sec:data_mw_panel}.
        The table shows the median of the estimated ZIP code-specific shares of 
        the additional income pocketed by landlords (``Share pocketed''), 
        defined as the ratio of the increase in income to the increase in rents,
        for different groups of ZIP codes.
        Panel A is based on a counterfactual increase from \$7.25 to \$9 in the 
        federal MW in January 2020, holding constant other MW policies in their 
        December 2019 levels.
        Panel B is based on a counterfactual increase from \$13 to \$14 in the 
        Chicago City MW, also holding constant other MW policies.
        In the last row of each panel, we report the total incidence of the 
        counterfactual policy.
        We also report the median change in residence MW, change in workplace MW,
        and share of ZIP code-specific housing expenditure ``Share of housing 
        exp.'') defined in Online Appendix \ref{sec:measure_housing_expenditure}.
        Increases in income and rents are simulated following the procedure 
        described in Section \ref{sec:counterfactual}.
        We assume the following parameter values: 
        $\beta = \betaCf$, $\gamma = \gammaCf$, and $\varepsilon = 0.1$.
        Panel A includes urban ZIP codes only and excludes ZIP codes located 
        in $\cbsaLowIncFedNine$ CBAs for which the average estimated change in 
        log total wage income was below 0.1.
        Panel B includes all ZIP codes with valid data in the 
        Chicago-Naperville-Elgin CBSA.
    \end{minipage}
\end{table}


\clearpage

\section*{\center{
    Online Appendix for \\
    ``Minimum Wage as a Place-Based Policy: \\
    Evidence from US Housing Rental Markets''}}

\vspace{5mm}

\appendix

\setcounter{page}{1}

\renewcommand\thetable{\arabic{table}} 
\renewcommand\thefigure{\arabic{figure}}
\renewcommand{\tablename}{Online Appendix Table}
\renewcommand{\figurename}{Online Appendix Figure}
\setcounter{table}{0}
\setcounter{figure}{0}


\section{Model Appendix}

\subsection{Proofs}\label{sec:proofs}

\begin{proof}[Proof of Proposition \ref{prop:comparative_statics}]
    Fully differentiate the market clearing condition with respect to $\ln R_i$ 
    and $\ln \MW_z$ for all $z\in\Z(i)$.
    Using \eqref{eq:equilibrium} and appropriate algebraic manipulations, 
    one can show that
    \begin{equation}\label{eq:diff_equilibrium}
        \Big(\eta_i - \sum_z \pi_{iz} \xi^R_{iz} \Big) d \ln R_i
        = 
        \sum_z \pi_{iz} \left(\xi^P_{iz} \epsilon_{i}^P d \ln \MW_i 
                            + \xi^Y_{iz} \epsilon_{iz}^Y d \ln \MW_z \right) ,
    \end{equation}
    where
    $\xi_{iz}^x = \frac{d h_{iz}}{d x_i} \frac{x_i}{\sum_z \pi_{iz} h_{iz}}$ for
    $x\in\{R,P\}$ is the elasticity of the per-capita housing demand with respect
    to $x$ evaluated at the average per-capita demand of ZIP code $i$,
    $\xi_{iz}^Y = \frac{d h_{iz}}{d Y_z} \frac{Y_z}{\sum_z \pi_{iz} h_{iz}}$ 
    represents the analogous elasticity with respect to income $Y$ from each 
    workplace $z$,
    $\epsilon_{i}^P = \frac{d P_i}{d \MW_i} \frac{\MW_i}{P_i}$ and 
    $\epsilon_{iz}^Y = \frac{d Y_z}{d \MW_z} \frac{\MW_z}{Y_z}$ are
    elasticities of prices and income to the MW, and
    $\eta_i = \frac{d S_i}{d R_i} \frac{R_i}{S_i}$ is the elasticity 
    of housing supply in ZIP code $i$.

    For any $z'\in\Z_0\setminus\{i\}$ the partial effect on rents of the policy
    is given by
    $$
    \frac{d\ln R_i}{d\ln\MW_{z'}} 
      = \left(\eta_i - \sum_z \pi_{iz} \xi^R_{iz}\right)^{-1} 
              \pi_{iz'}\xi^Y_{iz}\epsilon_{iz'}^Y.
    $$
    Because $\eta_i\geq0$ and $\xi^R_{iz} < 0$ for all $z\in\Z(i)$, 
    the first factor is positive.
    From Assumptions \ref{assu:mws} and \ref{assu:housing_demand},
    $\epsilon_{iz}^Y\geq0$ and $\xi^Y_{iz}>0$.
    Therefore, the effect is positive if for $z'$ we have 
    $\frac{d Y_{z'}}{d \MW_{z'}}>0$ (or $\epsilon_{iz'}^Y>0$), 
    and the effect is zero otherwise.

    For ZIP code $i$ the partial effect is given by
    $$
    \frac{d\ln R_i}{d\ln\MW_{i}} 
      = \left(\eta_i - \sum_z \pi_{iz} \xi^R_{iz}\right)^{-1} 
        \left(\epsilon_{i}^P \sum_z \pi_{iz}\xi^P_{iz} 
             + \pi_{ii}\xi^Y_{ii}\epsilon_{ii}^Y \right) .
    $$
    By Assumption \ref{assu:mws} we have that $\epsilon_{i}^P>0$ and that 
    $\epsilon_{ii}^Y\geq0$.
    By Assumption \ref{assu:housing_demand} we have that $\xi^Y_{ii}>0$ and that, 
    for all $z\in\Z(i)$, $\xi^P_{iz}<0$.
    Then, the second parenthesis has an ambiguous sign.
    The third statement of the Proposition follows directly.
\end{proof}

\begin{proof}[Proof of Proposition \ref{prop:representation}]
    Under the stated assumptions we can manipulate \eqref{eq:diff_equilibrium} 
    to write
    $$
    d r_{i} = \beta_i d \mw_{i}^{\wkp} + \gamma_i d \mw^{\text{res}}_i
    $$
    where
    $\beta_i = \frac{\xi_{i}^{Y}\epsilon_i^{Y}}
                     {\eta_{i} - \sum_z \pi_{iz} \xi_{iz}^R} 
              > 0$ and
    $\gamma_i = \frac{\sum_{z\in\Z(i)}\pi_{iz}\xi_{iz}^{P}\epsilon_{i}^{P}}
                    {\eta_{i} - \sum_z \pi_{iz} \xi_{iz}^R} 
              < 0$
    are parameters, which signs can be verified using
    Assumptions \ref{assu:mws} and \ref{assu:housing_demand}.
\end{proof}

\subsection{A dynamic supply and demand model}\label{sec:dyn_theory_model}

The geography is represented by a set of ZIP codes $\Z$.
There is an exogenously given distribution of workers with differing 
residence $i$ and workplace $z$ locations across these ZIP codes which,
as in the main body of the paper, we denote by $\{L_{iz}\}_{i,z\in\Z\times\Z}$.

Let $H_{it}$ be the stock of square feet rented in period $t$.
We assume that all contracts last for one year, so that
the stock is composed of contracts starting at different calendar months. 
We impose that $H_{it} \leq S_i$ for all $t$, 
where $S_i$ denotes the total number of available square feet in $i$.

We further decompose $H_{it}$ as follows.
Let $h_{izt} = h_{iz}\left(R_{it}, \MW_{it}, \MW_{zt}\right)$ be the 
per-capita demand of housing of group $(i,z)$ in period $t$,
which depends on the prevailing MW levels at the time of contract sign-up.
We assume that this demand function is decreasing in the residence MW and 
increasing in the workplace MW, just as in Section \ref{sec:model}.
For simplicity, we omitted the mediation channels of prices and income.
Let $\lambda_{it}$ denote the share of $i$'s residents who started their 
contracts in period $t$.%
\footnote{We assume that these shares do not vary by workplace.}
Then, we can write the stock of contracted square feet during period $t$ as
$$
H_{it} = \sum_{\tau = t - 11}^{t} \lambda_{i\tau} \sum_{z\in\Z} L_{iz} 
h_{iz\tau} (r_{i\tau}, \MW_{i\tau}, \MW_{z\tau})
$$
where 
$r_{i\tau}$ represents rents per square foot in period $\tau$, and
by assumption $\sum_{\tau = t - 11}^{t} \lambda_{i\tau} = 1$.
It is convenient to define the stock of contracted square feet excluding the ones 
that were signed 12 months ago:
$$
\tilde H_{it} = \sum_{\tau = t - 10}^{t} \lambda_{i\tau} \sum_{z\in\Z} L_{iz} 
h_{iz\tau} (r_{i\tau}, \MW_{i\tau}, \MW_{z\tau}) .
$$

We assume that all square feet are homogeneous, and so they have the same price 
in the market.

\subsubsection*{Within-period equilibrium}

We assume the following timing: 
(1) At the beginning of period $t$, a share $\lambda_{it}$ of contracts 
expire (the ones that started on $t-12$);
(2) The square feet from expiring contracts are added to the pool of available 
rental space for new renters;
(3) Renters in $t$ and a flow supply of rental space in $t$ determine equilibrium 
rents $R_{it}$.
Next, we develop each of these steps more formally.

At the start of every period $t$, $\lambda_{i,t-12} \sum_z L_{iz} h_{iz,t-12}$ 
square feet become available for rent from each group of workers $(i,z)$.
The square feet available to rent in period $t$ (vacant) are then
$$
\lambda_{i,t-12} \sum_z L_{iz} h_{iz,t-12} + (S_i - H_{i,t-1}) 
       = S_i - \tilde H_{i,t-1}.
$$
Note that this differs from $S_i - H_{i,t-1}$, the non-rented square feet as 
of $t-1$.
We denote by $V_{it}(R_{it}, \lambda_t)$ the supply of housing, increasing in 
$R_{it}$.
A feasibility constraint is that 
\begin{equation}\label{eq:feasibility}
    V_{it}(R_{it}, \lambda_t) \leq S_i - \tilde H_{i,t-1} .
\end{equation}

The flow demand for new rentals in $t$ by those whose contract expired is given 
by
$$
\lambda_{it} \sum_z L_{iz} h_{izt} \left(R_{it}, \MW_{it}, \MW_{zt} \right) .
$$
This demand arises because a share of the ZIP code's contracts expired. 
Those workers go to the market and may desire to rent more square feet given 
changes in their income.

The market in period $t$ clears if
\begin{equation}\label{eq:equilibrium_dynamic}
    \lambda_t \sum_z L_{iz} h_{iz} \left(R_{it}, \MW_{it}, \MW_{zt} \right) = 
    V_{it}(R_{it}, \lambda_t) .
\end{equation}
Given statutory MW levels in $t$, $\{\MW_{it}\}_{i\in\Z}$,
the share of workers looking to rent in period $t$, $\lambda_t$, and 
a number of vacancies that satisfies \eqref{eq:feasibility}, 
equation \eqref{eq:equilibrium_dynamic} determines equilibrium rents in 
period $t$.
Because the properties of housing demand and housing supply are the same as in 
the model in Section \ref{sec:model},
the equilibrium condition \eqref{eq:equilibrium_dynamic} implies an analogue of 
Propositions \ref{prop:comparative_statics} and \ref{prop:representation}.
The results in Section \ref{sec:model} can be extended to a dynamic setting if
the demand and supply functions in $t$ only depend on MW levels in $t$.

\subsection{A static model with flexible commuting shares}\label{sec:model_endogenous_shares}

In this section we use a simplified version of the partial-equilibrium model 
in Section \ref{sec:model}.
We do so to focus on the implications of relaxing the assumption of fixed 
commuting shares.

Assume that housing demand depends directly on the workplace and residence MW,
abstracting away from the mediation channels of prices and income.
Let commuting shares depend on the MW at the respective workplace location.
Then, the housing market equilibrium can be written as
\begin{equation}\label{eq:equilibrium_simplified}
	L_i \sum_{z\in\Z(i)} \pi_{iz} (\MW_z) h_{iz} \left(R_i, \MW_i, \MW_z\right) = S_i(R_i) .
\end{equation}
Following empirical results in \textcite{PerezPerez2021}, we assume that 
an increase in the workplace MW may decrease the share of workers traveling to 
a given destination.

\begin{assu}[Endogenous commuting shares]\label{assu:commuting_shares}
    Commuting shares of location $i$'s residents are given by 
    $\{\pi_{iz}\left(\MW_z\right)\}_{z\in\Z(i)}$, where $d\pi_{iz}/d\MW_z \leq 0$.
\end{assu}

With this assumption we are ready to prove the following result.

\begin{prop}[Representation with endogenous shares]\label{prop:representation_endog_shares}
    Assume that for all ZIP codes $z\in\Z(i)$ we have
    (a) homogeneous elasticity of per-capita housing demand to the MW 
    at workplace locations,
    $\xi^{\wkp}_{iz}=\xi^{\wkp}_{i}$, and
    (b) homogeneous elasticity of commuting shares to the MW,
    $\zeta_{iz}=\zeta_{i}$.
    Then, approximating $\pi_{iz} h_{iz} / \sum_{z'}\pi_{iz'} h_{iz'}\approx\pi_{iz}$,
    we can write
    $$
    d r_{i} = \left(\beta_i + \zeta_i\right) d \mw_{i}^{\wkp} 
            + \gamma_i d \mw^{\text{res}}_i
    $$
    where 
    $r_{i} = \ln R_i$,
    $\mw_{i}^{\wkp} = \sum_{z\in\Z(i)} \pi_{iz} \ln \MW_z$ 
    is ZIP code $i$'s \textbf{workplace MW}, 
    $\mw^{\res}_i = \ln \MW_i$ 
    is ZIP code $i$'s \textbf{residence MW}, and 
    $\beta_i > 0$, $\zeta_i<0$, and $\gamma_i < 0$ are parameters.
\end{prop}

\begin{proof}    
    Fully differentiate the market clearing condition with respect to $\ln R_i$ 
    and $\ln \MW_z$ for all $z\in\Z(i)$.
    Using \eqref{eq:equilibrium_simplified} and appropriate algebraic 
    manipulations, one can show that
    \begin{equation}\label{eq:diff_equilibrium_endog_shares}
        \Big(\eta_i - \sum_z \pi_{iz} \xi^R_{iz} \Big) d \ln R_i
        = 
        \left(\sum_z \pi_{iz} \xi^{\res}_{iz} \right)d \ln \MW_i 
        + \sum_z \pi_{iz} (\xi^{\wkp}_{iz} + \zeta_{iz}) d \ln \MW_z  ,
    \end{equation}
    where we also impose the approximation that 
    $\pi_{iz} h_{iz} / \sum_{z'}\pi_{iz'} h_{iz'}\approx\pi_{iz}$.
    In this expression
    $\xi^R_{iz} = \frac{d h_{iz}}{d R_i} \frac{R_i}{\sum_z \pi_{iz} h_{iz}}$
    is the elasticity of housing demand to rents,
    $\xi^{\res}_{iz} = \frac{d h_{iz}}{d \MW_i} \frac{\MW_i}{\sum_z \pi_{iz} h_{iz}}$
    is the elasticity of housing demand to the MW at $i$, 
    $\xi^{\wkp}_{iz} = \frac{d h_{iz}}{d \MW_z} \frac{\MW_z}{\sum_z \pi_{iz} h_{iz}}$
    is the elasticity of housing demand to the MW at workplace $z$, and
    $\zeta_{iz} = \frac{d \pi_{iz}}{d \MW_z} \frac{\MW_z}{\sum_z \pi_{iz} h_{iz}}$
    is the elasticity of commuting shares to the MW at $z$.
    
    Under the stated assumptions we can manipulate 
    \eqref{eq:diff_equilibrium_endog_shares} to write
    \begin{equation} \label{eq:theory_representation_endog_shares}
        d r_i = (\beta_i + \zeta_i)  d \mw^{\wkp}_i
              + \gamma_i d \mw^{\res}_i
    \end{equation}
    where
    $\beta_i = \frac{\xi^{\wkp}_i}
                    {\eta_{i} - \sum_z \pi_{iz} \xi_{iz}^R}
              > 0$,
    $\zeta_i = \frac{\zeta_i}
                    {\eta_{i} - \sum_z \pi_{iz} \xi_{iz}^R} \leq 0$, and
    $\gamma_i = \frac{\sum_z \pi_{iz} \xi^{\res}_{iz}}
                     {\eta_{i} - \sum_z \pi_{iz} \xi_{iz}^R} 
              < 0$
    are parameters.
    The sign of $\zeta_i$ is given by Assumption \ref{assu:commuting_shares}.
\end{proof}

This result implies that, up to an approximation, the response of rents to 
the workplace MW includes the negative effect of the MW on commuting shares.

\clearpage
\section{Data Appendix}

\subsection{Matching census blocks to USPS ZIP codes}
\label{sec:blocks_to_uspszip}

One challenge of this project is that LODES data on commuting patterns are 
aggregated at the level of the \textit{census block}.
However, Zillow data are aggregated at the level of \textit{USPS ZIP codes},
and blocks and ZIP codes are not nested.
In this appendix section we describe the steps we took to construct a 
correspondence table between these geographies.

First, we collected the GIS map of 11,053,116 blocks from \textcite{cbTiger} and
computed their centroids.
Second, we assigned each block to a unique ZIP code using the GIS map from 
\textcite{ESRI} based on assigning to each block the ZIP code that contains its 
centroid.
If the centroid falls outside the block, we pick a point inside it at random.
We assigned 11,013,203 blocks using the spatial match (99.64 percent of the 
total).%
\footnote{545,566 of ZIP codes assigned via spatial match use 
a point of the census block picked at random (4.94 percent of the total).}
Third, for the blocks that remain unassigned we used the tract-to-ZIP-code 
correspondence from \textcite{hudCrosswalks}.
Specifically, for each tract we keep the ZIP code where the largest number 
of houses of the tract fall, and we assign it to each block using the tract 
identifier.
We assigned 22,819 blocks using this approach (0.21 percent).
There remain 17,094 unassigned blocks (0.15 percent), which we drop from the 
analysis.
This creates a unique mapping from census blocks to ZIP codes.

In the end, there are 11,036,022 blocks which are assigned to 31,754 ZIP codes, 
implying an average of 347.55 blocks per ZIP code.
Thus, even though there may be blocks that go beyond one ZIP code, we expect 
the error introduced by this process to be very small.

\subsection{Assigning minimum wage levels to USPS ZIP codes}
\label{sec:assigning_mw_levels}

Our main rents data is aggregated at the level of the USPS ZIP code.
To match this geographical level, we assign statutory MW levels to ZIP codes.
ZIP codes usually cross jurisdictions, and as a result parts of them are subject
to different statutory MW levels.
Trying to overcome this problem, we assign averages of the relevant MW levels to
each ZIP code.

We proceed as follows.
First, we collect a census crosswalk constructed by \textcite{CensusLODES} that 
contains, for each block, identifiers for block group, tract, county, CBSA 
(i.e., Core-Based Statistical Area), place (i.e., Census-Designated Place), and 
state.
Second, we assign the MW level of each jurisdiction to the relevant block.
We use the state code for state MW policies, and we match local MW policies 
based on the names of the county and the place.
We define the statutory MW at each census block as the maximum of the federal,
state, county, and place levels.
Then, based on the original correspondence table described in Online Appendix 
\ref{sec:blocks_to_uspszip}, we assign a ZIP code to each block.
Finally, we define \textit{the statutory MW} for ZIP code $i$ and month $t$, 
$\MW_{it}$, as the weighted average of the statutory MW levels in its
constituent blocks, where the weights are given by the number of housing 
units.%
\footnote{ZIP codes between 00001 and 00199 correspond to federal territories.
Thus, we assign as statutory MW the federal level.}
For ZIP codes that have no housing units in them, such as those corresponding to 
universities or airports, we use a simple average instead.

\subsubsection*{Locating minimum wage earners}

We approximate the share of people that earn at or below the MW as follows.
First, we collect data on the number of workers in each tract from the 5-year 
2010-2014 American Community Survey \parencite{CensusACS}.
Using our assignment of hourly statutory MW levels in January 2014 we compute 
the total yearly wage of a full-time worker earning the MW in each tract, which 
we denote by $\underline{YW}$.%
\footnote{We use the definition of full-time workers from \textcite{IRSfulltime}.
Specifically, we assume that a full-time employee works for 130 hours per week
for 12 months.}
We keep track of what wage bin $\underline{YW}$ falls into.
We estimate the number of MW earners in a tract as the total number of workers 
in all bins below the one where $\underline{YW}$ falls plus a fraction of the 
total number of workers in the bin $\underline{YW}$ falls given by 
$\left(\underline{YW} - b_\ell\right)/\left(b_h - b_\ell\right),$
where $b_h$ and $b_\ell$ represent the upper and lower limits of 
the bin.
We impute the tract estimates to ZIP codes proportionally to the share of 
houses in each tract that fall in every ZIP code the tract overlaps with.%
\footnote{More precisely, we compute a tract-to-ZIP-code correspondence from
the LODES correspondence between blocks and tracts, available in 
\textcite{CensusLODES}, and the geographical match between blocks and ZIP codes
from Online Appendix \ref{sec:blocks_to_uspszip}.
For each tract, we compute the share of houses that fall in each ZIP code, and 
we assume that the share in the tract-ZIP code combination equals the share of
houses times the estimated number of MW workers in the tract.}
Finally, we compute the share of MW workers who reside in each ZIP 
code dividing our estimate of the number of MW workers by the total
number of workers in the data.

Due to limitations in the ACS data, it is not possible to use the MW at 
workplace locations in the computation, nor to estimate the share of MW workers 
by workplace.

\subsection{Measuring housing expenditure at the ZIP code level}
\label{sec:measure_housing_expenditure}

For our counterfactual exercises we require several pieces of information.
First, to estimate the overall incidence of a MW policy we need the levels 
of total wages and total housing expenditure in each location.
Second, to estimate the ZIP code-specific incidence, we require a housing 
expenditure share that varies by ZIP code.
We construct these measures for 2018 using data from the \textcite{IRS} and
the \textcite{hudSAFMR}.

To construct these data we proceed as follows.
We approximate the levels of total wages and housing expenditure using per 
household variables.
From the IRS we obtain annual wage per household, which we 
divide by 12 to obtain a monthly measure.
From the HUD, we use the 2-bedroom SAFMR series as our monthly housing 
expenditure variable.%
\footnote{Average rents in a location would be better approximated as a
weighted average of rents for houses with different number of bedrooms,
weighted by the share of households that rent each type of housing.
However, these data are not publicly available.}
We define the ZIP code-specific housing share as the ratio of these two 
variables.

The computed variables have several missing values across the entire US, and 
small percentage of missing values within urban CBSAs 
(as defined in Table \ref{tab:stats_zip_samples}).
We impute missing values independently for each variable using an OLS
regression based on sociodemographic characteristics of each ZIP code 
(including data from the US Census and LODES) and CBSA by county fixed effects.
To limit the influence of outliers, we winsorize the results at the 0.5th and 
99.5th percentiles. 
The percentage of urban ZIP codes with non-imputed housing expenditure shares 
is $93.2$.
%
%

\subsection{Posted rents and contract rents: Quora question}
\label{sec:posted_rents}
	
We asked the following question: ``How different is the rent paid by a tenant 
and the rent posted online of the same housing unit? Do tenants have space 
to bargain the posted price, or is it common for tenants to just accept the
posted price?''
The online version of the question can be accessed at
\href{https://www.quora.com/How-different-is-the-rent-paid-by-a-tenant-and-the-rent-posted-online-of-the-same-housing-unit-Do-tenants-have-space-to-bargain-the-posted-price-or-is-it-common-for-tenants-to-just-accept-the-posted-price}{\underline{this link}}.

\subsubsection*{Landlords}

\begin{itemize} \itshape
	\item As a landlord for over 40 years, I have never agreed to negotiate the rental
	price. I would rather lose that renter and later lower the list price and rent
	to someone else. ...
	
	\item ... The rent posted should be the actual rent available. As far as comparing 
	the rent to a newly available listed property and one that is being occupied 
	by another existing tenant, they will likely differ as the market continues 
	to evolve and inflation has an impact on everything as well as the law of 
	supply and demand. ...
	
	\item I've been a landlord for over 35 years, in three states and in three countries. 
	I do NOT ``bargain,'' and anyone trying would find their application in the round 
	file, instantly.
\end{itemize}
	
\subsubsection*{Tenants}
	
	\begin{itemize} \itshape
		\item Where I am, you accept the posted price unless there is something very 
		odd about the unit or you have a needed skill to offer. 
		Both are very rare.
		
		The payoff is that landlords often increase rent for renewing tenants at 
		a slightly lower rate than they set for a similar empty unit. 
		If you stay many years, you may have noticeably lower rent. ...
		
		\item Presumably, the landlord has done at least minimal research to figure 
		out how much the market is charging. That market price should bring in 
		multiple potential tenants wanting to rent. Given that, it would be very 
		rare for a landlord to accept lower offers. ...
	\end{itemize}
	
From a Real Estate Transaction Coordinator (boldface added by us):
	
\begin{itemize} \itshape
	\item 
	That's entirely up to the landlord.

	When dealing with a property management company or the manager of an
	apartment complex, they may have limits on what they can do as far as
	negotiations. I'm not familiar with any in my city who negotiate rent.
	\textbf{What's posted is the price. Period.}
	
	We mainly deal with private landlords who, of course, have 100\% control 
	over whether they are willing to negotiate the rent.
	\textbf{In 15 years, I've only seen it happen twice.}
	But our clients post the rent at a fair price—generally right in the 
	center of the ``fair market value range''—so they have no reason to 
	negotiate.

	What they WILL negotiate, and I've seen it done many times, is how the 
	deposits are collected. ...
\end{itemize}

\clearpage
\section{Identification in a Potential Outcomes Framework}
\label{sec:potential_outcomes}

Following Section \ref{sec:model}, we assume that the effect of MW policies 
across locations can be summarized in the residence and workplace MW measures.
Thus, we consider the following causal model
\begin{equation}\label{eq:causal_model}
    r_{it} = r_{it}(\mw_{it}^{\res}, \mw_{it}^{\wkp}) .
\end{equation}

For this section we represent our dataset as
$\left\{\{r_{it}, \mw^{\res}_{it}, \mw^{\wkp}_{it}\}
       _{t=\underline{T}}^{\overline{T}}\right\}_{i\in\Z}$.
Monthly dates run from $\underline{T}$ to $\overline{T}$ for every unit,
and $\Z$ is the set of ZIP codes.
We assume that the data are $iid$.
We impose no anticipation, so units do not change their pretreatment outcome 
given future changes in the MW measures.

Every month in which some jurisdiction changes the level of the MW there will 
be units that are treated directly and units that are treated indirectly.
We follow \textcite{AngristImbens1995, CallawayEtAl2021} to define the 
treatment effects of interest.
We denote a unit's causal response to the residence MW as
$\partial r_{it}(\mw_{it}^{\res}, \mw_{it}^{\wkp})/\partial \mw_{it}^{\res} ,$
and to the workplace MW as 
$\partial r_{it}(\mw_{it}^{\res}, \mw_{it}^{\wkp})/\partial \mw_{it}^{\wkp} .$
Let the federal MW level be $\mw^{\fed}$.

\begin{definition}[Treatment Effects]\label{def:treatment_effects}
    Consider a group with a residence MW level of $w^{\res}$ and a workplace
    MW level of $w^{\wkp}$.
    Focus on the effect of the workplace MW.
    The average treatment effect on that group is
    \begin{equation*}
        ATT^{\wkp}(w^{\wkp} | w^{\res}, w^{\wkp}) 
            = \E\left[r_{it}(w^{\res}, w^{\wkp}) - r_{it}(w^{\res}, \mw^{\fed}) 
                    \big| \mw_{it}^{\res}=w^{\res}, \mw_{it}^{\wkp}=w^{\wkp}\right] .
    \end{equation*}
    The average causal response of the same group to the workplace MW is given by
    \begin{equation*}
        ACRT^{\wkp}(w^{\wkp} | w^{\res}, w^{\wkp}) 
            = \frac{\partial \E\left[r_{it}(w^{\res}, l) | \mw_{it}^{\res}=w^{\res}, \mw_{it}^{\wkp}=w^{\wkp}\right]}
                    {\partial \mw^{\wkp}} \Bigg|_{l=w^{\wkp}} .
    \end{equation*}
    These treatment effects may be heterogeneous across the distribution 
    of $(\mw_{it}^{\res},\mw_{it}^{\wkp})$.
    The average causal response across all groups treated with different 
    levels of the workplace and residence MW is
    \begin{equation*}
        ACR^{\wkp}(w^{\wkp}) = \frac{\partial \E\left[r_{it}(w^{\res}, w^{\wkp}) \right] }
                                    {\partial w^{\wkp}} .
    \end{equation*}
    Analogously, for the residence MW we define: $ATT^{\res}$, 
    $ACRT^{\res}(w^{\res} | w^{\res}, w^{\wkp})$, and $ACR^{\res}(w^{\res})$.
\end{definition}

Our main interest lies in the rent gradient to the MW, i.e., the 
average causal response of rents to each of the MW measures.
For that, we make a parallel trends assumption.

\begin{assu}[Parallel trends] \label{assu:PT}
    We assume that, for all levels of $w^{\res}$ and $w^{\wkp}$,
    \begin{align*}\label{eq:PT}
        \E\big[r_{it}(\mw^{\fed}, \mw^{\fed}) & - r_{i,t-1}(\mw^{\fed}, \mw^{\fed}) 
                \big| \mw_{it}^{\res}=w^{\res}, \mw_{it}^{\wkp}=w^{\wkp} \big] \\
        & = \E\big[r_{it}(\mw^{\fed}, \mw^{\fed}) - r_{i,t-1}(\mw^{\fed}, \mw^{\fed})
                \big| \mw_{it}^{\res}=w^{\res}, \mw_{it}^{\wkp}=\mw^{\fed} \big] \\
        & = \E\Big[r_{it}(\mw^{\fed}, \mw^{\fed}) - r_{i,t-1}(\mw^{\fed}, \mw^{\fed})
                \big| \mw_{it}^{\res}=\mw^{\fed}, \mw_{it}^{\wkp}=w^{\wkp} \big] .
    \end{align*}
\end{assu}

Assumption \ref{assu:PT} states that the untreated outcomes evolve in parallel 
between ZIP codes experiencing treatment levels $(w^{\res},w^{\wkp})$ and 
(a) ZIP codes with the same level of the residence MW but unchanged workplace MW and
(b) ZIP codes with the same level of the workplace MW but unchanged residence MW.
We further maintain a second assumption.

\begin{assu}[No selection on gains] \label{assu:no_selection}
    We assume that
    \begin{equation*}\label{eq:no_selection_workplace}
        \frac{\partial ATT^{\wkp}(w^{\wkp} | w^{\res}, l)}
             {\partial \mw^{\wkp}} \Big|_{l = w^{\wkp}} = 0
        \quad\text{\textit{ and }}\quad
        \frac{\partial ATT^{\res}(w^{\res} | l, w^{\wkp})}
             {\partial \mw^{\res}} \Big|_{l = w^{\res}} = 0 .
    \end{equation*}
\end{assu}

To identify $ACRT^{\wkp}$ we will compare ZIP codes that received similar levels 
of the residence MW and different levels of the workplace MW.
Analogous comparisons of ZIP codes with different residence MW and similar 
workplace MW will idenfity $ACRT^{\res}$.

\begin{prop}[Identification]\label{prop:PO_identification}
    Under Assumption \ref{assu:PT} we have that
    \begin{equation*}
        \begin{split}
        \frac{\partial \E\left[r_{it}(w^{\res}, w^{\wkp}) 
                              | \mw_{it}^{\res}=w^{\res}, \mw_{it}^{\wkp}=w^{\wkp}\right]}
             {\partial \mw^{\wkp}} 
        & = ACRT^{\wkp}(w^{\wkp} | w^{\res}, w^{\wkp})            \\
        &  + \frac{\partial ATT^{\wkp}(w^{\wkp} | w^{\res}, l)}
                  {\partial \mw^{\wkp}} \Big|_{l = w^{\res}} .    \\
        \end{split}
    \end{equation*}
    Furthermore, if Assumption \ref{assu:no_selection} holds, then
    \begin{equation*}
        \frac{\partial \E\left[r_{it}(w^{\res}, w^{\wkp}) | w^{\res}, w^{\wkp}\right]}
        {\partial w^{\wkp}} 
       = ACRT^{\wkp}(w | w^{\res}, w) .
    \end{equation*}
    Analogous expressions hold for the residence MW.
\end{prop}
\begin{proof}
    The setting is analogous to \textcite{CallawayEtAl2021} but with two 
    treatment variables.
    The proof is analogous as well, with the only difference being that one 
    must condition on the residence MW when deriving the expression for the 
    workplace MW, and viceversa.
\end{proof}

As extensively discussed by \textcite{CallawayEtAl2021}, Assumption \ref{assu:PT} 
is not enough to identify the average causal response in the context of 
continuous treatments.
The gradient of our rents function for the group $(w^{\res}, w^{\wkp})$ is a mix
of the average causal response of interest and a ``selection bias'' term that
captures the fact that the treatment for the particular group that received
$(w^{\res}, w^{\wkp})$ may be different for other groups at that level of 
treatment.
Assumption \ref{assu:no_selection} imposes that those selection bias
terms are zero.%
\footnote{There are several alternatives to this assumption. 
See \textcite[][Section 3.3]{CallawayEtAl2021} and discussion therein.}
We discuss the plausibility of these assumptions in Section 
\ref{sec:empirical_strategy}.

Consider now a functional form for \eqref{eq:causal_model} like the one used in 
the main analysis:
$$
r_{it} = \alpha_i + \tilde\delta_t 
         + \gamma \mw_{it}^{\res} + \beta \mw_{it}^{\wkp}
         + \epsilon_{it}
$$
where we exclude the controls for simplicity.
It is easy to see, if 
$\E[\epsilon_{it} | \mw_{it}^{\res}, \mw_{it}^{\wkp}] = 0$,
then both Assumptions \ref{assu:PT} and \ref{assu:no_selection} hold under 
this linear functional form with constant effects.
Furthermore, in this case we have that
\begin{equation*}
    ACRT^{\wkp}(w^{\wkp} | w^{\res}, w^{\wkp}) 
        = ACR^{\wkp}(w^{\wkp} | w^{\res}, w^{\wkp}) 
        = \beta
\end{equation*}
and that
\begin{equation*}
    ACRT^{\res}(w^{\res} | w^{\res}, w^{\wkp})
        = ACR^{\res}(w^{\res} | w^{\res}, w^{\wkp})
        = \gamma
\end{equation*}
for any $w^{\res}\geq\mw^{\fed}$ and $w^{\wkp}\geq\mw^{\fed}$.

\clearpage
\section{Additional Tables and Figures}

\begin{table}[hbt!] \centering
    \caption{Summary statistics of baseline panel}
    \label{tab:stats_est_panel}
    \begin{tabular}{@{}lccccc@{}}
        \toprule
                                          & \multicolumn{1}{c}{N} 
                                          & \multicolumn{1}{c}{Mean} 
                                          & \multicolumn{1}{c}{St. Dev.} 
                                          & \multicolumn{1}{c}{Min} 
                                          & \multicolumn{1}{c}{Max}                 \\ \midrule
        \textit{Minimum wage variables}               &       &       &       &       &       \\
        $\quad$Statutory MW $\MW_{it}$                & 80,700  & 8.56  & 1.58  & 7.25  & 16.00  \\
        $\quad$Residence MW $\mw^{\res}_{it}$         & 80,700  & 2.132  & 0.168  & 1.981  & 2.773  \\
        $\quad$Workplace MW $\mw^{\wkp}_{it}$         & 80,700  & 2.136  & 0.163  & 1.981  & 2.694  \\
        $\quad$Workplace MW, low-income workers       & 80,700  & 2.134  & 0.161  & 1.981  & 2.681  \\
        $\quad$Workplace MW, young workers            & 80,700  & 2.135  & 0.163  & 1.981  & 2.707  \\[.3em]
        \textit{Median Rents}                         &       &       &       &       &       \\
        $\quad$SFCC                                   & 74,012  & 1,757.89  & 901.50  & 625.00  & 30,000.00  \\
        $\quad$SFCC per sqft.                         & 80,700  & 1.32  & 1.01  & 0.47  & 22.20  \\
        $\quad$Log(SFCC per sqft.)                    & 80,700  & 0.14  & 0.47  & -0.76  & 3.10  \\[.3em]
        \textit{Economic controls}                    &       &       &       &       &       \\
        $\quad$Avg.\ wage Business services           & 80,700  & 11.19  & 1.38  & 6.02  & 13.39  \\
        $\quad$Employment Business services           & 80,700  & 8.71  & 1.25  & 4.36  & 10.96  \\
        $\quad$Estab. count Business services         & 80,700  & 7.14  & 0.31  & 5.73  & 8.18  \\
        $\quad$Avg.\ wage Financial services          & 80,352  & 9.01  & 1.57  & 2.40  & 12.39  \\
        $\quad$Employment Financial services          & 80,700  & 6.13  & 1.35  & 1.61  & 9.53  \\
        $\quad$Estab. count Financial services        & 80,352  & 7.33  & 0.36  & 5.89  & 8.91  \\
        $\quad$Avg.\ wage Information services        & 80,688  & 10.23  & 1.43  & 4.75  & 12.90  \\
        $\quad$Employment Information services        & 80,700  & 8.01  & 1.21  & 3.66  & 10.34  \\
        $\quad$Estab. count Information services      & 80,688  & 7.31  & 0.37  & 6.33  & 9.16  \\ \bottomrule
    \end{tabular}

    \begin{minipage}{.95\textwidth} \footnotesize
        \vspace{2mm}
        Notes: This table shows summary statistics of the panel of ZIP codes 
        used in our baseline results, constructed as explained in Section 
        \ref{sec:data_final_panel}.
        All workplace MW variables use 2017 commuting data from LODES.
        The workplace MW variables ``Workplace MW, low-income workers'' and 
        ``Workplace MW, young workers'' are constructed using data for 
        workers who earn less \$1,251 and are aged less than 29, respectively.
    \end{minipage}
\end{table}

\clearpage
\begin{table}[hbt!] \centering
    \caption{Estimates of the effect of the minimum wage on rents in levels and first differences,
             baseline sample}
    \label{tab:autocorrelation}
    \begin{tabular}{@{}lcc@{}}
        \toprule
            & \multicolumn{2}{c}{Log rents}                            \\ \cmidrule(l){2-3}
            & \shortstack{Levels\\(1)} 
            & \shortstack{First Differences\\(2)}                      \\ \midrule
        Residence MW                      &  -0.0432   &  -0.0199              \\
                                          & (0.1751)  & (0.0195)             \\
        Workplace MW                      &  0.0376   &  0.0687              \\
                                          & (0.2033)  & (0.0298)             \\ \midrule
        Economic controls                 &  Yes   &  Yes              \\
        P-value autocorrelation test      &        &  $<0.0001$        \\
        R-squared                         &  0.9924   &  0.0216              \\
        Observations                      &  80,340  &  78,912             \\ \bottomrule
    \end{tabular}

    \begin{minipage}{.95\textwidth} \footnotesize
        \vspace{2mm}
        Notes:
        Data are from the baseline estimation sample described in Section 
        \ref{sec:data_final_panel}.
        Both columns report the results of regressions of the log of 
        median rents per square foot on our MW-based measures.
        Column (1) presents estimates of a model in levels, including 
        ZIP code and year-month fixed effects.
        Column (2), presents estimates of a model in first differences, 
        including year-month fixed effects 
        (note that the ZIP code fixed effect drops out).
        For the model in first differences, we also report the results of an 
        AR(1) auto-correlation test.
        We proceed as in \textcite[][Section 10.6.3]{wooldridge2010}.
        First, we compute the residuals of the model estimated in column (2), 
        and we regress those residuals on their lag.
        Let the auto-correlation coefficient of this model be $\phi$.
        The model in levels is efficient assuming no auto-correlation in the 
        error term, which would imply that the residuals of the 
        first-differenced model are auto-correlated with $\phi = -0.5$.
        The row ``P-value autocorrelation test'' reports the $p$-value of 
        a Wald test of that hypothesis.
        Standard errors in parentheses are clustered at the state level.
    \end{minipage}
\end{table}

\clearpage
\begin{table}[hbt!] \centering
    \caption{Estimates of the effect of the minimum wage on rents, stacked sample}
    \label{tab:stacked_w6}
    \begin{tabular}{l*{4}{c}}
        \toprule
        & \multicolumn{1}{c}{\shortstack{Change wkp.\ MW\\$\Delta\mw_{it}^{\wkp}$}}
            & \multicolumn{3}{c}{\shortstack{Change log rents\\$\Delta r_{it}$}} \\ \cmidrule(lr){2-2}\cmidrule(lr){3-5}
                                            & (1)   & (2)   & (3)   & (4)            \\ \midrule
        Change residence MW 
                    $\Delta\mw_{it}^{\res}$  &  0.5461  &  0.0051  &       &  -0.0444     \\
                                            & (0.0316) & (0.0109) &       & (0.0174)    \\
        Change workplace MW 
                    $\Delta\mw_{it}^{\wkp}$ &       &       &  0.0242  & 0.0906      \\
                                            &       &       & (0.0216) & (0.0391)    \\ \midrule
        Sum of coefficients                &       &       &       &  0.0463     \\
                                            &       &       &       & (0.0266)    \\ \midrule
        Economic controls                   &  Yes  & Yes   & Yes   & Yes      \\
        P-value equality                   &       &       &       & 0.0208      \\
        R-squared                          &  0.9763  &  0.0539  &  0.0540  & 0.0540      \\
        Observations                       & 98,326  & 98,326  & 98,326  & 98,326     \\\bottomrule
    \end{tabular}

    \begin{minipage}{.95\textwidth} \footnotesize
        \vspace{2mm}
        Notes: 
        Data are from Zillow, the MW panel described in Section \ref{sec:data_mw_panel},
        LODES origin-destination statistics, and the QCEW.
        The table mimics the estimates in Table \ref{tab:static} using a 
        ``stacked'' sample.
        To construct the sample we proceed as follows.
        First, we define a CBSA-month as treated if in that month there is at 
        least one ZIP code that had a change in the binding MW.
        We drop events that have less than 10 ZIP codes.
        For each of the selected CBSA-months we assign a unique event ID. 
        Second, for each event ID we take a window of 6 months, and we keep all 
        months within that window for the ZIP codes that belong to the treated 
        CBSA.
        If a ZIP code has missing data for some month within the window, we drop 
        the entire ZIP code from the respective event.
        For each column, we estimate the same regression as the analogous column 
        in Table \ref{tab:static} but include event ID by year-month
        fixed effects.
    \end{minipage}
\end{table}

\clearpage
\begin{table}[hbt!]
    \centering
    \caption{Estimates of the effect of the minimum wage on rents including one lag of the
             dependent variable, baseline sample}
    \label{tab:arellano_bond}

    \begin{tabular}{@{}lcccc@{}}
        \toprule
                                             & \multicolumn{4}{c}{Log rents}                                                       \\ \cmidrule(lr){2-5} 
                                             & \multicolumn{2}{c}{Levels}               & \multicolumn{2}{c}{First Differences}    \\ \cmidrule(lr){2-3}\cmidrule(lr){4-5} 
        \multicolumn{1}{c}{}                 & \shortstack{Baseline\\(1)} 
                                             & \shortstack{Arellano\\Bond (2)} 
                                             & \shortstack{Baseline\\(3)} 
                                             & \shortstack{Arellano\\Bond (4)}                                                     \\ \midrule
        Residence MW                         & -0.0432                  & -0.0055               & -0.0219                  & -0.0221               \\
                                             & (0.1751)                & (0.0298)             & (0.0175)                & (0.0234)             \\
        Workplace MW                         & 0.0376                  & 0.0065               & 0.0685                  & 0.0702               \\
                                             & (0.2033)                & (0.0346)             & (0.0288)                & (0.0390)             \\
        Lagged log rents                     &                      & 0.8421               &                      & 0.3299               \\
                                             &                      & (0.0179)             &                      & (0.0177)             \\ \midrule
        Economic controls                    & Yes                  & Yes               & Yes                  & Yes               \\
        P-value equality                     & 0.8264                  & 0.8481               & 0.0514                  & 0.1378               \\
        Observations                         & 80,340                 & 80,321              & 80,241                 & 80,217              \\ \bottomrule

    \end{tabular}

    \begin{minipage}{.95\textwidth} \footnotesize
        \vspace{2mm}
        Notes: 
        Data are from the baseline estimation sample described in Section 
        \ref{sec:data_final_panel}.
        All columns show the results of regressions of the log of median rents 
        per square foot on the residence and workplace MW measures.
        Columns (1) and (2) estimate two-way fixed-effects regressions in 
        levels that include ZIP code and year-month fixed effects.
        Columns (3) and (4) estimate models in first differences that include 
        year-month fixed effects.
        All regressions include economic controls (in levels or first differences,
        respectively) that vary at the county by month and county by quarter levels.
        Odd-numbered columns are estimated under OLS.
        Even-numbered columns include the lagged variable of the dependent variable
        as control, and are estimated using an IV strategy where the first lag is 
        instrumented with the second lag, following \textcite{ArellanoBond1991}.
        The measure of rents per square foot corresponds to the SFCC category 
        from Zillow.
        Economic controls from the QCEW include the log of the average wage, 
        the log of employment, and the log of the establishment count from the 
        sectors ``Information'', ``Financial activities'', and ``Professional
        and business services''.
        Standard errors in parentheses are clustered at the state level.
    \end{minipage}
\end{table}

\clearpage
\begin{table}[hbt!]
    \centering
    \caption{Estimates of the effect of the minimum wage on rents, different samples}
    \label{tab:static_sample}

    \begin{tabular}{@{}lccc@{}}
        \toprule
                                             & \multicolumn{3}{c}{Change log rents $\Delta r_{it}$}            \\ \cmidrule(l){2-4} 
                                             & \shortstack{Baseline\\(1)}       & \shortstack{Reweighted\\(2)} 
                                             & \shortstack{Unbalanced\\(3)}                                    \\ \midrule
        Change residence MW 
                  $\Delta\mw_{it}^{\res}$    & -0.0219      & -0.0039        & -0.0274         \\
                                             & (0.0175)    & (0.0128)      & (0.0237)       \\
        Change workplace MW 
                   $\Delta\mw_{it}^{\wkp}$   & 0.0685      & 0.0558        & 0.0528         \\
                                             & (0.0288)    & (0.0240)      & (0.0299)       \\ \midrule
        P-value equality                     & 0.0514      & 0.1026        & 0.1325         \\
        R-squared                            & 0.0213      & 0.0209        & 0.0309         \\
        Observations                         & 80,241     & 79,701       & 193,239        \\ \bottomrule
    \end{tabular}

    \begin{minipage}{.95\linewidth} \footnotesize
        \vspace{2mm}
        Notes:
        Data are from Zillow,
        the statutory MW panel described in Section \ref{sec:data_mw_panel},
        LODES origin-destination statistics, and the QCEW.
        Every column shows the results of regressions of the log of
        median rents per square foot on our MW-based measures.
        All regressions are estimated in first differences and include 
        time-period fixed effects and economic controls that vary at the 
        county by month and county by quarter levels.
        The measure of rents per square foot corresponds to the Single Family, 
        Condominium and Cooperative houses from Zillow.
        Columns (1) and (2) use our baseline sample defined in Section 
        \ref{sec:data_final_panel}.
        Column (3) uses the unbalanced sample of all ZIP codes with 
        Zillow rents data at any point in time, and controls for quarter-year 
        of entry to the panel by year-month fixed effects.
        Column (2) re-weights observations so that the sample of 
        ZIP codes in the data (column 3 of Table \ref{tab:stats_zip_samples}) 
        matches the averages of the set of ZIP codes located in urban CBSAs 
        (column 2 of Table \ref{tab:stats_zip_samples})
        in the following census variables:
        share of urban population
        share of renter-occupied households, and
        share of white population.
        Weights for each sample are computed following \textcite{Hainmueller2012}.
        Standard errors in parentheses are clustered at the state level.
    \end{minipage}
\end{table}

\clearpage
\begin{landscape}
\begin{table}[ht!]
    \centering
    \caption{Comparison of estimates of the effect of the minimum wage on rents across
             Zillow categories, unbalanced samples}
    \label{tab:zillow_categories}
        
    \begin{tabular}{@{}lccccc@{}}
        \toprule
                                             & \multicolumn{1}{c}{\shortstack{Change wkp.\ MW\\$\Delta\mw_{it}^{\wkp}$}} 
                                             & \multicolumn{3}{c}{\shortstack{Change log rents\\$\Delta r_{it}$}}
                                             &                                                                         \\ \cmidrule(lr){2-2}\cmidrule(lr){3-5}
                                                 & \multicolumn{1}{c}{\shortstack{Change res.\ MW\\$\Delta\mw_{it}^{\res}$}}
                                                 & \multicolumn{1}{c}{\shortstack{Change res.\ MW\\$\Delta\mw_{it}^{\res}$}} 
                                                 & \multicolumn{1}{c}{\shortstack{Change wkp.\ MW\\$\Delta\mw_{it}^{\wkp}$}} 
                                                 & \shortstack{Sum of\\coefficients} 
                                                 & N                                    \\ \midrule
        $\quad$(a) Unbalanced (SFCC)             &  0.8476  &  -0.0263  &  0.0479  &  0.0216  & 193,292 \\
                                                 & (0.0297) & (0.0213) & (0.0302) & (0.0157) &      \\
        $\quad$(b) Single family (SF)            &  0.8588  &  -0.0169  &  0.0429  &  0.0260  & 140,750 \\
                                                 & (0.0315) & (0.0399) & (0.0477) & (0.0138) &      \\
        $\quad$(c) Condo/Cooperatives (CC)       &  0.8019  &  -0.0648  &  0.0968  &  0.0320  & 29,817 \\
                                                 & (0.0288) & (0.0266) & (0.0417) & (0.0199) &      \\
        $\quad$(d) Studio                        &  0.8330  &  -0.0669  &  0.0776  &  0.0107  & 22,746 \\
                                                 & (0.0287) & (0.0520) & (0.0570) & (0.0206) &      \\
        $\quad$(d) 1 Bedroom                     &  0.7879  &  0.0287  &  -0.0327  &  -0.0039  & 53,538 \\
                                                 & (0.0300) & (0.0269) & (0.0456) & (0.0208) &      \\
        $\quad$(e) 2 Bedroom                     &  0.8022  &  -0.0069  &  0.0063  &  -0.0006  & 89,635 \\
                                                 & (0.0296) & (0.0232) & (0.0285) & (0.0114) &      \\
        $\quad$(f) 3 Bedroom                     &  0.8113  &  -0.0645  &  0.0920  &  0.0275  & 64,916 \\
                                                 & (0.0322) & (0.0475) & (0.0682) & (0.0328) &      \\
        $\quad$(g) Multifamily 5+ units          &  0.8072  &  -0.0133  &  0.0369  &  0.0236  & 142,759 \\
                                                 & (0.0314) & (0.0260) & (0.0362) & (0.0115) &      \\ \bottomrule
    \end{tabular}

    \begin{minipage}{.95\linewidth} \footnotesize
        \vspace{2mm}
        Notes:
        Data are from Zillow, the statutory MW panel described in
        Section \ref{sec:data_mw_panel}, LODES origin-destination statistics,
        and the QCEW.
        Each row of the table shows two estimations on the same sample of ZIP 
        codes and months.
        The first column shows the results of a regression of the change in the 
        workplace MW measure on the change in the residence MW measure.
        The second through fourth columns show the results of a regression of 
        the change in log rents on the change in the residence MW and the 
        workplace MW, with the fifth column showing the sum of the coefficients 
        on the MW measures.
        All rent variables correspond to the median per square foot rent in a 
        Zillow category.
        All estimated regressions include quarter of entry to Zillow by year-month 
        fixed effects and economic controls from the QCEW.
        Row (a) repeats the results of column (5) of Table \ref{tab:static_sample}, 
        using the Single Family, Condominium and Cooperative Houses category.
        Rows (b) through (g) estimate the same regression for different Zillow 
        categories.
        We exclude the rental categories ``4 bedroom,'' ``5 bedroom,'', and
        ``Duplex and triplex,'' all of which contain less than 15 thousand
        ZIP code by month observations.
        Standard errors in parentheses are clustered at the state level.
    \end{minipage}
\end{table}
\end{landscape}

\clearpage

\begin{figure}[h!]
    \centering
    \caption{Minimum wage levels in the US by jurisdiction between 
             January 2010 and June 2020}
    \label{fig:mw_policies}

    \begin{subfigure}{.7\textwidth}
        \caption{State policies}
        \includegraphics[width = \textwidth]
            {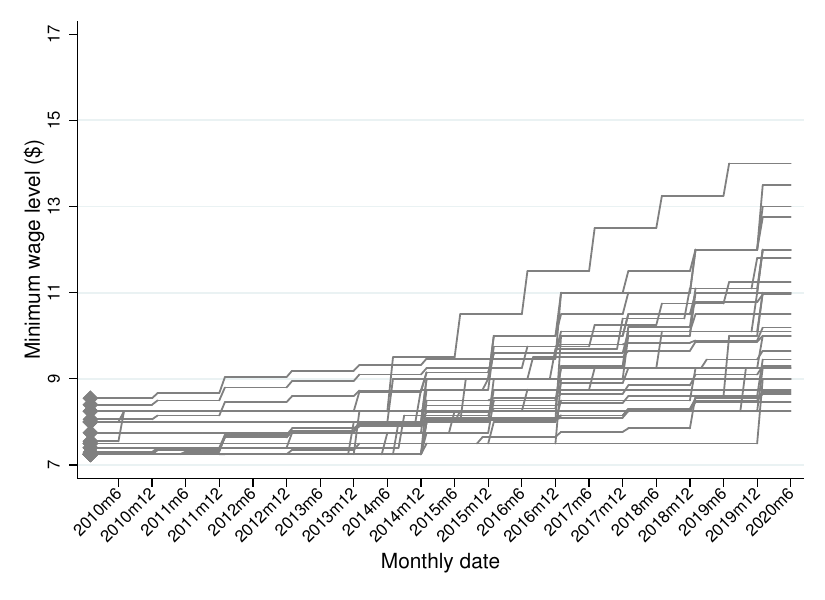}
    \end{subfigure}\\
    \begin{subfigure}{.7\textwidth}
        \caption{Sub-state policies}
        \includegraphics[width = \textwidth]
            {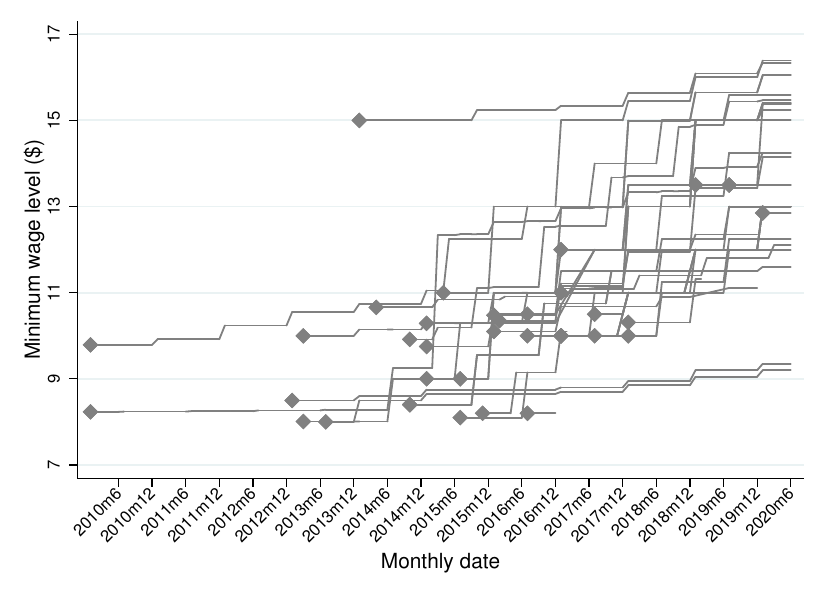}
    \end{subfigure}

    \begin{minipage}{.95\textwidth} \footnotesize
        \vspace{3mm}
        Notes:
        Data are from the MW panel described in Section
        \ref{sec:data_mw_panel}.
        Lines show the levels of the MW for jurisdictional policies
        that were binding for at least one ZIP code inside them in some month 
        between January 2010 and June 2020.
        Diamonds indicate the first month the MW policy became
        operational within the same period.
        Panel A reports state level policies.
        Panel B reports sub-state level policies.
    \end{minipage}
\end{figure}

\clearpage
\begin{figure}[h!]
    \centering
    \caption{Probability of being a head of household, by income decile}
    \label{fig:ahs_hhead}

    \includegraphics[width = .75\textwidth]{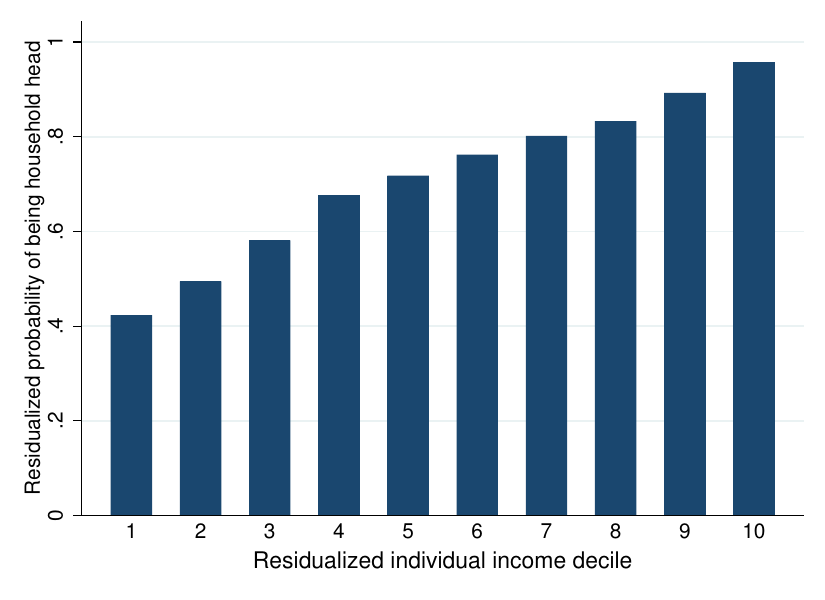}
    
    \begin{minipage}{.95\textwidth} \footnotesize
        \vspace{3mm}
        Notes: Data are from the 2011 and 2013 American Housing Surveys.
        The figure shows the probability that an individual is a head of
        household, by individual income decile.
        We construct the figure as follows.
        First, we residualize the variable in the y-axis and individual income 
        by SMSA indicators, the closest analogue of CBSAs available in the data.
        Second, we construct deciles of the residualized individual income 
        variable.
        Finally, we take the average of the residualized y-variable within each 
        decile.
        Individuals that do not work are excluded from the figure.
    \end{minipage}
\end{figure}

\clearpage
\begin{figure}[h!]
    \centering
    \caption{Average rent, square footage, and rent per square foot by household 
             income decile, renters sample}
    \label{fig:ahs_rent_sqft}

    \begin{subfigure}{.69\textwidth} \centering
        \caption*{Rents per square foot}
        \includegraphics[width = 1\textwidth]
            {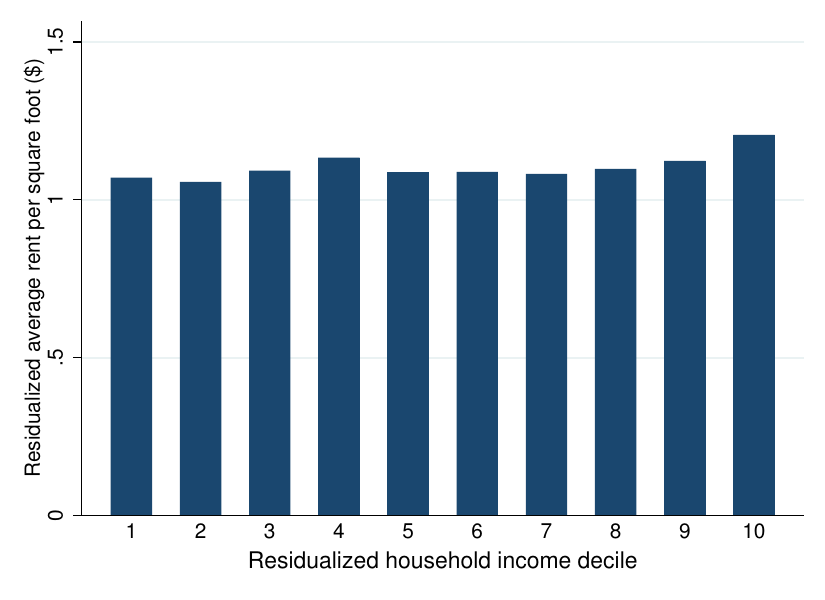}
    \end{subfigure}\\
    \begin{subfigure}{.48\textwidth} \centering
        \caption*{Rents}
        \includegraphics[width = .99\textwidth]
            {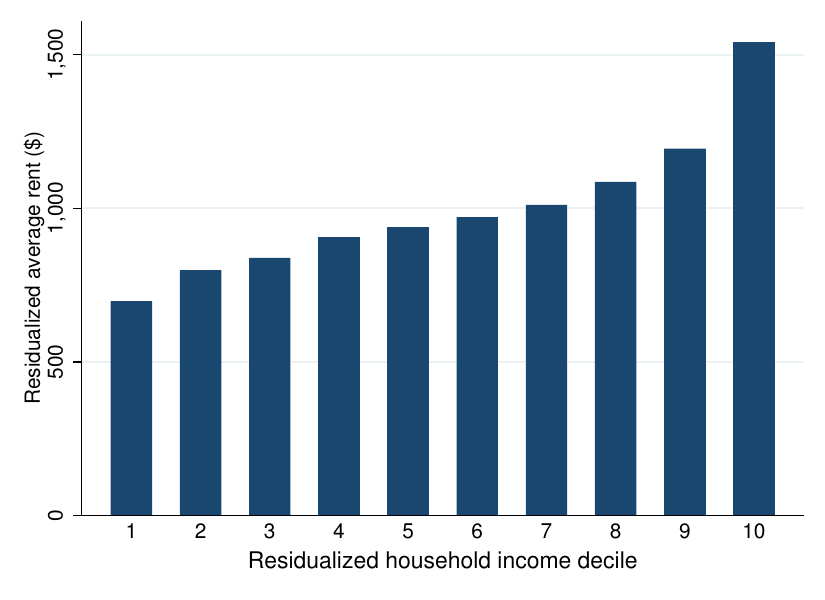}
    \end{subfigure}%
    \begin{subfigure}{.48\textwidth} \centering
        \caption*{Square footage}
        \includegraphics[width = .99\textwidth]
            {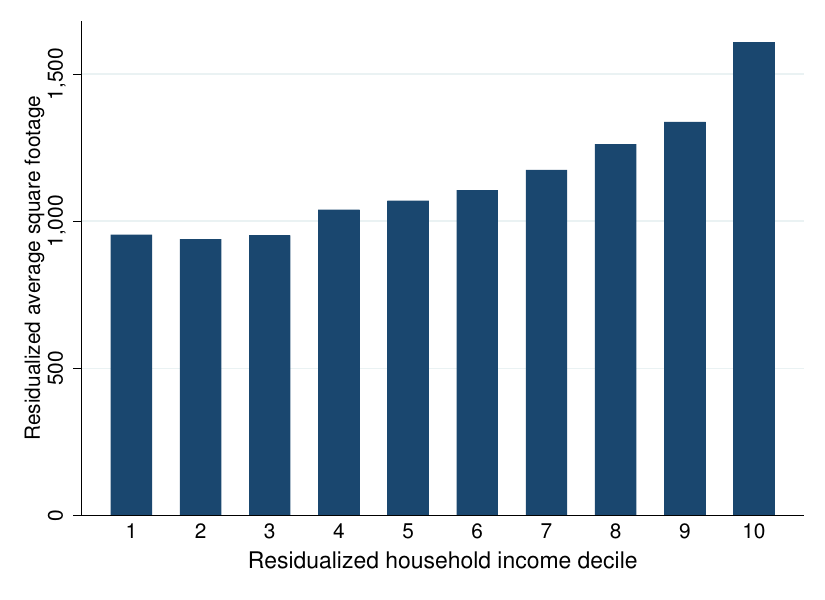}
    \end{subfigure}

    \begin{minipage}{.95\textwidth} \footnotesize
        \vspace{3mm}
        Notes: Data are from the 2011 and 2013 American Housing Surveys.
        The top figure shows average rents per square foot by household income.
        The bottom left figure shows average rents by household income.
        The bottom right figure shows average square feet by household income.
        The variable rent per square foot is defined as total rental payments 
        divided by total square feet.
        We construct the figure as follows. 
        First, we residualize the variable in the y-axis and household income 
        by SMSA indicators, the closest analogue of CBSAs available in the data.
        Second, we construct deciles of the residualized household income 
        variable.
        Finally, we take the average of the residualized y-variable within each 
        decile.
        The sample includes households with non-missing values for square 
        footage and rental payments.
        We exclude from the calculation non-conventional housing units, such as 
        mobile homes, hotels, and others.
    \end{minipage}
\end{figure}

\clearpage
\begin{figure}[h!]
    \centering
    \caption{Properties of building where household unit is located by
             household income decile, full sample}
    \label{fig:ahs_unit_types}

    \begin{subfigure}{.75\textwidth} \centering
        \caption*{Distribution of the number of units in building}
        \includegraphics[width = 1\textwidth]
            {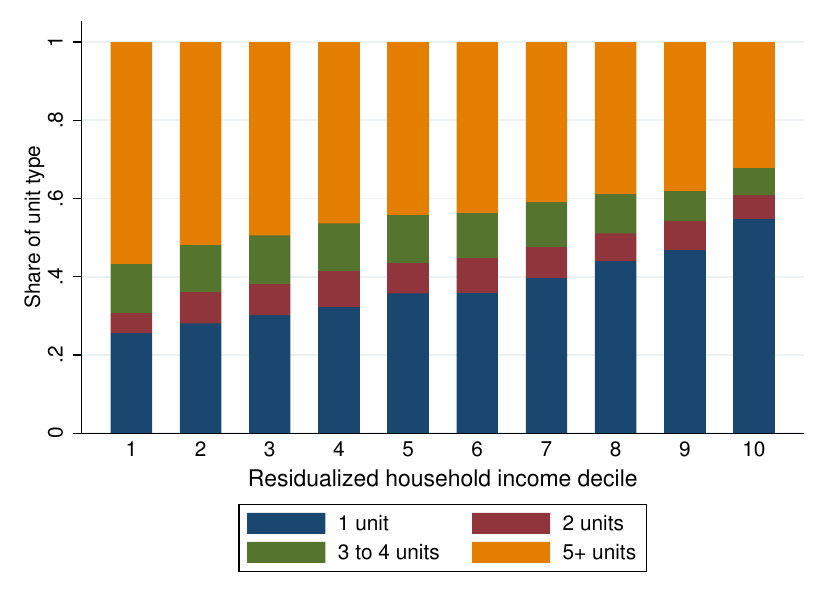}
    \end{subfigure}\\
    \begin{subfigure}{.75\textwidth} \centering
        \caption*{Probability building is a condominium or cooperative housing}
        \includegraphics[width = 1\textwidth]
            {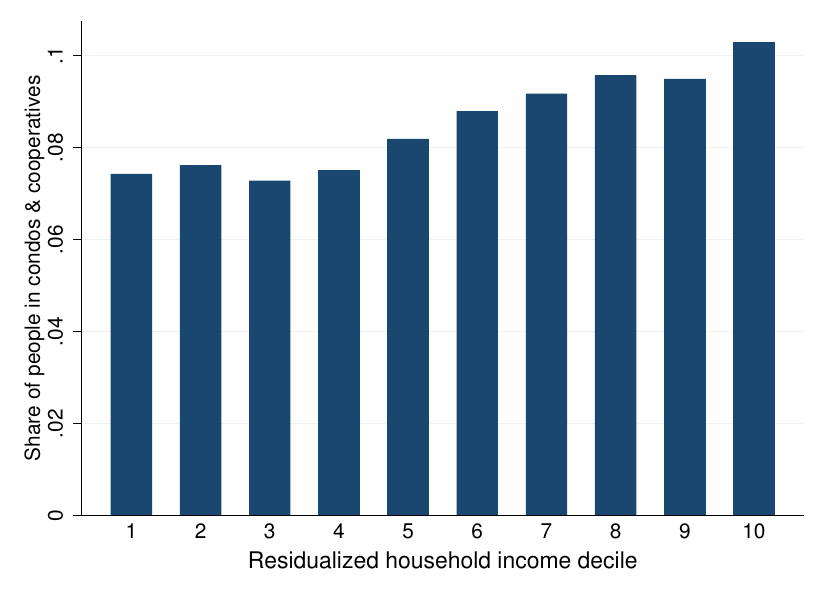}
    \end{subfigure}%

    \begin{minipage}{.95\textwidth} \footnotesize
        \vspace{3mm}
        Notes: Data are from the 2011 and 2013 American Housing Surveys.
        The top figure shows the number of housing units in the building where 
        the household is located, and the bottom figure shows the share of 
        housing units located in condominiums or cooperative housing, both by 
        household income.
        We construct the figure as follows.
        First, we residualize the variable in the y-axis and household income 
        by SMSA indicators, the closest analogue of CBSAs available in the data.
        Second, we construct deciles of the residualized household income 
        variable.
        Finally, we take the average of the residualized y-variable within each 
        decile.
        We exclude from the calculation non-conventional housing units, such as 
        mobile homes, hotels, and others.
    \end{minipage}
\end{figure}

\clearpage
\begin{figure}[h!]
    \centering
    \caption{Estimated housing expenditure shares in 2018, Chicago-Naperville-Elgin CBSA}
    \label{fig:map_hous_exp_chicago}

    \includegraphics[width = 0.65\textwidth]
            {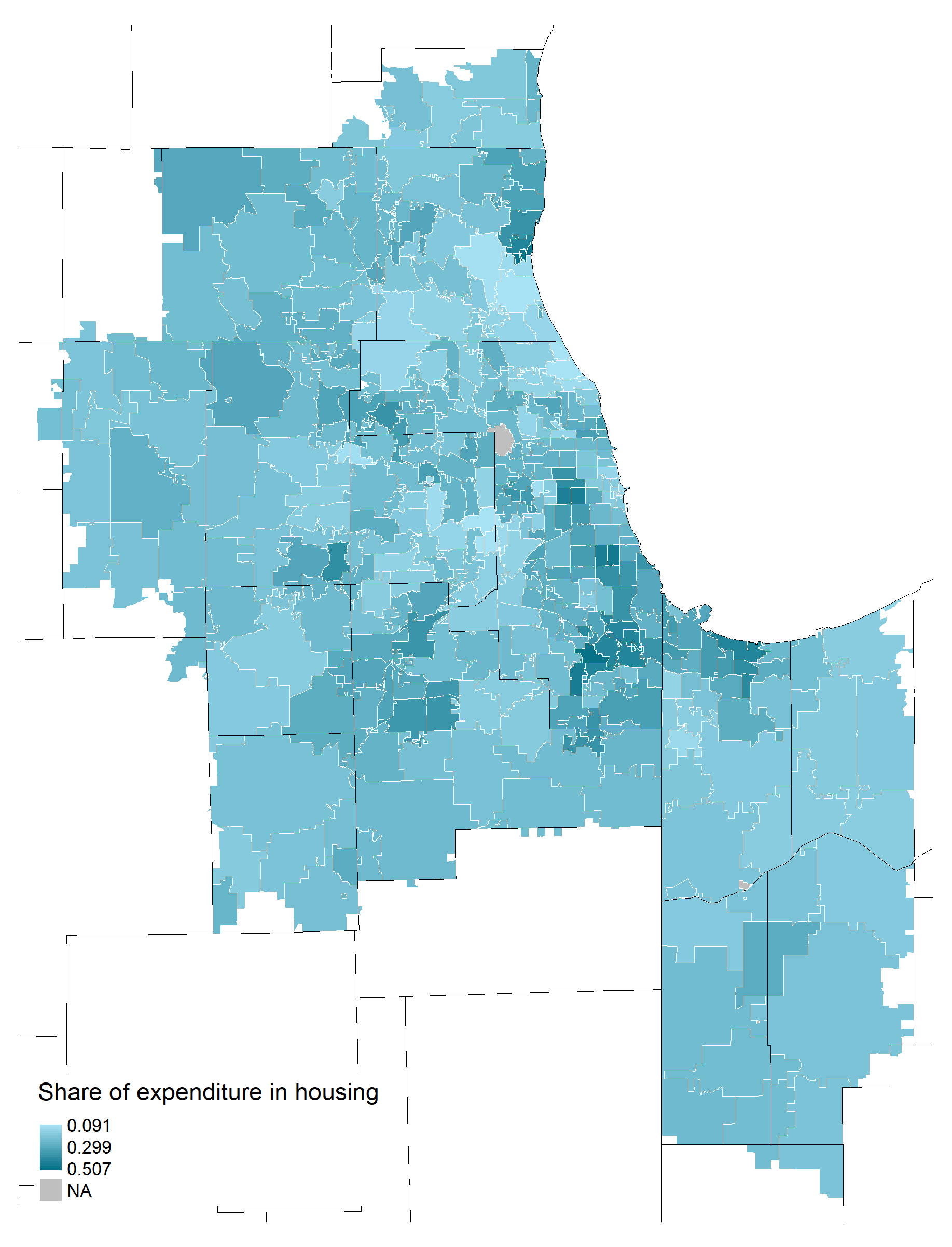}

    \begin{minipage}{.95\textwidth} \footnotesize
        \vspace{3mm}
        Notes:
        Data are from the \textcite{IRS} and the \textcite{hudSAFMR}.
        The figure shows housing expenditure shares as computed in
        Online Appendix \ref{sec:measure_housing_expenditure}, namely,
        by dividing the SAFMR 40th percentile rental value for a 2-bedroom 
        apartment by average monthly wage per household divided, both
        for 2018.
        We include ZIP codes located in the Chicago-Naperville-Elgin CBSA.
    \end{minipage}
\end{figure}

\clearpage

\begin{figure}[h!]
    \centering
    \caption{Time trends in rents according to Zillow and SAFMR}
    \label{fig:trend_zillow_safmr}

	\includegraphics[width = 0.8\textwidth]
        {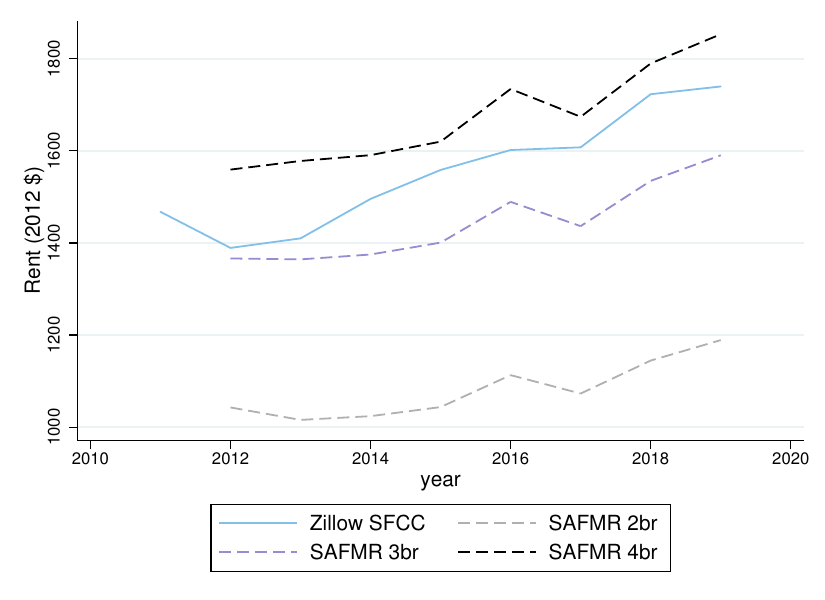}

    \begin{minipage}{.95\textwidth} \footnotesize
        \vspace{3mm}
        Notes:
        Data are from \textcite{ZillowData} and Small Area Fair Market Rents 
        (\citeyear[SAFMR]{hudSAFMR}).
        The figure compares the evolution of the median rental value in Zillow
        to three SAFMRs series, for 2, 3, and 4 or more bedrooms.
        SAFMR data generally corresponds to the 40th percentile of the
        distribution of paid rents in a given ZIP code.
        For more information on how SAFMRs are calculated, see 
        \textcite[][page 41641]{hudPreamble}.
    \end{minipage}
\end{figure}

\clearpage
\begin{figure}[h!]
    \centering
    \caption{Sample of ZIP codes in Zillow data and population density, mainland US}
    \label{fig:map_zillow_sample}

    \begin{subfigure}{1\textwidth}
        \centering
        \caption*{Zillow ZIP codes}
        \includegraphics[width = 0.95\textwidth]
            {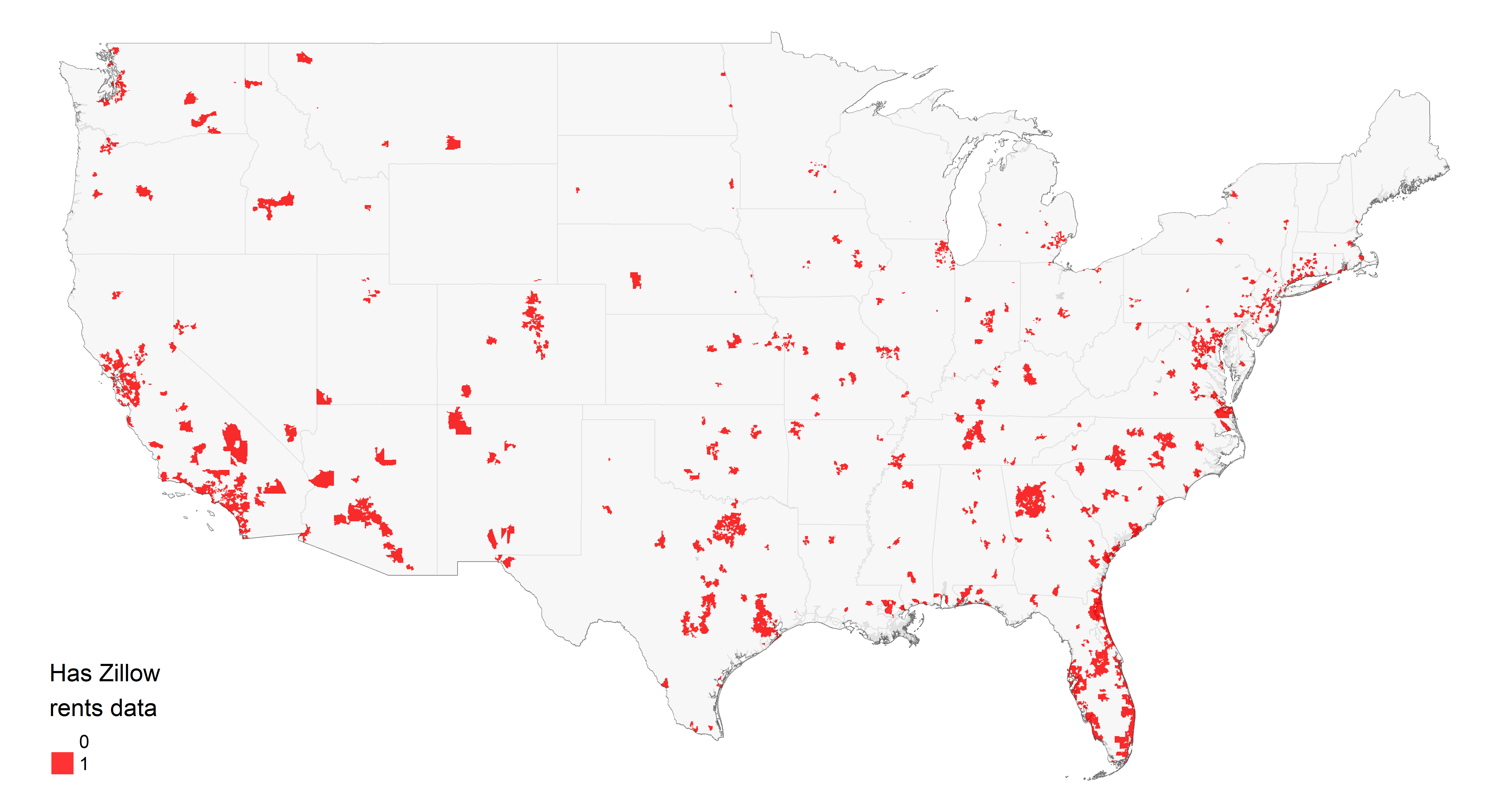}
    \end{subfigure}\\
    \begin{subfigure}{1\textwidth}
        \centering
        \caption*{Population Density}
        \includegraphics[width = 0.95\textwidth]
            {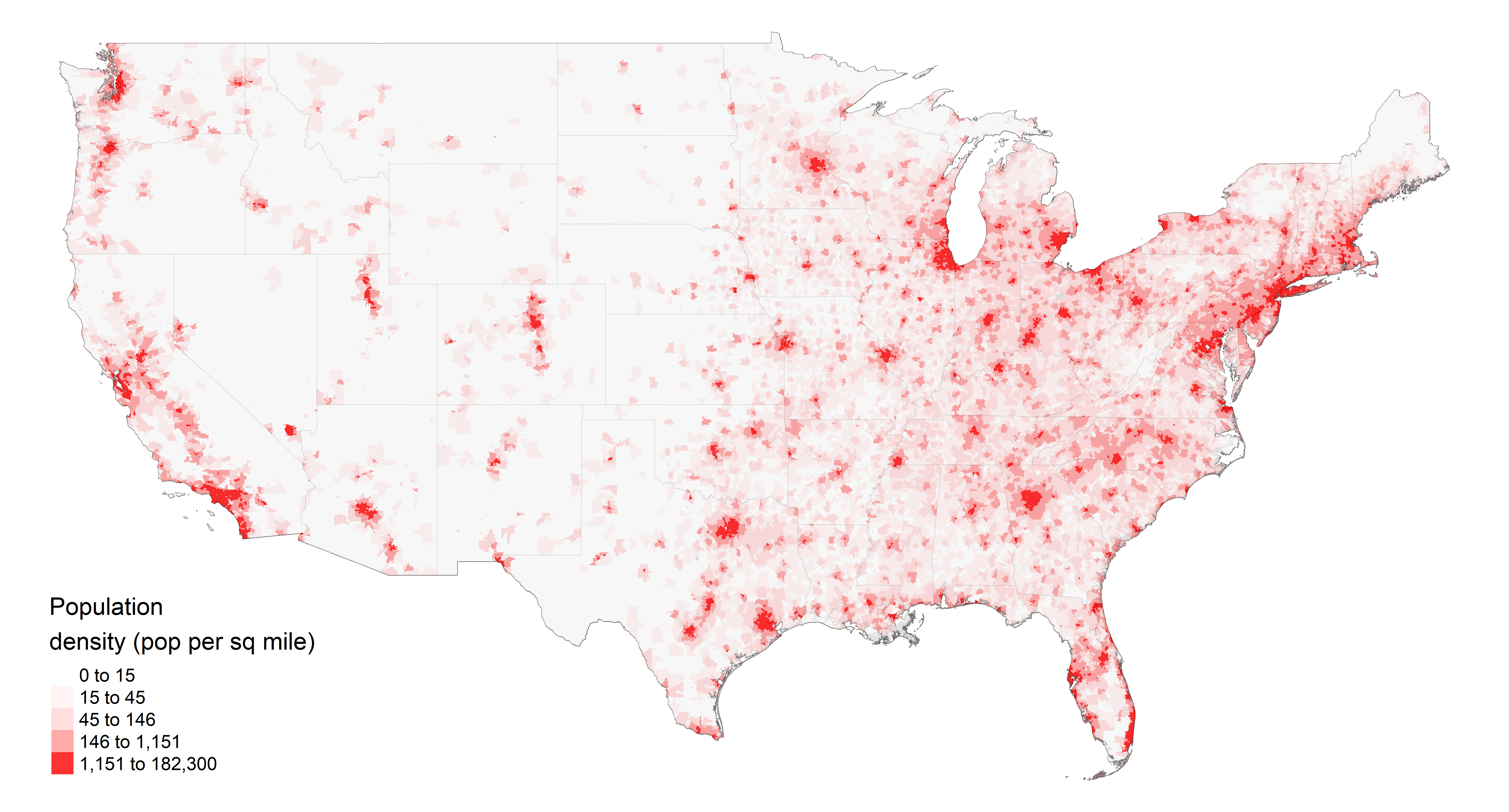}
    \end{subfigure}

    \begin{minipage}{.95\textwidth} \footnotesize
        \vspace{3mm}
        Notes:
        Data are from \textcite{ZillowData} and \textcite{ESRI}.
        The figure compares the sample of ZIP codes available in Zillow to
        population density at the ZIP code level.
        The top figure shows the sample of the ZIP codes that have rents data in 
        the SFCC category at any point in the period 2010--2019.
        The bottom figure shows quintiles of population density according to the
        2010 US Census, and measured in population per square mile.
    \end{minipage}
\end{figure}

\clearpage
\begin{figure}[h!]
    \centering
    \caption{Changes in log rents in the Chicago-Naperville-Elgin CBSA, July 2019}
    \label{fig:map_rents_chicago_jul2019}

    \includegraphics[width = 0.65\textwidth]
            {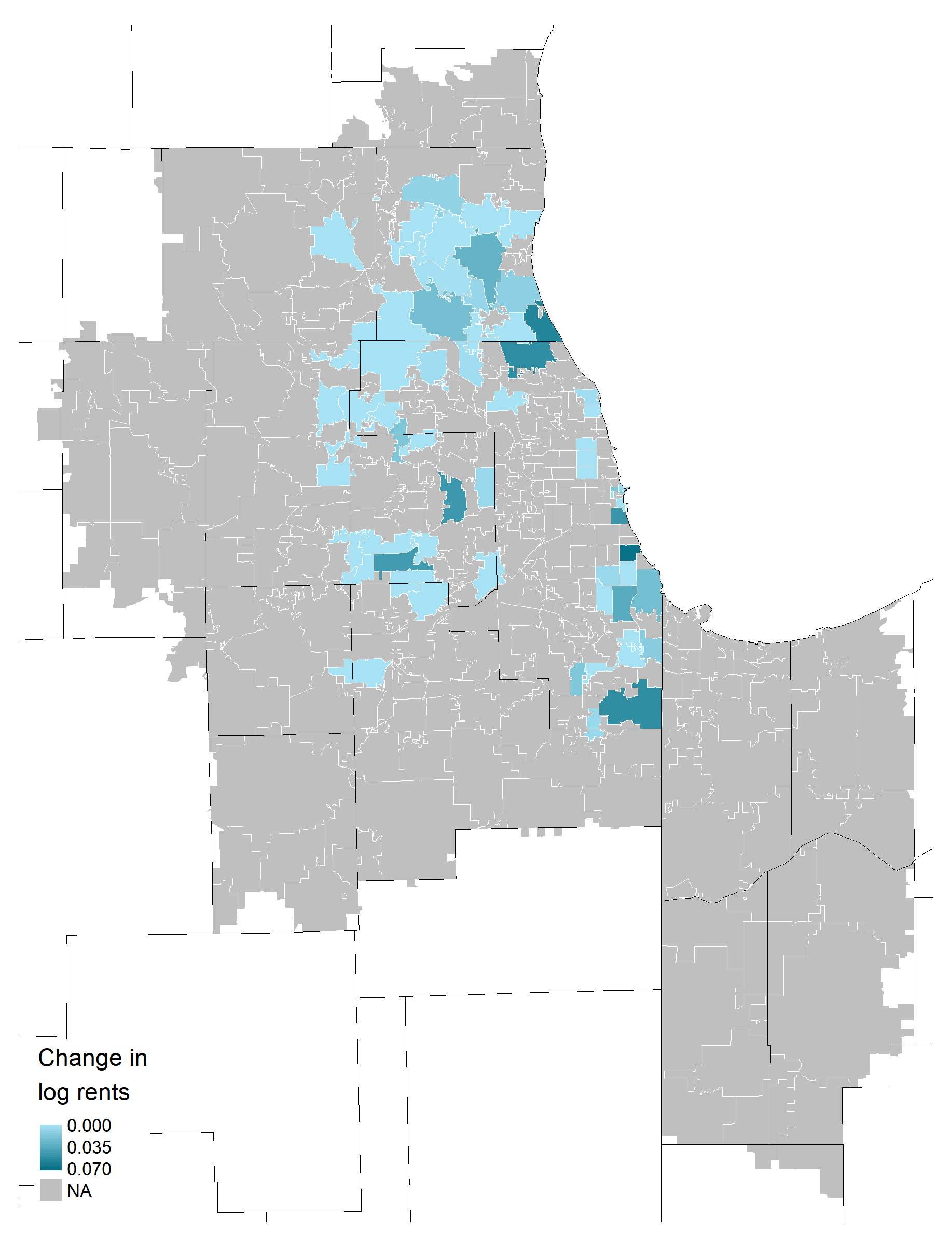}

    \begin{minipage}{.95\textwidth} \footnotesize
        \vspace{3mm}
        Notes: 
        Data are from \textcite{ZillowData}.
        The figure shows the change in the log of median rents per square foot 
        in the SFCC category in the month of June 2019 in ZIP codes located in 
        the Chicago-Naperville-Elgin CBSA.
    \end{minipage}
\end{figure}

\clearpage
\begin{figure}[h!]
    \centering
    \caption{Distribution of statutory minimum wage changes, Zillow sample}
    \label{fig:mw_changes_dist_zillow}

    \begin{subfigure}{.7\textwidth} \centering
        \caption*{Intensity}
        \includegraphics[width = 1\textwidth]
            {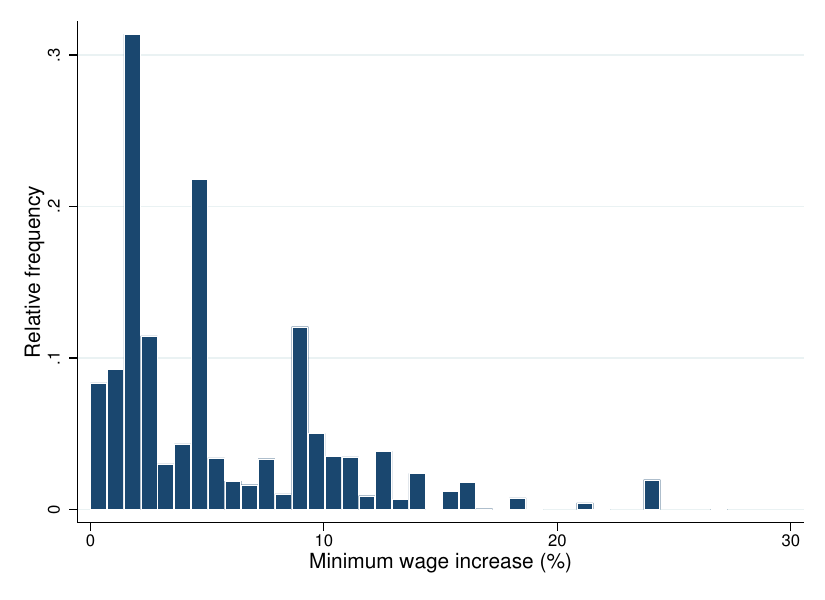}
    \end{subfigure}\\
    \begin{subfigure}{.7\textwidth} \centering
        \caption*{Timing}
        \includegraphics[width = 1\textwidth]
            {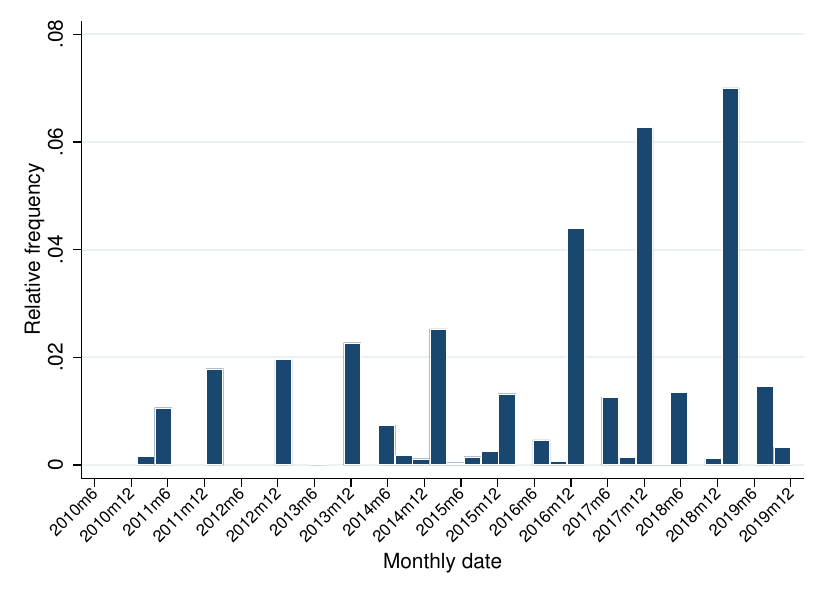}
    \end{subfigure}

    \begin{minipage}{.95\textwidth} \footnotesize
        \vspace{3mm}
        Notes:
        Data are from the MW panel described in
        Section \ref{sec:data_mw_panel}.
        The histograms show the distribution of positive MW changes in the 
        sample of ZIP codes available in the Zillow data.
        We exclude a few negative changes for expository purposes.
        The top figure (``Intensity'') reports the intensity of the changes in 
        percentage terms.
        The bottom figure (``Timing'') reports the distribution of such changes 
        over time.
    \end{minipage}
\end{figure}

\clearpage

\begin{figure}[hbt!]
    \centering
    \caption{Estimates of the effect of the minimum wage on rents, baseline sample
             including leads and lags of the residence MW}
    \label{fig:dynamic_residence}

    \includegraphics[width = .65\textwidth]
        {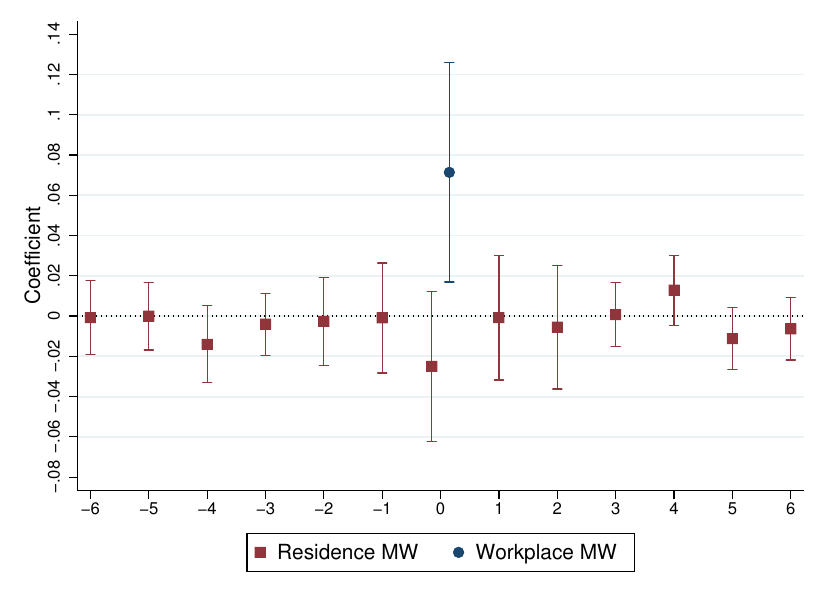}

    \begin{minipage}{.95\textwidth} \footnotesize
        \vspace{3mm}
        Notes:
        Data are from the baseline estimation sample described in Section 
        \ref{sec:data_final_panel}.
        The figure shows coefficients from regressions of the log of rents per
        square foot on the residence and workplace MW measures, including 
        six leads and lags of the residence MW.
        All regressions include time-period fixed effects and economic controls 
        that vary at the county by month and county by quarter levels.
        The measure of rents per square foot correspond to the Single Family, 
        Condominium and Cooperative houses from Zillow.
        The residence MW is defined as the log statutory MW in the same ZIP code.
        The workplace MW is defined as the log statutory MW where the average 
        resident of the ZIP code works, constructed using LODES 
        origin-destination data.
        Economic controls from the QCEW include the change of the following 
        variables: the log of the average wage, the log of employment, and 
        the log of the establishment count for the sectors 
        ``Information,'' ``Financial activities,'' and ``Professional and 
        business services.''
        95\% pointwise confidence intervals are obtained from standard errors 
        clustered at the state level.
    \end{minipage}
\end{figure}

\clearpage
\begin{figure}[htb!]
    \centering
    \caption{Relationship between log rents and the minimum wage measures, 
             sample of affected ZIP code-months}
    \label{fig:non_parametric}
    
    \begin{minipage}{.95\textwidth} \centering
        Panel A: Raw data
        \vspace{1mm}
    \end{minipage}
    \begin{subfigure}{0.5\textwidth}
        \caption*{Residence MW}
        \includegraphics[width = 1\textwidth]{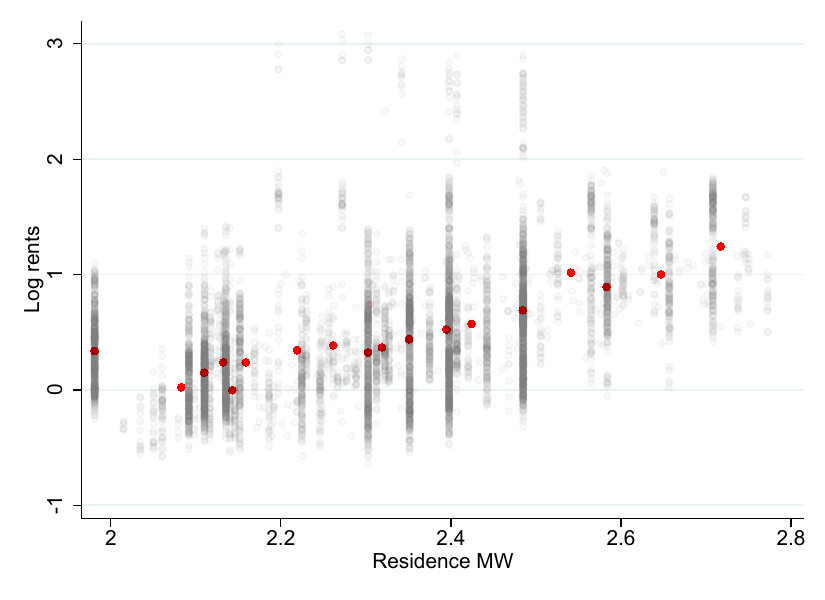}
    \end{subfigure}%
    \begin{subfigure}{0.5\textwidth}
        \caption*{Workplace MW}
        \includegraphics[width = 1\textwidth]{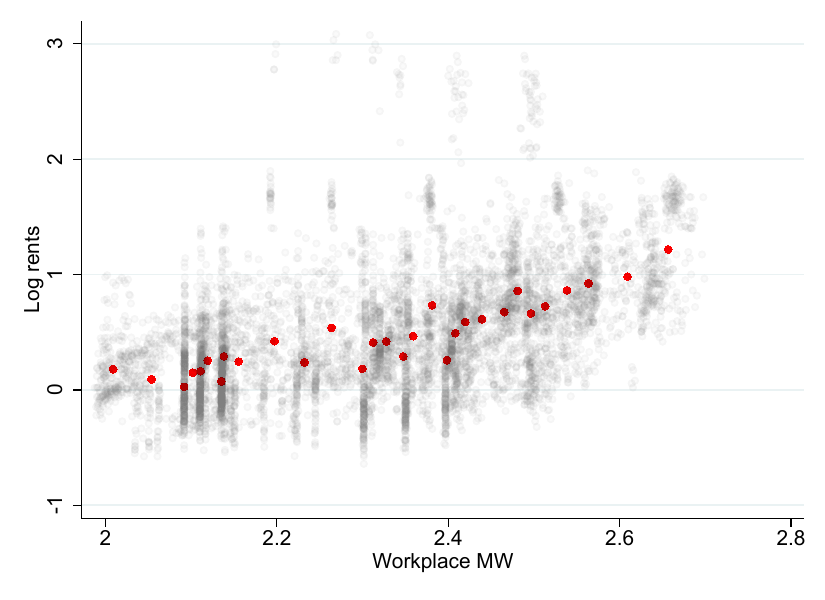}
    \end{subfigure}\\

    \vspace{2mm}
    \begin{minipage}{.95\textwidth} \centering
        Panel B: Conditional on ZIP code FE and the other MW measure
        \vspace{1mm}
    \end{minipage}
    \begin{subfigure}{0.5\textwidth}
        \caption*{Residence MW}
        \includegraphics[width = 1\textwidth]{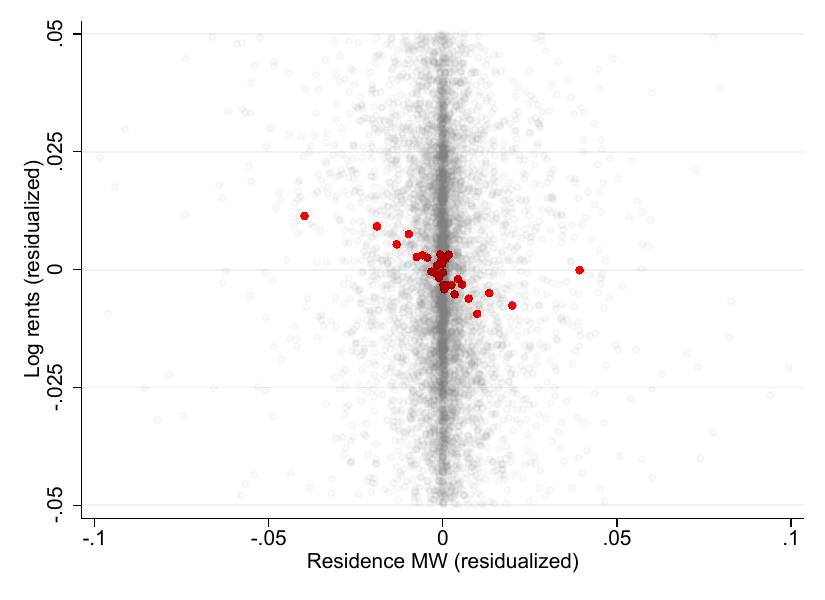}
    \end{subfigure}%
    \begin{subfigure}{0.5\textwidth}
        \caption*{Workplace MW}
        \includegraphics[width = 1\textwidth]{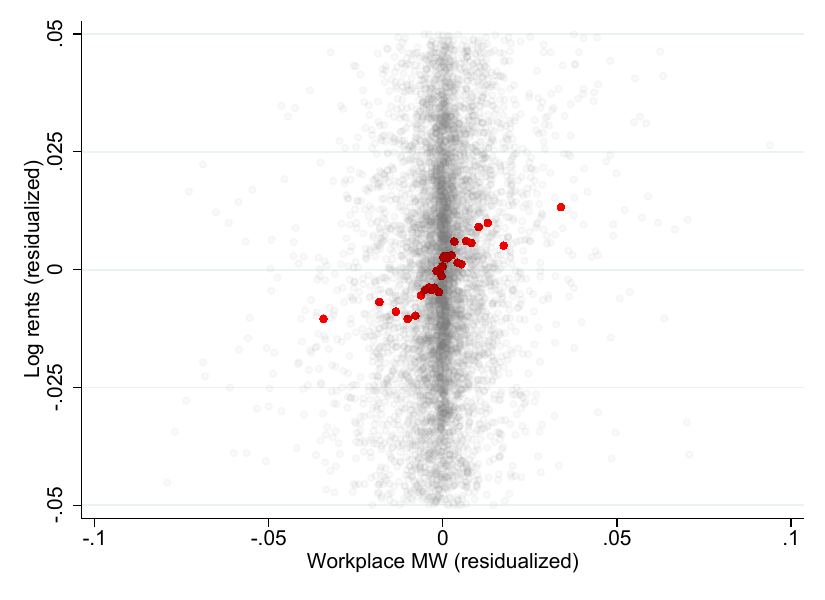}
    \end{subfigure}

    \begin{minipage}{.95\textwidth} \footnotesize
        \vspace{3mm}
        Notes:
        Data are from Zillow and LODES.
        The plot shows the unconditional and conditional relationship between 
        log rents and the MW measures.
        The sample is composed of ZIP code-month observations located in CBSAs 
        where there was some statutory MW increase in the month of interest. 
        The rents variable correspond to log rents per square foot in the SFCC 
        category in Zillow.
        The workplace MW measure is constructed using commuting data from the 
        closest prior year.
        Panel A shows the raw relationship between log rents and workplace 
        and residence MW levels.
        Panel B shows the same relationship using residuals from regressions 
        on ZIP code indicators and 100 indicators of the other MW measure.
        Red dots correspond to 30 equally-sized bins of the $x$-axis variable.
        Gray dots correspond to all data points in Panel A, and only those 
        data points that fall within the range of the plot axes in Panel B.
    \end{minipage}
\end{figure}

\clearpage
\begin{figure}[h!]
    \centering
    \caption{Residualized changes in the workplace minimum wage and log rents,
             Chicago-Naperville-Elgin CBSA on July 2019}
    \label{fig:map_residuals_chicago_jul2019}

    \begin{subfigure}{0.5\textwidth}
        \centering
        \caption*{Residualized change in wkp.\ MW}
        \includegraphics[width = 1\textwidth]
            {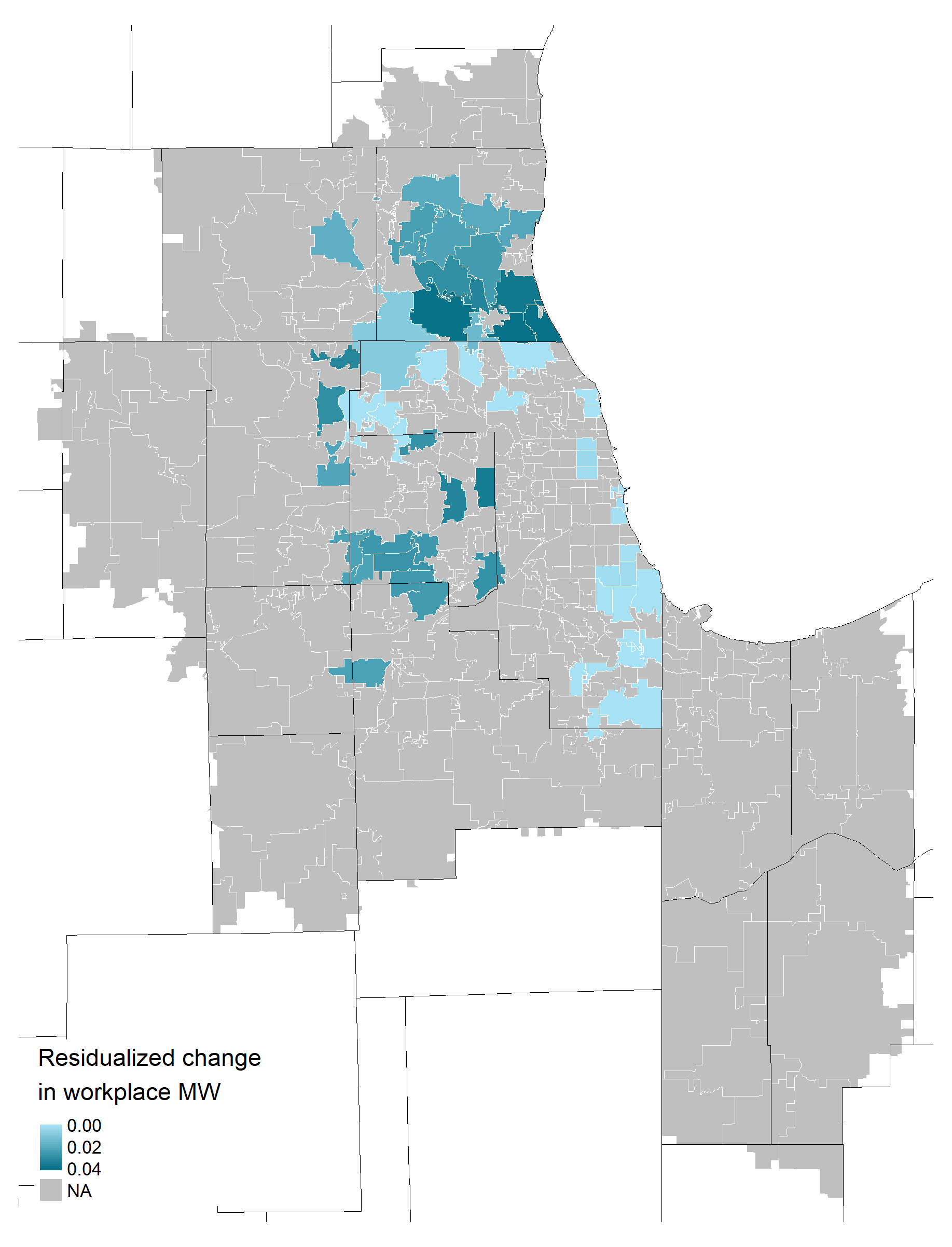}
    \end{subfigure}%
    \begin{subfigure}{0.5\textwidth}
        \centering
        \caption*{Residualized change in log rents}
        \includegraphics[width = 1\textwidth]
            {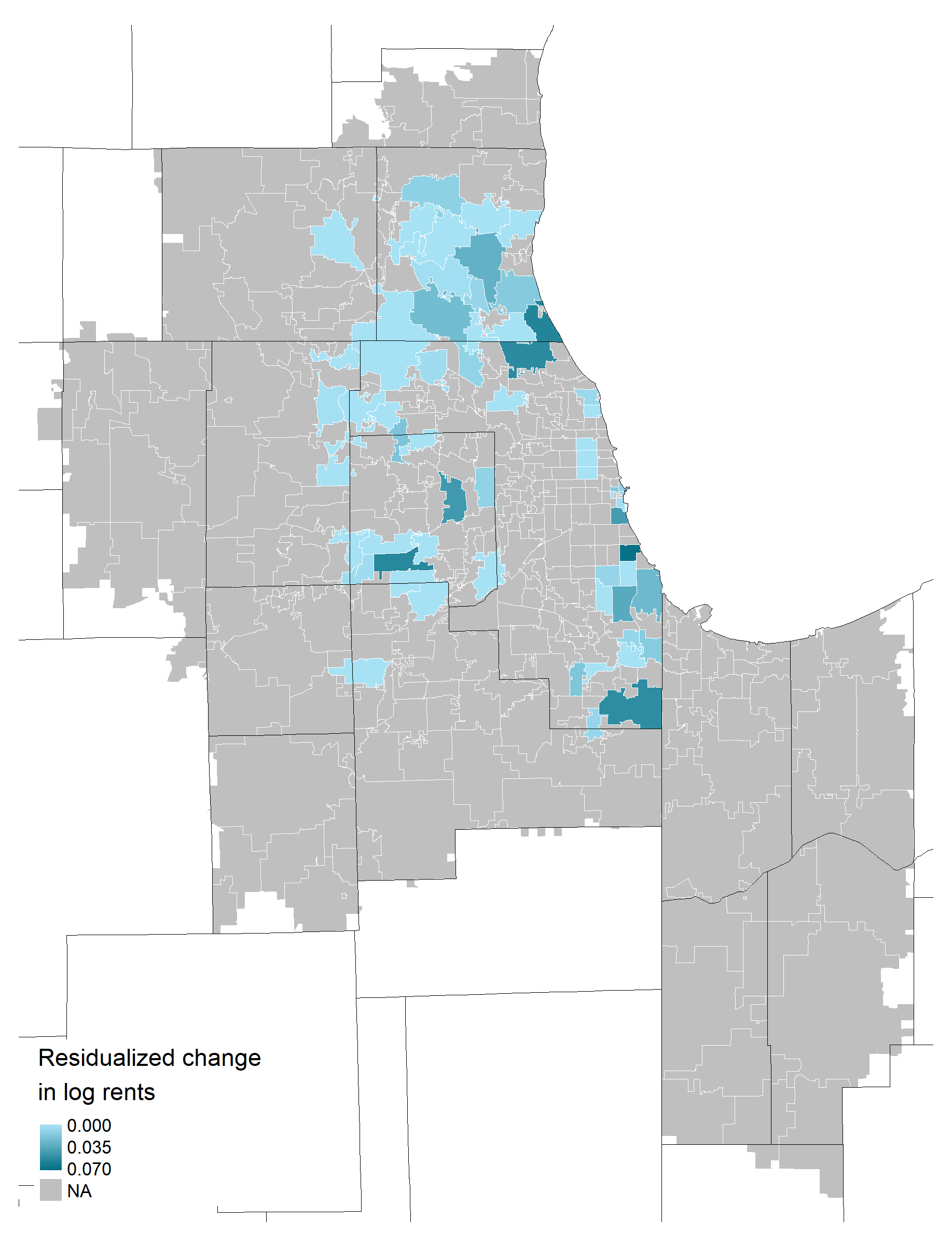}
    \end{subfigure}

    \begin{minipage}{.95\textwidth} \footnotesize
        \vspace{3mm}
        Notes: 
        Data are from the unbalanced estimation panel described in Section
        \ref{sec:data_final_panel}.
        The left figure maps the residuals of a regression of the change in the 
        workplace MW measure on the change in the residence MW measure, 
        including economic controls and year-month fixed effects.
        The right figure maps the residuals of a regression of the change in log 
        rents on economic controls and year-month fixed effects.
        The residence MW is defined as the log statutory MW in the same ZIP code.
        The workplace MW is defined as the statutory MW where the average 
        resident of the ZIP code works, constructed using LODES 
        origin-destination data.
        Economic controls from the QCEW include the change of the following 
        variables: the log of the average wage, the log of employment, and the 
        log of the establishment count for the sectors ``Information,''
        ``Financial activities,'' and ``Professional and business services.''
    \end{minipage}
\end{figure}

\clearpage

\begin{figure}[h!]
    \centering
    \caption{Estimates of the effect of the minimum wage on rents, 
             Zillow rental index}
    \label{fig:dynamic_zori}

    \begin{subfigure}{.65\textwidth}
        \caption{Control for year-month by CBSA fixed effects}
        \includegraphics[width = 1\textwidth]
            {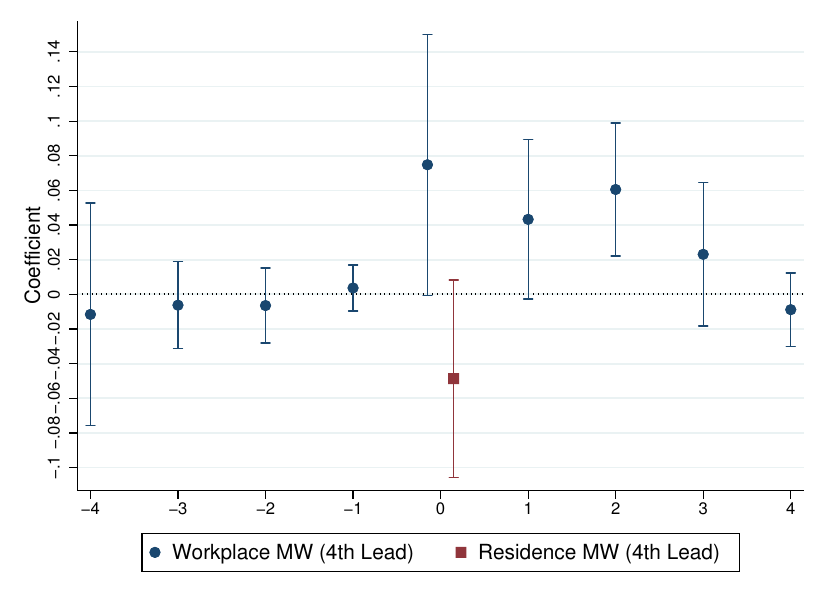}
    \end{subfigure}\\
    \begin{subfigure}{.65\textwidth}
        \caption{Control for year-month fixed effects}
        \includegraphics[width = 1\textwidth]
            {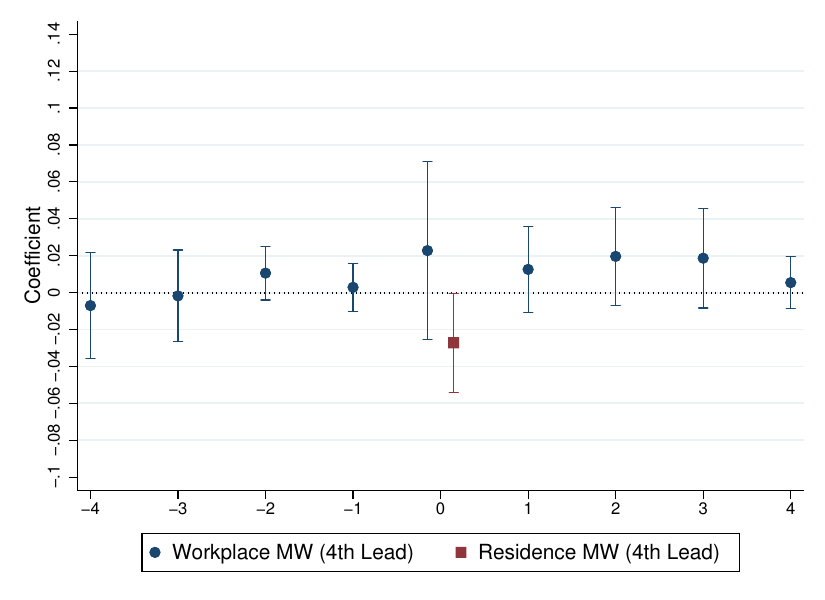}
    \end{subfigure}

    \begin{minipage}{.95\textwidth} \footnotesize
        \vspace{3mm}
        Notes:
        Data are from the baseline estimation sample described in Section 
        \ref{sec:data_final_panel}.
        The figures show coefficients from regressions of the change in 
        log of Zillow rental index on leads and lags of the change in the 
        workplace MW and the change in the residence MW, using the 
        4th lead of the MW-based measures.
        The top panel includes CBSA by year-month fixed effects, whereas
        the bottom panel includes year-month fixed effects.
        Both regressions include economic controls that vary at the county by 
        month and county by quarter levels, which include the change of the 
        following variables: the log of the average wage, the log of employment, 
        and the log of the establishment count for the sectors 
        ``Information,'' ``Financial activities,'' and ``Professional and 
        business services.''
        95\% pointwise confidence intervals are obtained from standard errors 
        clustered at the state level.
    \end{minipage}
\end{figure}

\clearpage

\begin{figure}[h!]
    \centering
    \caption{Estimates of the effect of the minimum wage on rents, county by month data}
    \label{fig:dynamic_county_month}

	\includegraphics[width = 0.75\textwidth]{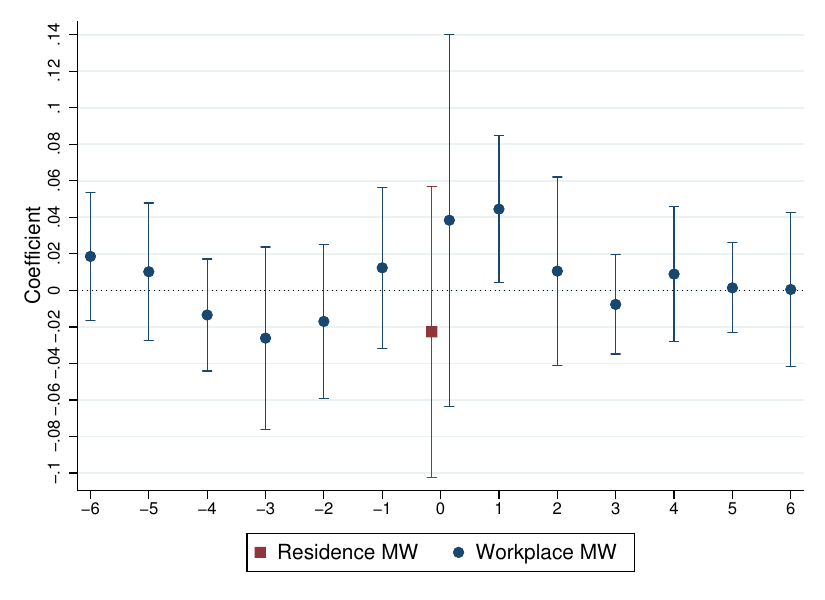}

    \begin{minipage}{.95\textwidth} \footnotesize
        \vspace{3mm}
        Notes:
        Data are from the county-by-month panel described in Section 
        \ref{sec:data_final_panel}.
        We plot coefficients from regressions of the log of rents per
        square foot on the residence MW and workplace MW, including 
        six leads and lags of the workplace MW measure.
        All regressions are estimated in first differences and include 
        time-period fixed effects and economic controls that vary at the 
        county by month and county by quarter levels.
        The measure of rents per square foot corresponds to the Single Family, 
        Condominium and Cooperative houses from Zillow.
        The residence MW is defined as the log statutory MW at the county.
        The workplace MW is defined as the log statutory MW where the average 
        resident of the county works, constructed using LODES 
        origin-destination data.
        Economic controls from the QCEW include the change of the following 
        variables: the log of the average wage, the log of employment, and the 
        log of the establishment count for the sectors ``Information,'' 
        ``Financial activities,'' and ``Professional and business services.''
        95\% pointwise confidence intervals are obtained from standard errors 
        clustered at the state level.
    \end{minipage}
\end{figure}

\clearpage

\begin{figure}[h!]
    \centering
    \caption{Estimates of the effect of the minimum wage on rents, stacked sample including
             leads and lags}
    \label{fig:dynamic_stacked}

    \includegraphics[width = 0.75\textwidth]{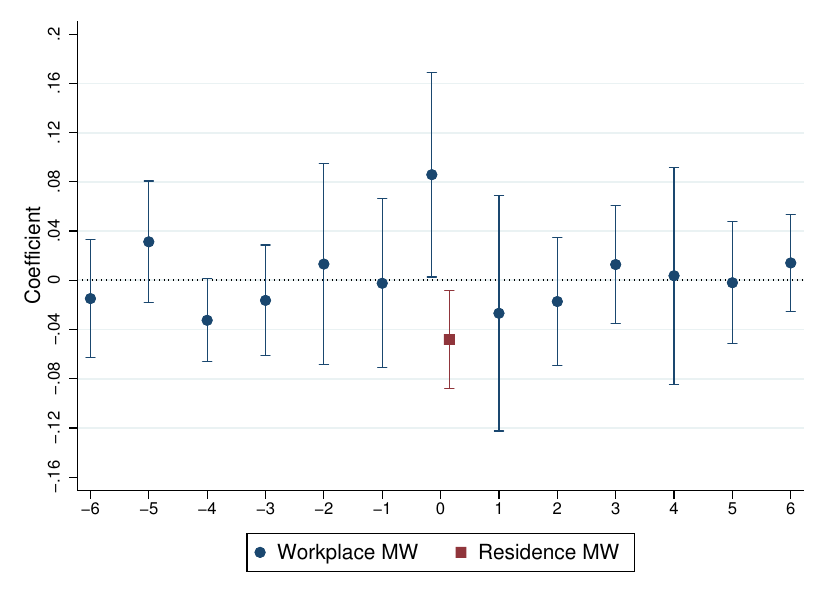}

    \begin{minipage}{.95\textwidth} \footnotesize
        \vspace{3mm}
        Notes:
        Data are from Zillow,
        the MW panel described in Section \ref{sec:data_mw_panel},
        LODES origin-destination statistics,
        and the QCEW.
        The figure mimics estimates in Figure \ref{fig:dynamic_workplace} 
        using a ``stacked'' sample.
        We construct the sample as explained in Online Appendix Table 
        \ref{tab:stacked_w6}.
        95\% pointwise confidence intervals are obtained from standard errors 
        clustered at the state level.
    \end{minipage}
\end{figure}

\clearpage
\begin{figure}[h!]
    \centering
    \caption{Distribution of changes in minimum wage measures under a 
             counterfactual federal minimum wage of \$9, urban ZIP codes}
    \label{fig:cf_hist_res_and_wkp_mw}
    
    \begin{subfigure}{0.5\textwidth}
        \caption*{Residence MW}
        \includegraphics[width = 1\textwidth]{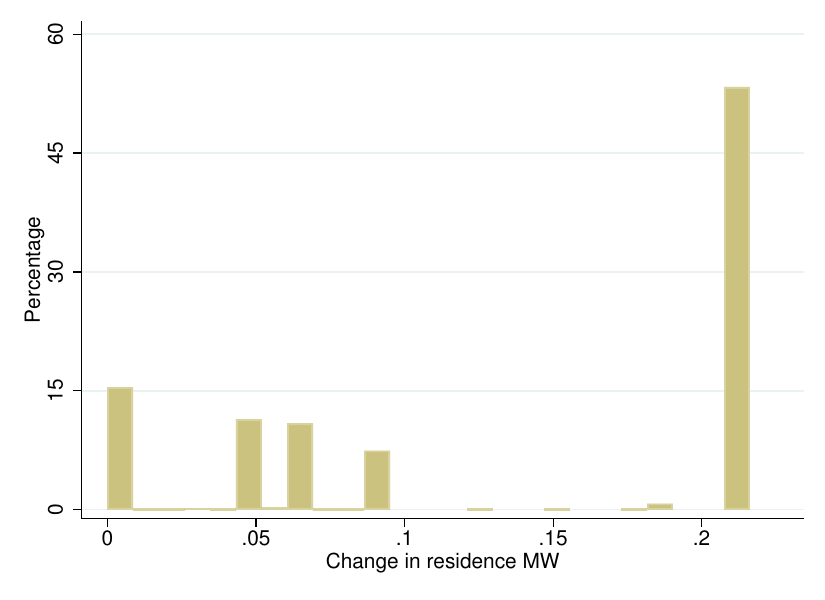}
    \end{subfigure}%
    \begin{subfigure}{0.5\textwidth}
        \caption*{Workplace MW}
        \includegraphics[width = 1\textwidth]{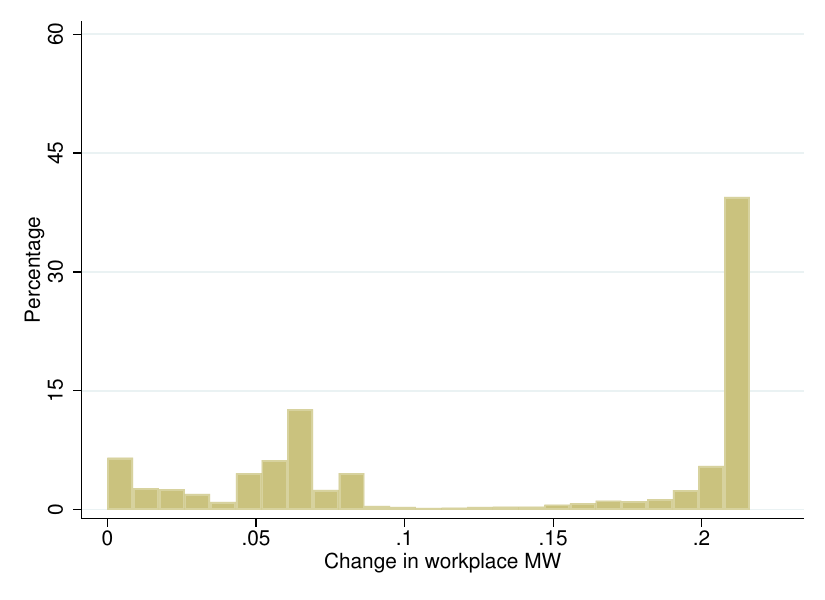}
    \end{subfigure}

    \begin{minipage}{.95\textwidth} \footnotesize
        \vspace{3mm}
        Notes:
        Data are from LODES and the MW panel described in Section
        \ref{sec:data_mw_panel}.
        The figures show the distribution of changes in the residence and 
        workplace MW measures generated by a counterfactual increase to \$9 
        in the federal MW in January 2020, holding constant other MW policies 
        in their December 2019 levels.
        The unit of observation is the urban ZIP code, where we define a ZIP code 
        as urban if it belongs to a CBSA with at least 80\% of its population 
        classified as urban by the 2010 Census.
        We exclude ZIP codes located in CBSAs where the estimated increase in 
        income was higher than 0.1.
    \end{minipage}
\end{figure}

\clearpage
\begin{figure}[h!]
    \centering
    \caption{Changes in the minimum wage measures under counterfactual
             minimum wage policies, Chicago-Naperville-Elgin CBSA}
    \label{fig:map_chicago_cf_wkp_res}
    
    \begin{minipage}{.95\textwidth} \centering
        Panel A: Increase in federal MW from \$7.25 to \$9
        \vspace{1.5mm}
    \end{minipage}

    \begin{subfigure}{.4\textwidth}  \centering
        \caption*{Changes in residence MW}
        \includegraphics[width = 1\textwidth]
            {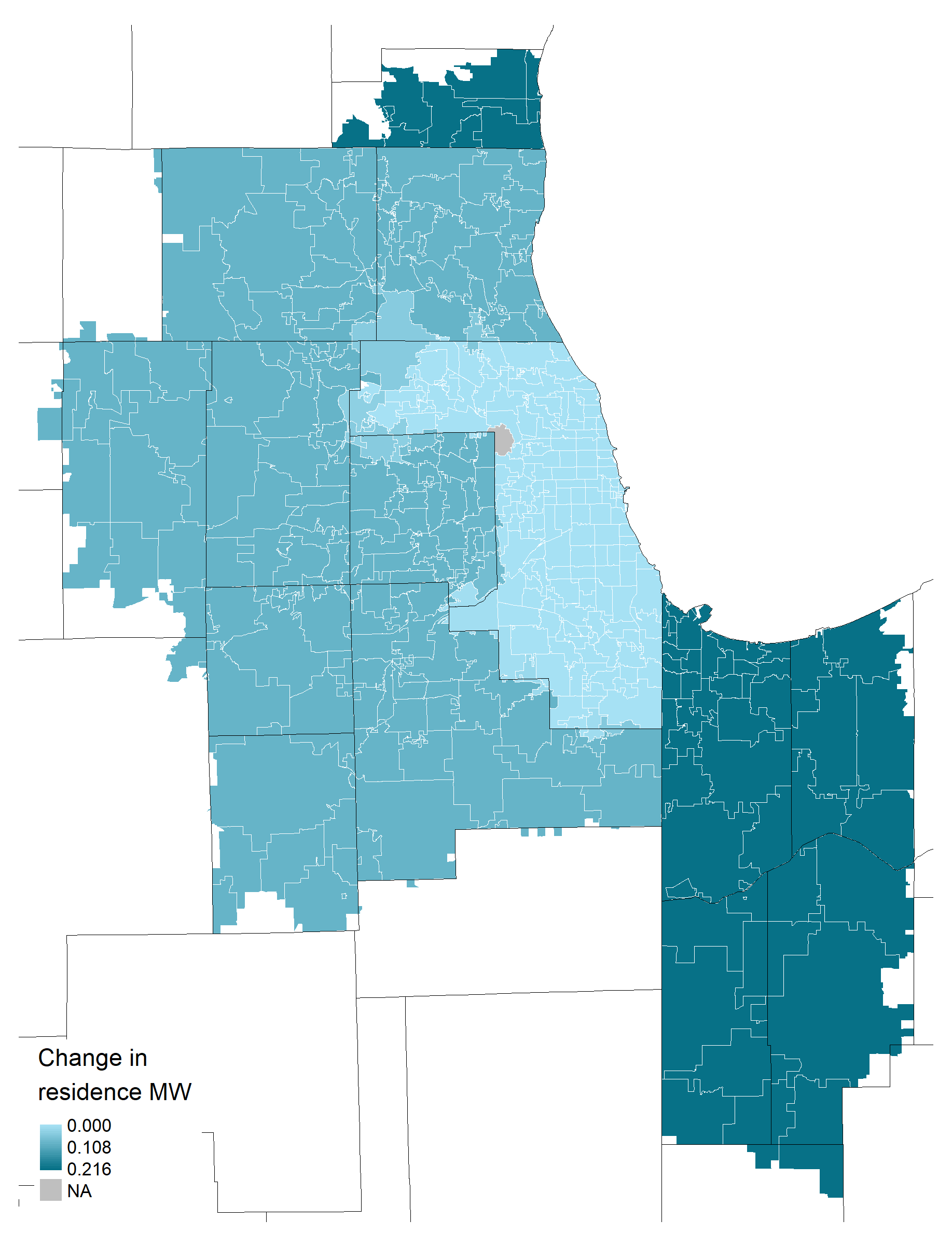}
    \end{subfigure}%
    $\quad\quad\quad\quad$%
    \begin{subfigure}{.4\textwidth}  \centering
        \caption*{Changes in workplace MW}
        \includegraphics[width = 1\textwidth]
            {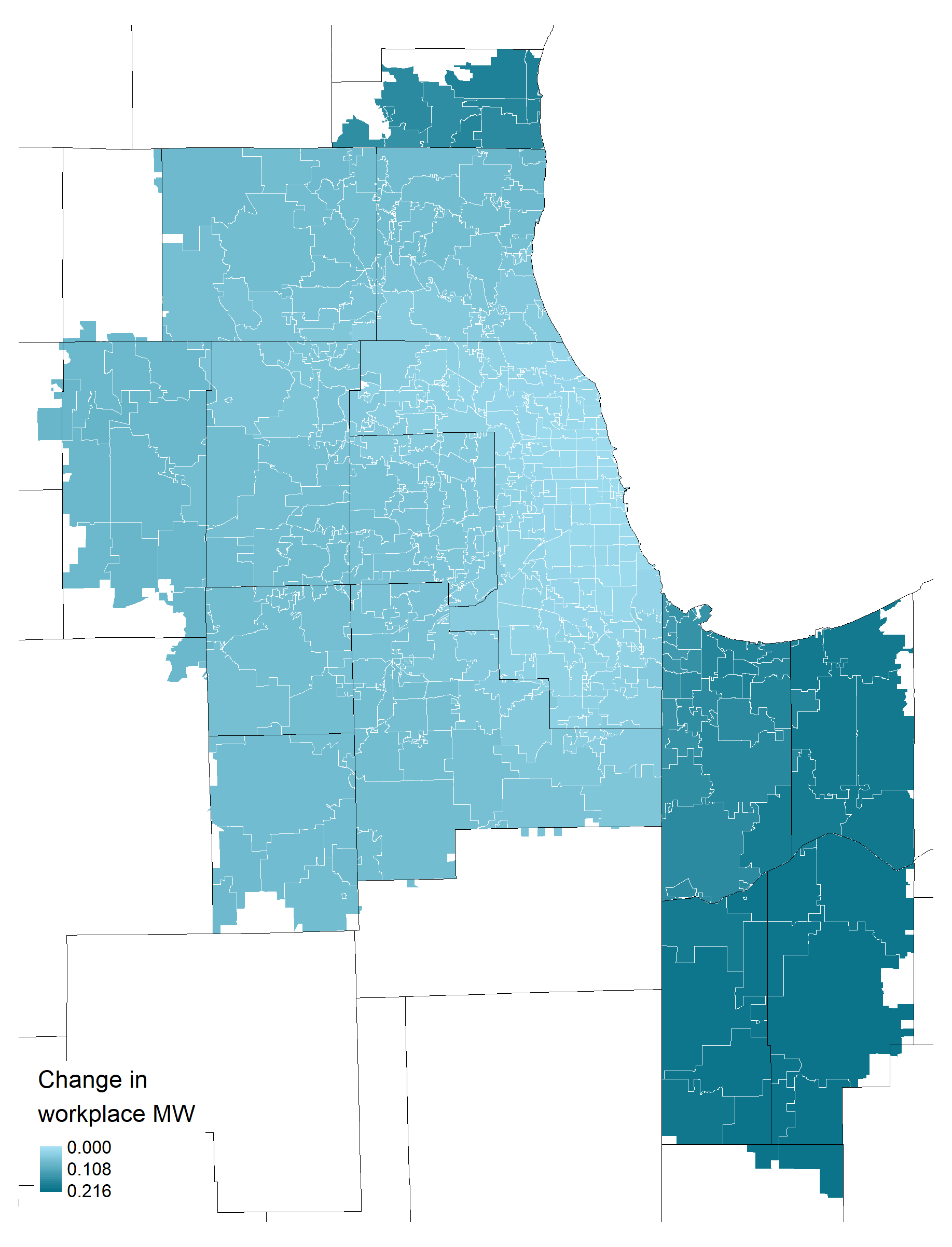}
    \end{subfigure}

    \begin{minipage}{.95\textwidth} \centering
        \vspace{2mm}
        Panel B: Increase in Chicago MW from \$13 to \$14
        \vspace{1.5mm}
    \end{minipage}

    \begin{subfigure}{.4\textwidth}  \centering
        \caption*{Changes in residence MW}
        \includegraphics[width = 1\textwidth]
            {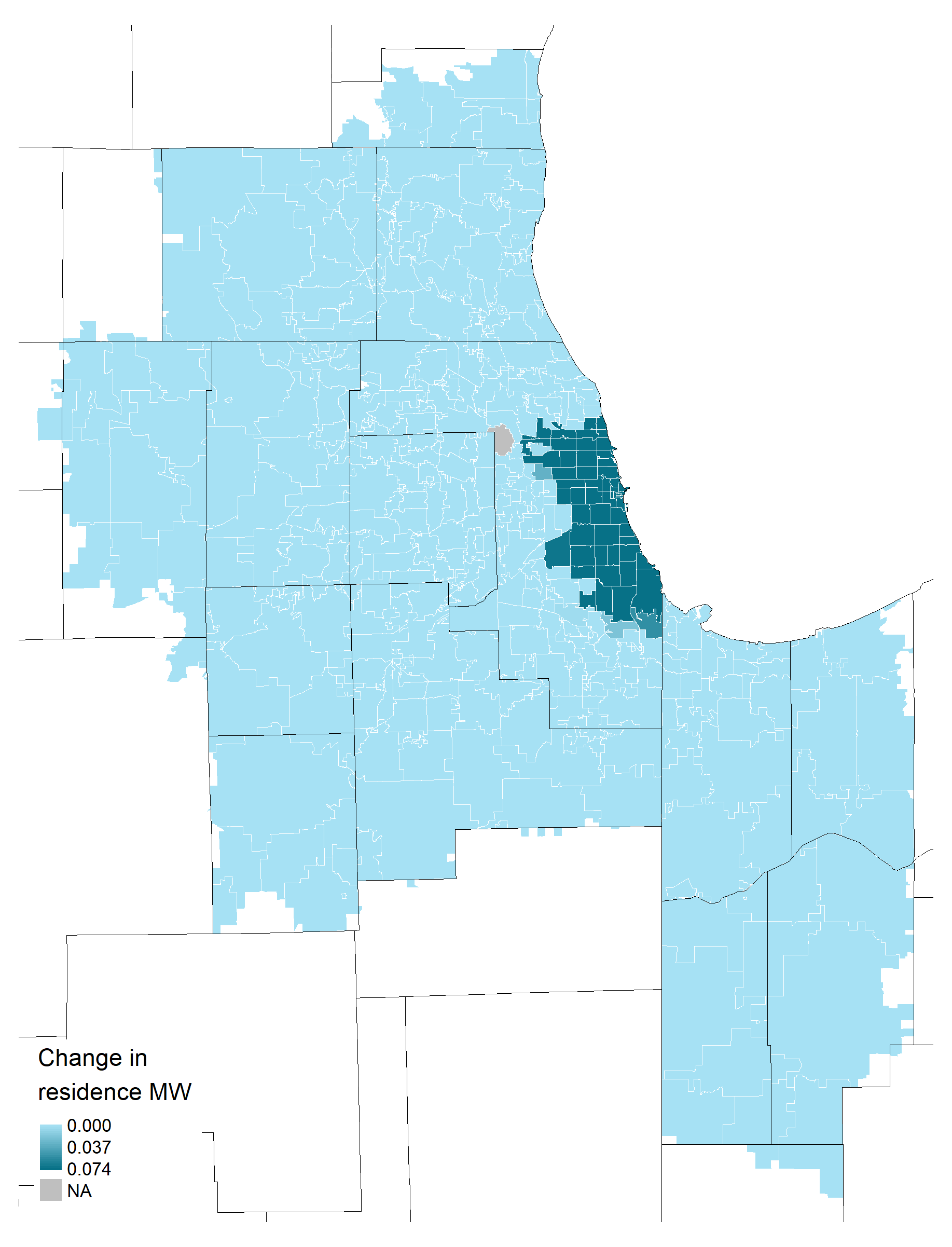}
    \end{subfigure}%
    $\quad\quad\quad\quad$%
    \begin{subfigure}{.4\textwidth}  \centering
        \caption*{Changes in workplace MW}
        \includegraphics[width = 1\textwidth]
            {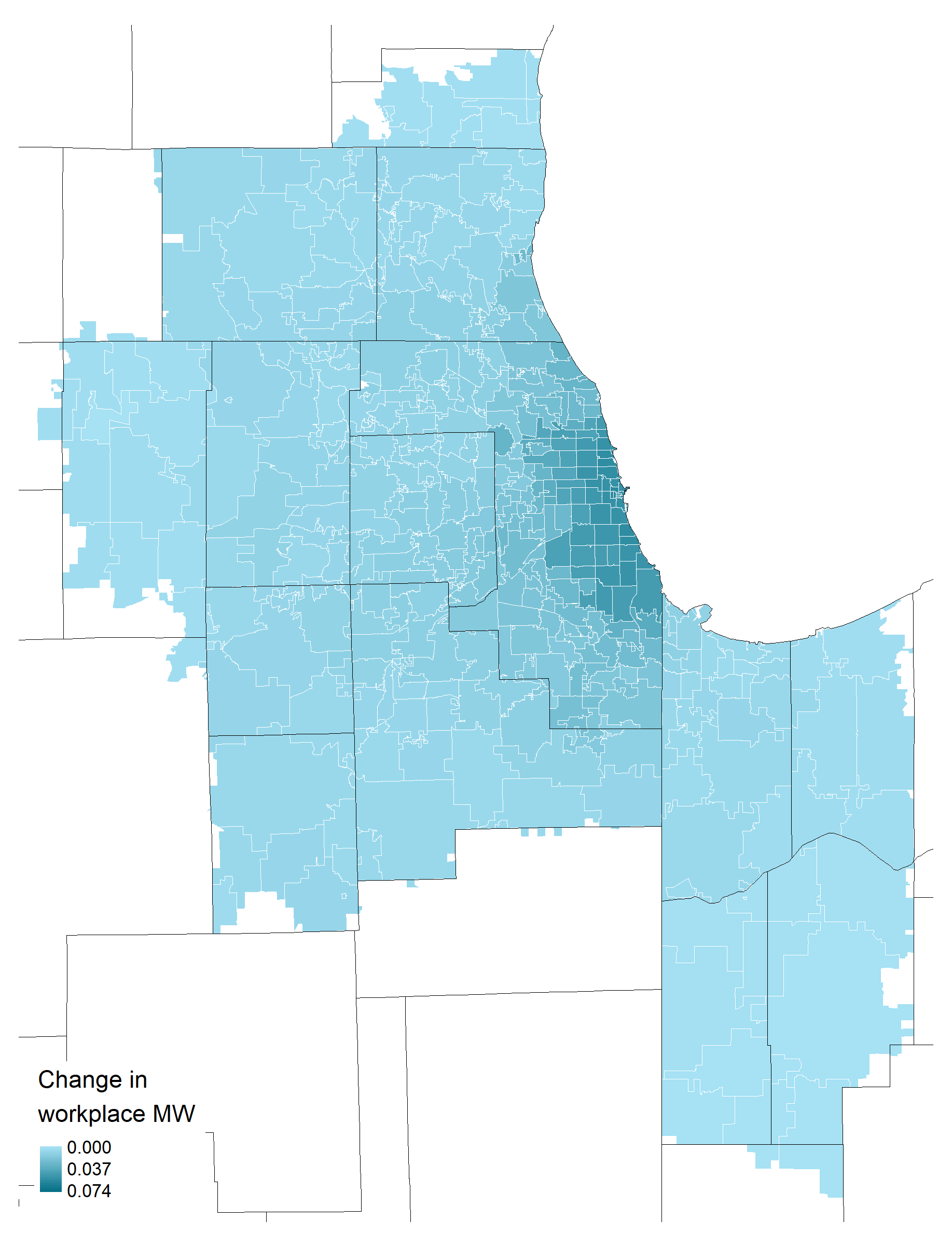}
    \end{subfigure}

    \begin{minipage}{.95\textwidth} \footnotesize
        \vspace{2.5mm}
        Notes:
        Data are from the MW panel described in Section \ref{sec:data_mw_panel} 
        and from LODES.
        The figures map changes in the residence and workplace MW measures 
        by counterfactual MW policies in the Chicago-Naperville-Elgin CBSA.
        Panel A shows a policy where the federal MW increases from \$7.25 to \$9 
        in January 2020, holding constant other MW policies at their December 
        2019 levels.
        Panel B shows a policy where the city of Chicago increases its MW 
        from \$13 to \$14 in January 2020, holding constant other MW policies 
        at their December 2019 levels as well.
    \end{minipage}
\end{figure}

\clearpage
\begin{figure}[h!]
    \centering
    \caption{Changes in log rents and log total wages under counterfactual
             minimum wage policies, Chicago-Naperville-Elgin CBSA}
    \label{fig:map_chicago_cf_rents_wages}

    \begin{minipage}{.95\textwidth} \centering
        Panel A: Increase in federal MW from \$7.25 to \$9
        \vspace{1mm}
    \end{minipage}
    
    \begin{subfigure}{.4\textwidth}  \centering
        \caption*{Changes in log rents}
        \includegraphics[width = 1\textwidth]
            {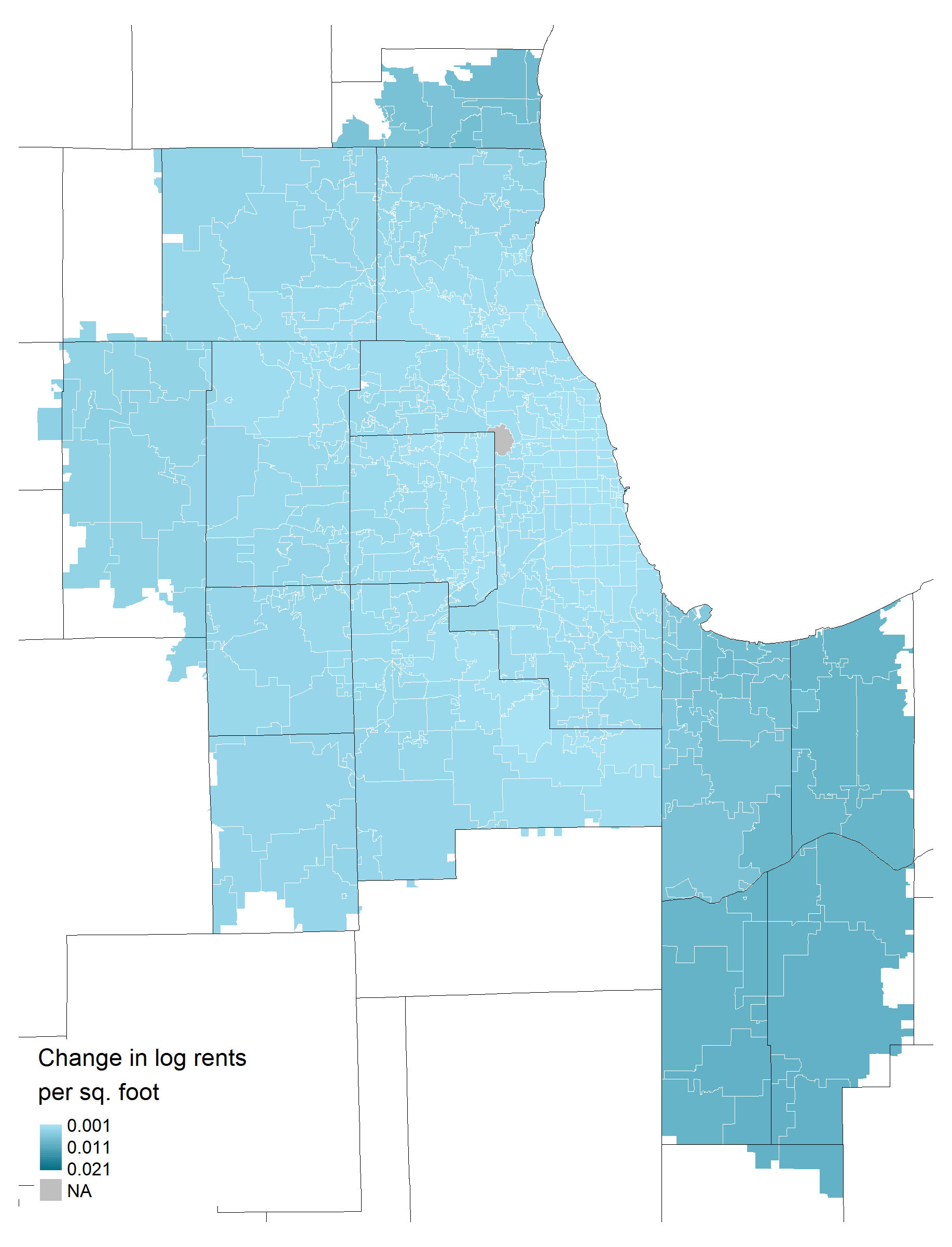}
    \end{subfigure}%
    $\quad\quad\quad\quad$%
    \begin{subfigure}{.4\textwidth}  \centering
        \caption*{Changes in log total wages}
        \includegraphics[width = 1\textwidth]
            {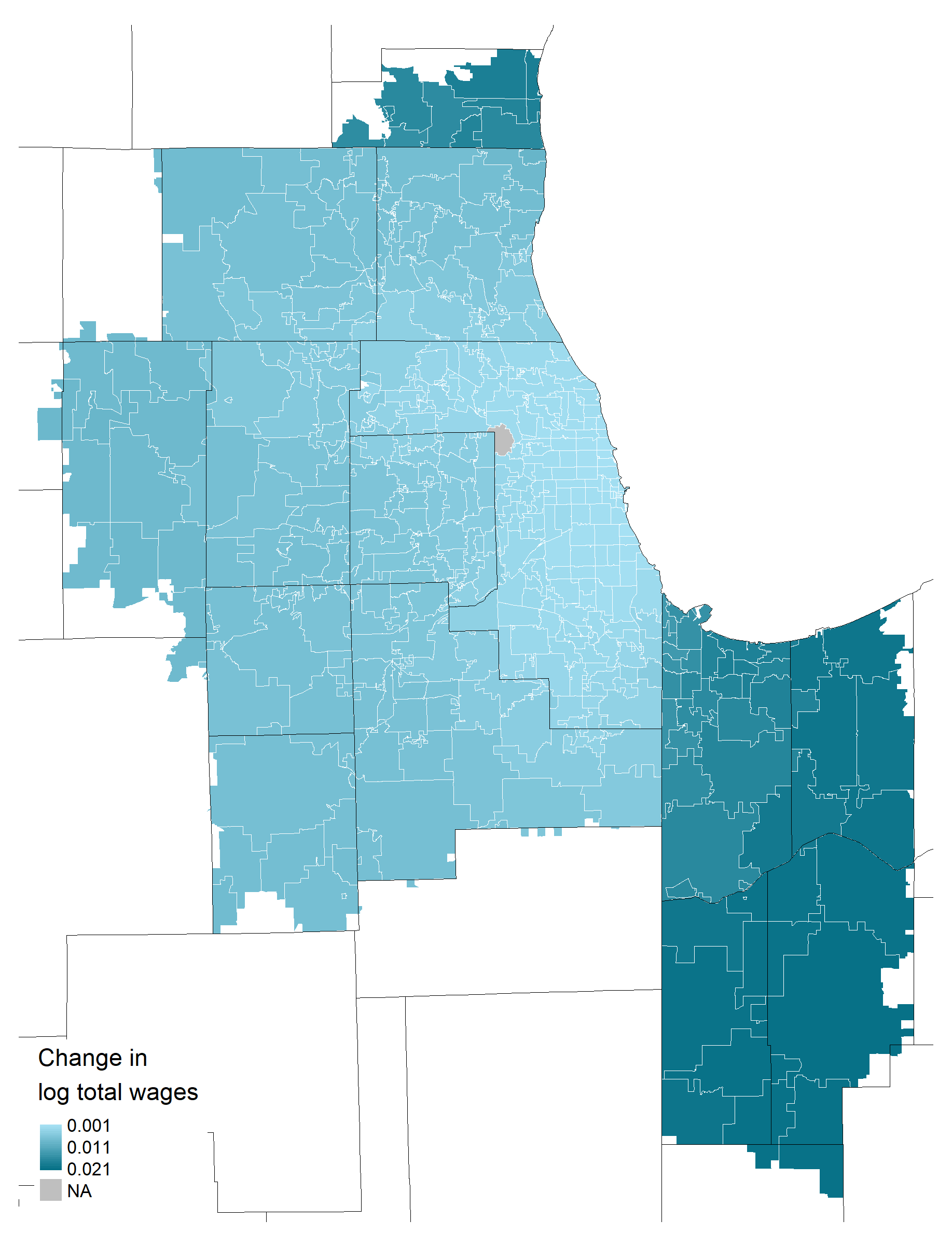}
    \end{subfigure}

    \begin{minipage}{.95\textwidth} \centering
        \vspace{2mm}
        Panel B: Increase in Chicago MW from \$13 to \$14
        \vspace{1mm}
    \end{minipage}
    
    \begin{subfigure}{.4\textwidth}  \centering
        \caption*{Changes in log rents}
        \includegraphics[width = 1\textwidth]
            {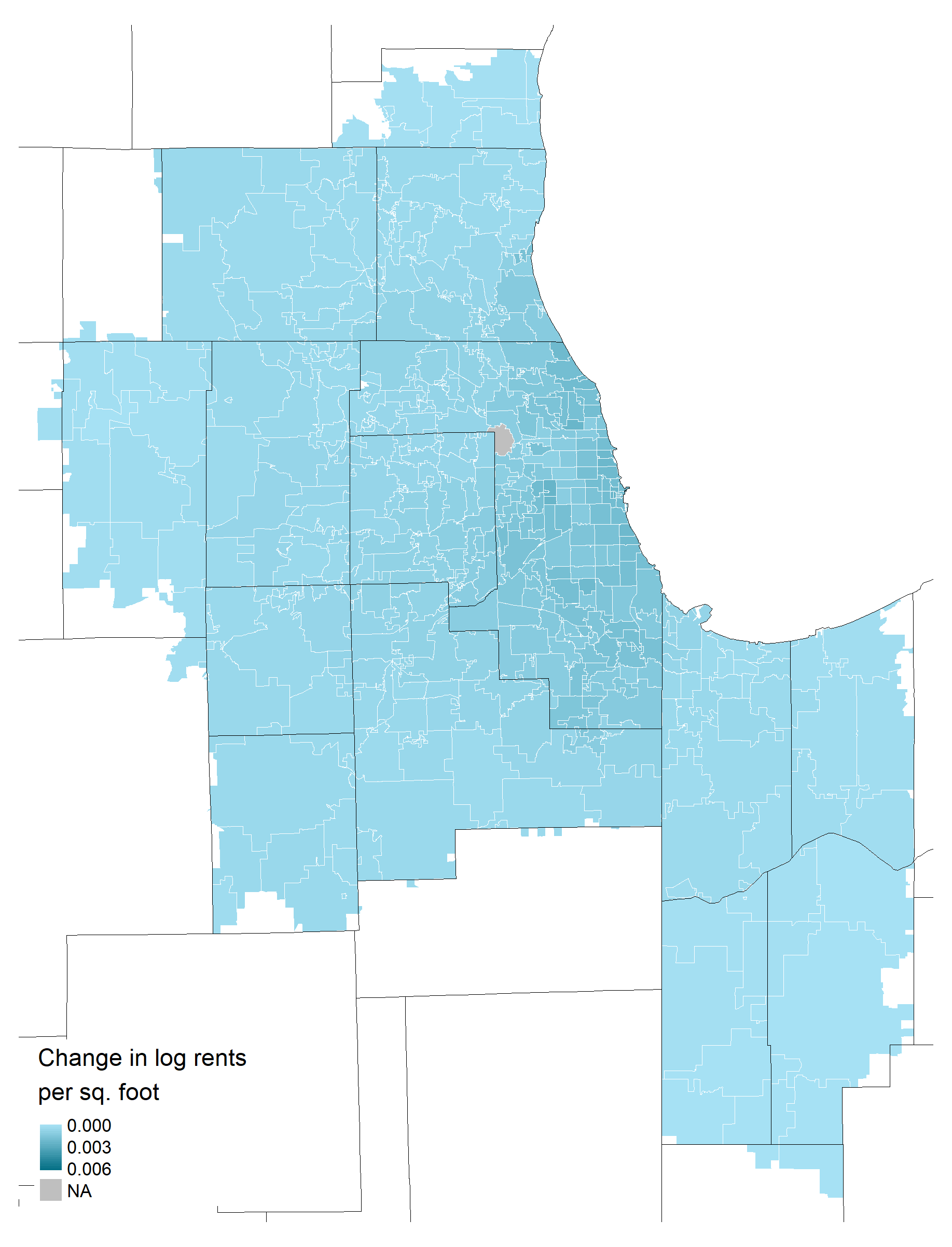}
    \end{subfigure}%
    $\quad\quad\quad\quad$%
    \begin{subfigure}{.4\textwidth}  \centering
        \caption*{Changes in log total wages}
        \includegraphics[width = 1\textwidth]
            {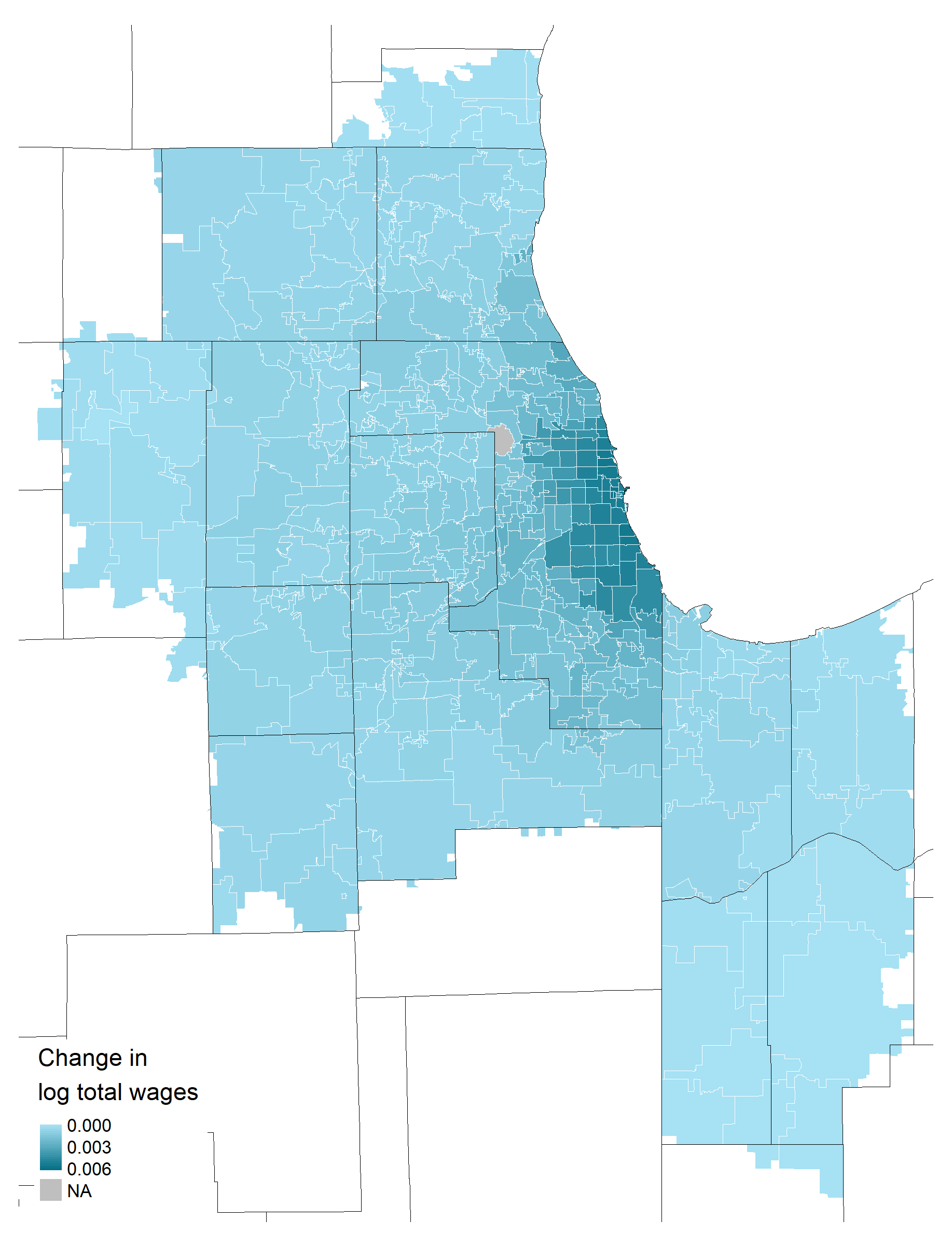}
    \end{subfigure}

    \begin{minipage}{.95\textwidth} \footnotesize
        \vspace{2.5mm}
        Notes: 
        Data are from the MW panel described in section \ref{sec:data_mw_panel} 
        and from LODES.
        The figures map the estimated changes in log total rents per square foot
        and log total wage income under different counterfactual MW policies 
        in the Chicago-Naperville-Elgin CBSA.
        Panel A is based on a counterfactual increase to \$9 in the 
        federal MW in January 2020, and Panel B on a counterfactual increase 
        from \$13 to \$14 in the Chicago City MW, both holding constant other 
        MW policies.
        The color scale has been standardized within each panel.
        To estimate the changes we follow the procedure described in Section 
        \ref{sec:counterfactual} assuming the following parameter values: 
        $\beta = \betaCf$, $\gamma = \gammaCf$, and $\varepsilon = 0.1$.
    \end{minipage}
\end{figure}

\clearpage
\begin{figure}[hbt!]
    \centering
    \caption{Estimated shares pocketed by landlords for different values of 
             the elasticity of wage income to the MW}
    \label{fig:cf_share_by_epsilon}

    \includegraphics[width = .75\textwidth]{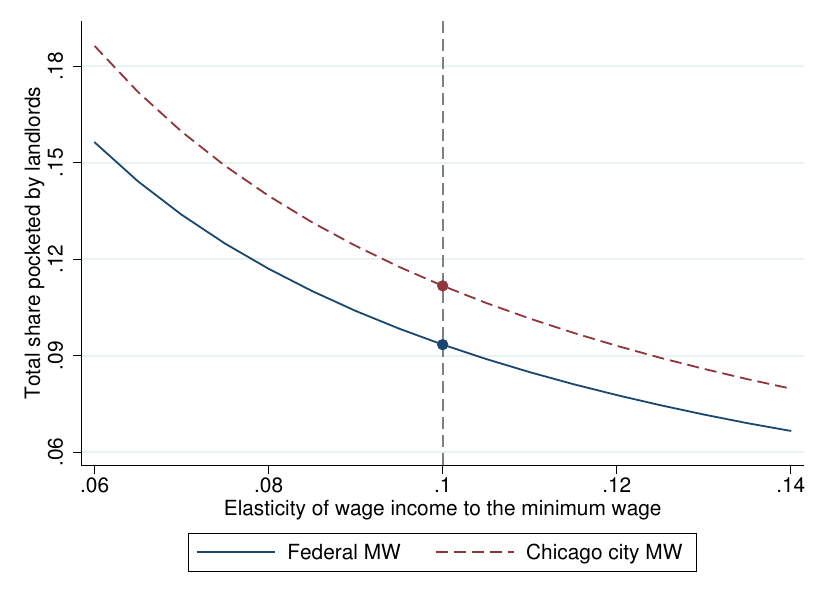}

    \begin{minipage}{.95\textwidth} \footnotesize
        \vspace{3mm}
        Notes:
        Data are from the MW panel described in section \ref{sec:data_mw_panel} 
        and from LODES.
        The figures show the estimated ZIP-code specific share of additional 
        income pocketed by landlords (``share pocketed'')
        under different counterfactual policies:
        an increase to \$9 in the federal MW in January 2020, and
        an increase from \$13 to \$14 in the Chicago City MW, 
        both holding constant other MW policies.
        The unit of observation is the ZIP code.
        To estimate it we follow the procedure described in Section 
        \ref{sec:counterfactual}, assuming the following parameter values: 
        $\beta = \betaCf$, $\gamma = \gammaCf$. 
        The x-axis shows a range of values for the elasticity of wage 
        income to the minimum wage $\varepsilon$.
        The line at $\epsilon=0.1$ corresponds to the estimates reported in
        Table \ref{tab:counterfactuals}.
    \end{minipage}
\end{figure}

\end{document}